\documentclass[sigconf]{aamas}

\usepackage{balance} 

\setcopyright{ifaamas}
\acmConference[AAMAS '26]{Proc.\@ of the 25th International Conference
on Autonomous Agents and Multiagent Systems (AAMAS 2026)}{May 25 -- 29, 2026}
{Paphos, Cyprus}{C.~Amato, L.~Dennis, V.~Mascardi, J.~Thangarajah (eds.)}
\copyrightyear{2026}
\acmYear{2026}
\acmDOI{}
\acmPrice{}
\acmISBN{}

\acmSubmissionID{320}

\title{Maximin Shares with Lower Quotas}

\author{Hirota Kinoshita}
\affiliation{
  \institution{Toyota Technological Institute at Chicago}
  \city{Chicago}
  \state{IL}
  \country{USA}}
\email{hirotak@ttic.edu}

\author{Ayumi Igarashi}
\affiliation{
  \institution{University of Tokyo}
  \city{Tokyo}
  \country{Japan}}
\email{igarashi@mist.i.u-tokyo.ac.jp}

\begin{abstract}
  We study the fair division of indivisible items among $n$ agents with heterogeneous additive valuations, subject to \textit{lower} and \textit{upper quotas} on the number of items allocated to each agent.
  Such constraints are crucial in various applications, ranging from personnel assignments to computing resource distribution.
  This paper focuses on the fairness criterion known as \textit{maximin shares (MMS)} and its approximations.
  Under arbitrary lower and upper quotas, we show that a $\left(\frac{2n}{3n-1}\right)$-MMS allocation of goods exists and can be computed in polynomial time, while we also present a polynomial-time algorithm for finding a $\left(\frac{3n-1}{2n}\right)$-MMS allocation of chores.
  Furthermore, we consider the generalized scenario where items are partitioned into multiple \textit{categories}, each with its own lower and upper quotas.
  In this setting, our algorithm computes an $\left(\frac{n}{2n-1}\right)$-MMS allocation of goods or a $\left(\frac{2n-1}{n}\right)$-MMS allocation of chores in polynomial time.
  These results extend previous work on the \textit{cardinality constraints}, i.e., the special case where only upper quotas are imposed.
\end{abstract}


\usepackage[table]{xcolor}
\usepackage{algorithm}
\usepackage{algpseudocodex}
\usepackage{mathtools}
\usepackage{enumitem}
\usepackage{etoolbox}
\usepackage[capitalize]{cleveref}
\usepackage[normalem]{ulem}
\usepackage{booktabs}
\usepackage{multirow}
\usepackage{cellspace}
\usepackage{adjustbox}
\usepackage{subcaption}
\usepackage{tikz}
\usepackage{listings}
\usepackage{todonotes}
\usepackage{thmtools}
\usepackage{thm-restate}

\AddToHook{cmd/appendix/before}{\crefalias{section}{appendix}}

\setcitestyle{numbers,comma,square,nosort}

\makeatletter
\newcommand{\figcaption}[1]{\def\@captype{figure}\caption{#1}}
\newcommand{\tblcaption}[1]{\def\@captype{table}\caption{#1}}
\makeatother

\theoremstyle{plain}
\newtheorem{theorem}{Theorem}
\newtheorem{lemma}{Lemma}
\newtheorem{corollary}{Corollary}

\theoremstyle{definition}
\newtheorem{definition}{Definition}

\theoremstyle{remark}
\newtheorem{remark}{Remark}

\newtheoremstyle{case}{}{}{}{}{\bfseries}{:}{ }{}
\theoremstyle{case}
\newtheorem{case}{Case}
\newcounter{caseoffset}
\newcounter{casedisp}
\newcommand{\resetcase}{
  \setcounter{caseoffset}{\value{case}}
}
\AtBeginEnvironment{proof}{
  \resetcase
}
\AtBeginEnvironment{case}{
  \setcounter{casedisp}{\numexpr\thecase-\value{caseoffset}+1}
  \renewcommand{\thecase}{\arabic{casedisp}}
}
\crefname{case}{Case}{Cases}
\crefrangeformat{case}{Cases~#3#1#4 to~#5#2#6}

\makeatletter
\providecommand\theHALG@line{}
\renewcommand\theHALG@line{\thealgorithm.\arabic{ALG@line}}
\makeatother

\makeatletter
\newcommand{\leqnomode}{\tagsleft@true}
\newcommand{\reqnomode}{\tagsleft@false}
\makeatother

\newcommand{\bR}{\mathbb{R}}

\newcommand{\bZ}{\mathbb{Z}}

\newcommand{\cC}{\mathcal{C}}
\newcommand{\cF}{\mathcal{F}}
\newcommand{\cI}{\mathcal{I}}

\newcommand{\cS}{\mathcal{S}}

\newcommand{\floordiv}[2]{\left\lfloor\frac{#1}{#2}\right\rfloor}

\let\oldComment\Comment
\renewcommand{\Comment}[1]{\oldComment{\textcolor{gray}{#1}}}

\begin{document}


\pagestyle{fancy}
\fancyhead{}


\maketitle


\section{Introduction} \label{sec:intro}

Fairly allocating indivisible items among agents with differing preferences is a fundamental problem in society, arising in various important scenarios ranging from commodity distribution to property settlement to computing resource management~\citep{brams1996fair,moulin2004fair,moulin2019fair}.
Long-standing research on the allocation of divisible resources, known as ``cake-cutting'', has introduced compelling fairness notions such as \emph{envy-freeness~(EF)} \citep{foley1967resource} and \emph{proportionality (PROP)} \citep{steinhaus1948problem}.
However, both EF and PROP are scarcely attainable with indivisible items.%
\footnote{Neither is feasible even in the simplest case where two agents must divide a single indivisible item.}
While various relaxations have been explored \cite{amanatidis2023fair}, the concept of \textit{maximin shares~(MMS)}~\citep{budish2011combinatorial} has garnered significant attention within the class of share-based notions.

Generalizing the idea of the \textit{cut-and-choose} protocol \citep{steinhaus1948problem}, the maximin share (MMS) of an agent is defined as the maximum value that he can guarantee for himself if he were to freely partition all items into the same number of bundles as there are agents, and then choose his least desirable bundle.
An allocation is considered fair if every agent receives at least their own MMS.
However, such an \textit{MMS allocation} may not exist~\citep{kurokawa2018fair,feige2021tight}, and finding an MMS allocation is strongly NP-hard~\citep{bouveret2016characterizing,aziz2017algorithms}.
Consequently, substantial research has sought to achieve (multiplicatively) \textit{approximate} MMS allocations \citep{kurokawa2018fair}, leading to remarkable progress in recent years~\citep{barman2020approximation,amanatidis2017approximation,garg2021improved,akrami2023simplification,akrami2024breaking,aziz2017algorithms,huang2021algorithmic,huang2023reduction,heidari2026improved,huang2025fptas}.
Continuing in similar vein, this paper adopts approximate MMS guarantees as our fairness criterion.

In addition to fairness desiderata, real-world applications often involve various restrictions on resource allocation \citep{suksompong2021constraints}.
One of the simplest constraints, which has gained little attention in the literature, is to impose \textit{lower quotas} on the number of items assigned to each agent.
For instance, consider the following scenario in a company:
the directors of upcoming projects are planning to allocate employees to their respective projects.
Each director prefers employees who are well-suited and capable of performing well for their project, while each project requires a minimum number of employees to be carried out.
Then, how can these directors fairly distribute the employees among their teams?
Beyond this specific example, lower quotas on the sizes of allocated bundles can arise in various applications of fair division, such as allocating courses to college students~\citep{budish2017course}, distributing computing resources to tasks~\citep{moulin2019fair}, assigning conference papers to reviewers~\citep{garg2010assigning}, and dividing food donations to individuals or communities in need~\citep{aleksandrov2015online,mertzanidis2024automating}.

\subsection{Our Contribution} \label{sec:intro:contrib}

Building on the preceding motivations, we initiate the study of MMS guarantees under lower quotas, where agents are assumed to have heterogeneous additive valuations.
Specifically, an allocation is feasible if the number of items assigned to each agent is between lower and upper quotas.
Under these constraints, we establish that a $\left(\frac{2n}{3n - 1}\right)$-MMS allocation exists and can be computed in polynomial time for any instance of goods\footnote{An item is called a \textit{good} (resp.~\textit{chore}) if it bears a non-negative (resp.~non-positive) value for each agent.} with $n$ agents.
For the case of chores, we present a polynomial-time algorithm to obtain a $\left(\frac{3n - 1}{2n}\right)$-MMS allocation.
These results both generalize and improve (notably for small $n$) on the best-known results for the special case without lower quotas \citep{hummel2022maximin}, while also implying the first MMS guarantees for another subcase studied as \textit{balanced} allocations \citep{suksompong2021constraints}.%
\footnote{An allocation is said to be \textit{balanced} if the number of items assigned to each agent differs by at most one.}

Furthermore, we study a generalized setting where items are partitioned into categories, each with its own lower and upper quotas on the number of items assigned from the category to any single agent.
In this setup, we show that an $\left(\frac{n}{2n - 1}\right)$-MMS $\left(\mbox{resp. $\left(\frac{2n - 1}{n}\right)$-MMS}\right)$ allocation exists and can be found in polynomial time for any instance of goods (resp.~chores) with $n$ agents.
These extend the previous results for the special case in which all categories have only upper quotas \citep{biswas2018fair,hummel2022maximin}.

In both the former (single-category) and latter (multi-category) settings, our approach builds on \citet{hummel2022maximin}'s work for the special case where only upper quotas are imposed.
While the main routine of our algorithms for the multi-category setting is a natural extension of its counterpart in \citep{hummel2022maximin}, the technical novelty of this paper lies primarily in the single-category setting, to which \citet{hummel2022maximin}'s solutions do not extend in a straightforward manner;
see \cref{sec:single:goods:novelty} for a further discussion.

In addition to these results, formally stated in \cref{sec:main_results}, \cref{sec:special} explores specific classes of instances that admit (almost) exact solutions, while \cref{sec:inapprox} investigates inapproximability bounds.

\begin{table}[tb]
  \caption{%
    Our polynomial-time approximate MMS guarantees and the best-known inapproximability bounds (implied by the unconstrained problem) for $n$ heterogeneous agents with additive valuations.
  }
  \label{table:contrib}

  \centering
  \newcommand\fracscale{1.0}

  \begin{subtable}{0.5\textwidth}
    \centering
    \caption{Goods (non-negative item values).}

    \begin{tabular}{Sl Sc Sc}
      \toprule
      \multirow{1}{*}{Category} & Lower bound & Upper bound
      \\ \midrule
      Single & \scalebox{\fracscale}{$\dfrac{2n}{3n - 1}$} [Thm.~\ref{thm:single:goods:two-third}] & \multirow[b]{2.2}{*}{\scalebox{\fracscale}{$
          \begin{cases}
            \dfrac{39}{40} & \!\!\!\!(n = 3) \vspace{3pt} \\
            1 - n^{-4} & \!\!\!\!(n \ge 4)
      \end{cases}$} \citep{feige2021tight}} \\
      \addlinespace[6pt]
      Multiple & \scalebox{\fracscale}{$\dfrac{n}{2n - 1}$} [Thm.~\ref{thm:multi:goods:half}] & \\ \bottomrule \\
    \end{tabular}
  \end{subtable}

  \begin{subtable}{0.5\textwidth}
    \centering
    \caption{Chores (non-positive item values).}

    \begin{tabular}{Sl Sc Sc}
      \toprule
      \multirow{1}{*}{Category} & Upper bound & Lower bound
      \\ \midrule
      Single & \scalebox{\fracscale}{$\dfrac{2n}{3n - 1}$} [Thm.~\ref{thm:single:goods:two-third}] & \multirow[b]{2.2}{*}{\scalebox{\fracscale}{$\dfrac{44}{43}$ $(n = 3)$} \citep{feige2021tight}} \\
      \addlinespace[6pt]
      Multiple & \scalebox{\fracscale}{$\dfrac{2n - 1}{n}$} [Thm.~\ref{thm:multi:chores:two}] & \\ \bottomrule
    \end{tabular}
  \end{subtable}

\end{table}




\subsection{Related Work} \label{sec:related}

\paragraph{Related Constraints in Fair Division}
The earliest related studies include \citet{ferraioli2014regular}'s work on the constraint requiring every agent to receive the same number of items.
\citet{mackin2016allocating} consider the setting where items are partitioned into categories, and each agent must receive at least one item per category.
A subsequent series of work has studied the so-called \textit{cardinality constraints}, where items are categorized, and allocations are restricted by category-wise upper quotas.
\citet{biswas2018fair} reveal that, under cardinality constraints, both an EF1 allocation and a $1/3$-MMS allocation exist and can be found in polynomial time.
\citet{hummel2022maximin} improve the MMS approximation to $1/2$ and further establish $2/3$-MMS allocations for single-category instances, both in polynomial time.
They also study the case of chores, achieving a $2$-approximation for general instances and a $3/2$-approximation for single-category instances in polynomial time.
In the case of two agents, \citet{shoshan2023efficient} develop a polynomial-time algorithm that computes a feasible allocation satisfying \textit{Pareto optimality (PO)} and EF1 for either goods or chores, which has been generalized by \citet{igarashi2025fair}.
\citet{cookson2025constrained} consider the same model as ours, i.e., an extension of cardinality constraints endowed with lower quotas, where they prove that the \textit{maximum Nash welfare (MNW)} solutions achieve PO and approximate EF1.
Aside from the lower quotas we focus on, other generalizations of cardinality constraints include the \emph{budget constraints} \citep{wu2025approximate,barman2023finding,barman2023guaranteeing,garbea2023efx,dai2023maximum,deng2024budgeted,elkind2024fair,wang2025guaranteeing}, where items of varying sizes are assigned to agents subject to their capacities, and the \emph{matroid constraints} \citep{biswas2018fair,biswas2019matroid,dror2023fair,wang2024fairness,cookson2025constrained,akrami2026matroids}, which require assigned bundles to be independent sets.
These constraints are further subsumed under \emph{independence systems}, in which MMS approximation has been studied \citep{li2021fair-division,hummel2025maximin}.

\paragraph{Unconstrained MMS Approximation and Inapproximability}
Given the non-existence and computational intractability of exact MMS allocations even in the unconstrained setting~\citep{procaccia2014fair,aziz2017algorithms}, extensive work has studied approximate MMS guarantees in the unconstrained setting.
The best-known approximation ratio is $\frac{7}{9}$ for the case of goods~\citep{huang2025fptas} and $\frac{13}{11}$ for chores~\citep{huang2023reduction}.
On the negative side, \citet{feige2021tight} establish inapproximability bounds of $\frac{39}{40}$ for goods and $\frac{44}{43}$ for chores, which remain the best-known results even in our constrained setup; see \cref{sec:inapprox} for a computational approach to improving the bounds.

\paragraph{Lower Quotas in Two-Sided Matching}
In the many-to-one matching framework, commonly referred to as the \textit{hospitals/residents problem} or \textit{college admissions problem} \citep{gale1962college,gusfield1989stable}, one side often imposes lower quotas on the number of acceptable matches from the other side \citep{arulselvan2018matchings,aziz2022matching}.
The primary objective in this model is typically to ensure \emph{stability} or its suitable relaxations \citep{biro2010college,huang2010classified,hamada2016hospitals,kamada2017stability,nasre2017popular,yokoi2020envy}.

\section{Preliminaries} \label{sec:prelim}



Our problem instance is a tuple $\cI = \left(N,M,(v_i)_{i\in N},\cC,(q_{C}^-,q_{C}^+)_{C\in\cC}\right)$, where $N = \{1,2,\ldots,|N|\}$ is a finite set of \textit{agents}, $M$ is a finite set of \textit{items}, $(v_i:2^M\to\bR)_{i\in N}$ are the agents' respective \textit{valuations},
and $\cC$ is a partition of the item set $M$ into disjoint \textit{categories}, each endowed with a \textit{lower quota} $q^-_C\in \bZ_{\ge 0}$ and an \textit{upper quota} $q^+_C\in \bZ_{\ge 0}$, such that
\begin{align}
  &q_C^-|N| \le |C| \le q_C^+|N|. \label{eq:def:quotas}
\end{align}
Each valuation $v_i:2^M\to\bR$ is additive, i.e., $v_i(S) = \sum_{g \in S} v_i(g)$ for any $S\subseteq M$, where $g\in M$ is understood as $\{g\}\subseteq M$ and treated similarly throughout the paper.
A subset of items, $S\subseteq M$, is called a \emph{bundle} and said to be \emph{feasible} if $q^-_C \le |S\cap C| \le q^+_C$ for any category $C\in\cC$.
An \textit{allocation} for $\cI$ is an ordered partition $A = (A_i)_{i\in N}$ of $M$, and said to be \textit{feasible} if each bundle $A_i \subseteq M$ is feasible.
$\cF(\cI)$ denotes the set of all feasible allocations for $\cI$, which is non-empty due to \cref{eq:def:quotas}.

An instance $\cI$ is said to be \textit{of goods} (resp.~\textit{of chores}) if $v_i(g) \ge 0$ (resp.~$\le 0$) for every pair $(i, g)\in N\times M$, where each item is called a \textit{good} (resp.~\textit{chore}).
An instance $\cI$ is said to be \textit{ordered} \citep{hummel2022maximin} when each category $C$ has an ordering of its items $C = \left\{g^C_1,g^C_2,\ldots,g^C_{|C|}\right\}$ such that
\begin{align}
  &v_i\!\left(g^C_1\right) \ge v_i\!\left(g^C_2\right) \ge \cdots \ge v_i\!\left(g^C_{|C|}\right) &\forall i\in N.
\end{align}
Here, item $g^C_j\in C$ is said to be \textit{more valuable} than item $g^C_{j'}\in C$ if $j < j'$.
An instance $\cI$ is said to be \textit{single-category} if $\cC = \{M\}$, where we let $(q^-,q^+) \coloneqq \left(q^-_M, q^+_M\right)$ and equate $\cI$ with the tuple $\left(N, M, (v_i)_{i\in N}, (q^-, q^+)\right)$; when $\cI$ is also ordered, we let $g_j \coloneqq g^M_j$ for each $j\in\{1,2,\ldots,|M|\}$.

Note that cardinality constraints are the special cases when $q^-_C = 0$ for every $C\in\cC$, and that the unconstrained setting is represented by single-category instances with $(q^-,q^+) = (0,|M|)$.

Let us define the fairness criteria based on \textit{maximin shares (MMS)}.
We aim to design an algorithm that, given an arbitrary instance, computes an $\alpha$-MMS allocation for some $\alpha$ as close to $1$ as possible.

\begin{definition} \label{def:MMS}
  For an instance $\cI = \left(N,M,(v_i)_{i\in N},\cC,(q_{C}^-,q_{C}^+)_{C\in\cC}\right)$ and an agent $i\in N$, we define agent $i$'s \textit{maximin share (MMS)} as
  \begin{align*}
    \mu_i(\cI) \coloneqq \max \left\{\min_{k\in N}\, v_i(P_k) \mid (P_k)_{k\in N}\in\cF(\cI) 
    \right\},
  \end{align*}
  where any $(P_k)_{k\in N}$ achieving the maximum is called agent $i$'s \textit{MMS partition}.
  A feasible allocation $(A_i)_{i\in N}$ is said to be an \textit{$\alpha$-MMS allocation} for $\cI$ and some $\alpha\in\bR$ if
  $v_i(A_i) \ge \alpha\,\mu_i(\cI)$ for all $i\in N$.
  Particularly when these inequalities hold for $\alpha = 1$, $(A_i)_{i\in N}$ is said to be an \textit{MMS allocation}.
\end{definition}


It is shown that an arbitrary instance can be reduced to an ordered instance, which is well-known in the unconstrained setting \citep{bouveret2016characterizing} and also under cardinality constraints \citep{hummel2022maximin}.
Because the known reduction preserves an instance except the valuation profile and does not alter the cardinality of solutions, we can focus solely on ordered instances as input even in the presence of lower quotas.

\begin{restatable}[{Corollary of \citep[Theorem 3]{hummel2022maximin}}]{lemma}{OrderedInstance} \label{thm:common:ordered}
  For an arbitrary instance $\cI = \left(N,M,(v_i)_{i\in N},\cC,(q_{C}^-,q_{C}^+)_{C\in\cC}\right)$ and any $\alpha\in\bR$, one can compute an ordered instance $\tilde\cI$ in time $O(|N||M|\log |M|)$, such that an $\alpha$-MMS allocation for $\cI$ can be computed from any $\alpha$-MMS allocation for $\tilde\cI$ in time $O(|N||M|\log |M|)$.
\end{restatable}

\begin{algorithm}[tb]
  \centering

  \begin{algorithmic}[1]
    \Require An instance $\cI = \left(N,M,({v}_i)_{i\in N},\cC,(q_{C}^-,q_{C}^+)_{C\in\cC}\right)$. An allocation $(\tilde A_i)_{i\in N}$ for the ordered instance $\tilde\cI = \left(N,M,(\tilde{v}_i)_{i\in N},\cC,(q_{C}^-,q_{C}^+)_{C\in\cC}\right)$ defined from $\cI$ in the proof of \cref{thm:common:ordered}
    \Ensure An allocation $(A_i)_{i\in N}$ for the instance $\cI$.
    \State $A_i \gets \emptyset$\quad for each $i\in N$.
    \For{\textbf{each} $C\in\cC$}
    \For{$j = 1,2,\ldots,|C|$}
    \State Find $i^*\in N$ s.t. $g^C_j \in \tilde A_{i^*}$.
    \State Find $g^*\in\arg\max \left\{v_{i^*}(g) \mid g\in C\setminus \bigcup_{i\in N} A_i\right\}$. \label{alg:common:ordered:find-item}
    \State $A_{i^*}\gets A_{i^*}\cup\{g^*\}$. \label{alg:common:ordered:assign}
    \EndFor
    \EndFor
  \end{algorithmic}

  \caption{Obtain a solution of the original instance from that of its corresponding ordered instance \citep[Algorithm 2]{hummel2022maximin}.}
  \label{alg:common:ordered}

\end{algorithm}

\begin{proof}
  Let $\cI = \left(N,M,({v}_i)_{i\in N},\cC,(q_{C}^-,q_{C}^+)_{C\in\cC}\right)$ be an arbitrary instance.
  Fix any ordering in each category $C = \left\{g^C_1,g^C_2,\ldots,g^C_{|C|}\right\}$.
  Following \citep[Algorithm 1]{hummel2022maximin}, we define an ordered instance $\tilde\cI \coloneqq \left(N,M,(\tilde{v}_i)_{i\in N},\cC,(q_{C}^-,q_{C}^+)_{C\in\cC}\right)$ by
  \begin{align*}
    \tilde{v}_i\!\left(g^C_j\right) &\coloneqq \left(\mbox{the $j$-th largest value in $\left\{v_i(g)\mid g\in C\right\}$}\right) \\
    &\forall j\in\{1,2,\ldots,|C|\},\forall C\in\cC,\forall i\in N, 
  \end{align*}
  where ties are broken arbitrarily.
  Clearly, $\mu_i(\tilde\cI) = \mu_i(\cI)$ holds for every $i\in N$.
  Furthermore, given an allocation $\tilde A = (\tilde A_i)_{i\in N}$ for $\tilde\cI$, Line~\ref{alg:common:ordered} computes an allocation $A = (A_i)_{i\in N}$ for $\cI$ such that
  \begin{align*}
    |A_i\cap C| &= |\tilde A_i \cap C| &\forall C\in \cC,\forall i\in N, \\
    v_i(A_i) &\ge \tilde{v}_i(\tilde A_i) &\forall i\in N.
  \end{align*}
  Therefore, if $\tilde A$ is an $\alpha$-MMS allocation for $\tilde\cI$, then $A$ is also an $\alpha$-MMS allocation for $\cI$.
  The ordered instance $\tilde I$ can be obtained in time $O(|N||M|\log |M|)$ by sorting the item values $(v_i(g))_{g\in M}$ for each agent $i\in N$.
  In addition, \cref{alg:common:ordered} can be implemented to run in time $O(|N||M| \log |M|)$ by, for each $i\in N$, similarly computing the sorted array of item values in advance and using a one-directional pointer over it to avoid scanning the same item value more than once at Line~\ref{alg:common:ordered:find-item}.
\end{proof}

We also extend the concept of \textit{valid reduction}, which has been leveraged largely in the unconstrained problem \citep{bouveret2016characterizing,kurokawa2018fair,amanatidis2017approximation,ghodsi2021fair,garg2019approximating,garg2021improved,akrami2023simplification,akrami2024breaking,heidari2026improved,huang2025fptas} but also under constraints \citep{gourves2019maximin,li2021fair-division,hummel2022maximin,deng2024budgeted}.
\cref{thm:common:valid-reduction} allows us to assign the feasible bundle $B$ to any agent without decreasing other agents' MMS in the new instance $\cI'$.
This bundle $B$ is designed not only to ensure that $\cI'$ is well-defined, but also so that any agent's MMS partition contains an \emph{item-wise more valuable} bundle than $B$ by the pigeonhole principle, as formalized in the proof of \cref{thm:common:valid-reduction}.
\cref{thm:common:cor-valid-reduction} then suggests that, as long as some agent $i^*$ enjoys $v_{i^*}(B) \ge \alpha\,\mu_{i^*}(\cI)$, finding an $\alpha$-MMS allocation for the original instance $\cI$ can be reduced to finding one for $\cI'$.

\begin{restatable}{lemma}{LemValidReduction}
  \label{thm:common:valid-reduction}
  Let $\cI = \left(N,M,(v_i)_{i\in N},\cC,(q_{C}^-,q_{C}^+)_{C\in\cC}\right)$ be an ordered instance of goods, $i^*\in N$ be arbitrary, $C^*\in\cC$ be non-empty, and $d\in\bZ_{\ge 0}$ satisfy $|C^*| \ge d|N| + 1$.
  Let us also define the following:
  \begin{align*}
    \begin{split}
      B &\coloneqq \left\{g^{C^*}_{j} \mid d(|N| - 1) < j \le d|N| + 1 \right\} \\
      &\quad\cup \left\{g^{C^*}_{|C^*| - j} \mid 0\le j < \max\left\{q_{C^*}^-,|C^*|-q_{C^*}^+(|N|-1)\right\} - d - 1 \right\} \\
      &\quad\cup \bigcup_{C\in\cC\setminus\{C^*\}}\left\{g^{C}_{|C| - j} \mid 0\le j < \max\left\{q_{C}^-,|C|-q_{C}^+(|N|-1)\right\}\right\},
    \end{split} \\ 
    \begin{split}
      \cI' &\coloneqq \left(N\setminus\{i^*\}, M\setminus B, \left(v_{i}|_{2^{M\setminus B}}\right)_{i\in N\setminus \{i^*\}}, \right. \\
      & \hspace{64pt} \left. \{C\setminus B\mid C\in\cC\}, \left(q_{C}^-,q_{C}^+\right)_{C\in \cC}\right).
    \end{split} 
  \end{align*}
  Then $\cI'$ is well-defined as an ordered instance of goods and satisfies that
  \begin{align*}
    \mu_{i}(\cI') &\ge \mu_{i}(\cI) &\forall i\in N\setminus\{i^*\}. 
  \end{align*}
\end{restatable}

\begin{proof}
  By the definition of $B$, each category $C \in \cC$ satisfies that
  \begin{align}
    |B \cap C| &=
    \begin{cases}
      \max\left\{d + 1, q^-_C, |C| - q_C^+(|N| - 1)\right\} & \text{ if }\ C = C^*; \\
      \max\left\{q^-_C, |C| - q_C^+(|N| - 1)\right\} & \text{ otherwise.}
    \end{cases}
    \label{eq:common:goods:valid-reduction:num-per-category}
  \end{align}
  \cref{eq:def:quotas} and that $|C^*| \ge d|N| + 1$ together imply
  \begin{align}
    d + 1 & \le
    \begin{cases}
      |C^*| - \floordiv{|C^*|}{|N|}(|N| - 1)  & \text{ if }\ d = \floordiv{|C^*|}{|N|}, \\
      \floordiv{|C^*|}{|N|} & \text{ otherwise}
    \end{cases} \\
    & \le |C^*| - q^-_{C^*}(|N| - 1). \label{eq:common:goods:valid-bound-(d+1)}
  \end{align}
  It follows from \cref{eq:def:quotas,eq:common:goods:valid-bound-(d+1),eq:common:goods:valid-reduction:num-per-category} that
  \begin{align}
    |C| - q_C^+(|N| - 1) &\le |B \cap C| \le |C| - q_C^-(|N| - 1) &\forall C\in\cC,
    \label{eq:common:goods:valid-remain-category}
  \end{align}
  which allows the instance $\cI'$ to be well-defined as an ordered instance of goods.

  Let $i\in N\setminus\{i^*\}$ be arbitrary.
  Let $(P_k)_{k=1}^{|N|}$ be agent $i$'s MMS partition, which can be assumed, without loss of generality, to satisfy that
  \begin{align}
    \left|P_1 \cap \left\{g^{C^*}_j\mid 1\le j\le d|N| + 1\right\}\right| \ge d + 1 \label{eq:common:goods:valid-reduce-pigeon}
  \end{align}
  by the pigeonhole principle and that $|C^*| \ge d|N| + 1$.
  Given \cref{eq:common:goods:valid-reduce-pigeon} and the definition of $B$, as well as that $(P_k)_{k=1}^{|N|}\in\cF(\cI)$, there is an injection $f\colon B\setminus P_1 \to P_1\setminus B$ that satisfies
  \begin{align}
    f(g) &\in C \ \text{ and } \ v_{i}(g) \le v_{i}(f(g)) &\forall g\in C\cap(B\setminus P_1),\forall C\in\cC, \label{eq:common:goods:injection}
  \end{align}
  Now, we define
  \begin{align}
    \begin{split}
      \nu(M') \coloneqq \mu_i\left(N\setminus\{i^*\}, M', \left(v_{i'}|_{2^{M'}}\right)_{i'\in N\setminus \{i^*\}}, \right. \\
      \left. \{C\cap M'\mid C\in\cC\}, \left(q_{C}^-,q_{C}^+\right)_{C\in\cC}\right)
    \end{split} \label{eq:common:goods:mu-reduced}
  \end{align}
  for any $M'\subseteq M$ that ensures the instance in the right-hand side is well-defined.
  Then from \cref{eq:common:goods:valid-remain-category,eq:common:goods:injection,eq:common:goods:mu-reduced}, as well as the definition of $(P_k)_{k=1}^{|N|}$, we finally obtain
  \begin{align*}
    \mu_{i}(\cI') &= \nu(M\setminus B) \\
    &\ge \nu(M\setminus ((B \cap P_1) \cup f(B\setminus P_1))) \\
    &\ge \nu(M\setminus P_1) \\
    &\ge \min \left\{v_{i}(P_{k}) \mid k \in \{2,3,\ldots,|N|\}\right\} \\
    &\ge \mu_{i}(\cI),
  \end{align*}
  concluding the proof of \cref{thm:common:valid-reduction}.
\end{proof}

\begin{corollary} \label{thm:common:cor-valid-reduction}
  Let $\alpha\in\bR$ be arbitrary.
  In \cref{thm:common:valid-reduction}, if we additionally have that $v_{i^*}(B) \ge \alpha\,\mu_{i^*}(\cI)$,
  and an $\alpha$-MMS allocation $\left(A'_i\right)_{i\in N\setminus\{i^*\}}$ for $\cI'$ is given, the following $(A_i)_{i\in N}$ is an $\alpha$-MMS allocation for $\cI$:
  \begin{align*}
    A_i &\coloneqq
    \begin{cases}
      B & \text{ if }\ i = i^*, \\
      A'_i & \text{ otherwise}
    \end{cases}
    &\forall i\in N.
  \end{align*}
\end{corollary}

\begin{proof}
  This is immediate from \cref{def:MMS,thm:common:valid-reduction}.
\end{proof}



\section{Main results} \label{sec:main_results}

We present a polynomial-time algorithm to compute a $\left(\frac{2n}{3n - 1}\right)$-MMS allocation for a single-category instance of goods with $n$ agents, thereby establishing \cref{thm:single:goods:two-third}.
This result extends and strictly improves upon \citet{hummel2022maximin}'s $2/3$-approximation for the special case without lower quotas.
\cref{sec:single:goods} is devoted to examining an existing approach and its issues, defining the proposed algorithm, discussing its key components, and establishing the desired guarantees through a careful scrutiny of the algorithm.

\begin{theorem} \label{thm:single:goods:two-third}
  For an arbitrary single-category instance of goods $\cI = (N, M, (v_i)_{i\in N}, (q^-,q^+))$, a $\left(\frac{2}{3 - (\max\left\{|N|, 1\right\})^{-1}}\right)$-MMS allocation for $\cI$ can be computed in time $O(|N||M|\log |M|)$.%
  \footnote{We assume the uniform cost model and access to constant-time valuation oracles.}
\end{theorem}


A similar approach leads to the following result for the case of chores, which we show in \cref{sec:single:chores}.
As in the case of goods, \cref{thm:single:chores:three-halves} extends and strictly improves upon \citet{hummel2022maximin}'s $3/2$-approximation for their special case.

\begin{restatable}{theorem}{SingleChoresMain} \label{thm:single:chores:three-halves}
  For an arbitrary single-category instance of chores $\cI = (N, M, (v_i)_{i\in N}, (q^-,q^+))$, a $\left(\frac{3 - (\max\left\{|N|, 1\right\})^{-1}}{2}\right)$-MMS allocation for $\cI$ can be computed in time $O(|N||M|\log |M|)$.
\end{restatable}


For the more general setting with categorized items, we derive \cref{thm:multi:goods:half,thm:multi:chores:two} for goods and chores, respectively, in \cref{sec:multi} by leveraging \citet{hummel2022maximin}'s results for the special case without lower quotas.

\begin{restatable}{theorem}{MultiGoodsMain} \label{thm:multi:goods:half}
  For an arbitrary (multi-category) instance of goods $\cI = \left(N,M,(v_i)_{i\in N},\cC,(q_{C}^-,q_{C}^+)_{C\in\cC}\right)$, a $\left(\frac{1}{2 - \left(\max\left\{|N|,1\right\}\right)^{-1}}\right)$-MMS allocation for $\cI$ can be computed in time $O(|N||M|\log |M|)$.
\end{restatable}

\begin{restatable}{theorem}{MultiChoresMain} \label{thm:multi:chores:two}
  For an arbitrary (multi-category) instance of chores $\cI = \left(N,M,(v_i)_{i\in N},\cC,(q_{C}^-,q_{C}^+)_{C\in\cC}\right)$, a $\left(2 - \frac{1}{\max\left\{|N|,1\right\}}\right)$-MMS allocation for $\cI$ can be computed in time $O(|N||M|\log |M|)$.
\end{restatable}

\section{A $\left(\frac{2n}{3n-1}\right)$-MMS allocation algorithm for single-category goods} \label{sec:single:goods}

\subsection{An Existing Approach for Upper Quotas} \label{sec:single:goods:novelty}

Before discussing the proposed algorithm, we revisit an existing approach for approximate MMS guarantees only with upper quotas \cite{hummel2022maximin}.
Whether constrained or unconstrained, one of the most common techniques is a simple greedy algorithm known as \emph{bag-filling} \citep{ghodsi2021fair,garg2019approximating,garg2021improved,akrami2023simplification,akrami2024breaking,li2021fair-scheduling,li2023fair,deng2024budgeted,heidari2026improved}.
In each round of this algorithm, a new bag is created as an empty set, and remaining items are added one by one until an unassigned agent finds it sufficiently valuable; then the round ends with the bag assigned to this agent.
By adopting an additional type of update, \citet{hummel2022maximin} develops an algorithm that computes a $2/3$-MMS allocation under an upper quota.
In each round of their algorithm, a new bag is created with the least valuable remaining items, many enough to ensure that the other remaining items can be divided into bundles under the upper quota.
Then the other items are added one by one, as in the standard bag-filling, but not to violate the upper quota.
If the bag reaches the quota and is still not valuable enough for any unassigned agent, the algorithm repeatedly exchanges the least valuable item in the bag with a higher-valued remaining item.
This way of updates allows both the remaining items added to the bag and those not added to be sufficiently valuable and not too many (relative to the upper quota).
However, additionally imposing a lower quota would also need both of them not to become too few, thereby raising a non-trivial challenge.
Indeed, the correctness of their algorithm crucially relies on invariants concerning the values of certain subsets of remaining items, which would no longer hold consistently with lower quotas.

\subsection{Overview of Our Algorithm}

\begin{algorithm}[tb]
  \centering

  \begin{algorithmic}[1]
    \Function{ApproxGoods}{$\cI = \left(N,M,(v_i)_{i\in N},(q^-,q^+)\right),\alpha$} \label{alg:single:goods:main-start}

    \If{$|M| \le |N|$} \label{alg:single:goods:check-few-items}
    \Comment{A trivial case.}
    \State $A_i\gets
    \begin{cases}
      \{g_{i}\} & \text{if } i \le |M|, \\
      \emptyset & \text{otherwise}
    \end{cases}$ for each $i \in N = \{1,\ldots,|N|\}$. \label{alg:single:goods:bundle-few-items}
    \State \Return $(A_i)_{i\in N}$ \label{alg:single:goods:allocation-few-items}
    \EndIf \label{alg:single:goods:end-check-few-items}


    \For{$k \gets |N|,|N|-1,\ldots,1$} \label{alg:single:goods:B_k-init-for-start} \\
    \Comment{Initialize the bags; see \cref{fig:single:goods:init}.}
    \State $b_k \gets \min\left\{q^+, |M| - \sum_{k'=k+1}^{|N|}b_{k'} - (k - 1)\max\left\{q^-,1\right\}\right\}$. \label{alg:single:goods:b_k-init-n}
    \State $\begin{aligned}
      B^{(|N|)}_k \gets & \left\{g_k\right\} \cup \\
      & \left\{g_{|N| + j} \mid \sum_{k'=k+1}^{|N|}(b_{k'} - 1) < j \le \sum_{k'=k}^{|N|}(b_{k'} - 1) \right\}.
    \end{aligned}$ \label{alg:single:goods:B_k-init-n}
    \EndFor \label{alg:single:goods:B_k-init-for-end}

    \State $\hat\mu_i \gets \min \left\{ \frac{1}{|N|-r+1}\, v_i\left(\bigcup_{k = r}^{|N|} B^{(|N|)}_k\right) \mid 1\le r\le |N| \right\}$ \\ for each $i\in N$. \label{alg:single:goods:mu-hat}

    \If{$\exists i^*\in N$ s.t. $v_{i^*}\!\left(\left\{g_{|N|}, g_{|N|+1}\right\}\right) \ge \alpha\,\hat\mu_{i^*}$} \label{alg:single:goods:check-two-from-n-th} \\
    \Comment{Valid reduction.}
    \State $\begin{aligned}
      A_{i^*}\gets & \left\{g_{|N|}, g_{|N|+1}\right\} \cup \left\{g_{|M|-j} \mid \right. \\
      & \quad \left. 0\le j< \max\left\{q^-, |M| - q^+ (|N| - 1)\right\} - 2 \right\}.
    \end{aligned}$ \label{alg:single:goods:valid-reduce-assign}
    \State $\cI' \gets \left(N\setminus\left\{i^*\right\},M\setminus A_{i^*}, (v_{i}|_{2^{M\setminus A_{i^*}}})_{i\in N\setminus\{i^*\}},\left(q^-, q^+\right)\right)$. \hspace{-5pt} \label{alg:single:goods:valid-reduced-instance}
    \State $(A_{i})_{i\in N\setminus\{i^*\}}\gets\Call{\hyperref[alg:single:goods:main-start]{ApproxGoods}}{\cI',\alpha}$. \label{alg:single:goods:valid-reduce}
    \State \Return $(A_i)_{i\in N}$ \label{alg:single:goods:allocation-valid-reduction}
    \EndIf \label{alg:single:goods:valid-reduction-end}


    \State $N^{(|N|)}\gets N$. \label{alg:single:goods:var-init}

    \For{$t \gets |N|,|N|-1,\ldots,1$} \label{alg:single:goods:outer-for-start}
    \State $(B^{(t-1)}_1, B^{(t-1)}_2, \ldots, B^{(t-1)}_{t-1}, B) \gets (B^{(t)}_1, B^{(t)}_2, \ldots, B^{(t)}_{t})$. \label{alg:single:goods:bag-init}

    \For{$k \gets t - 1,t - 2,\ldots,1$} \label{alg:single:goods:inner-for-start}

    \While{$\left(\exists i\in N^{(t)}~\mbox{s.t.}~v_i(B) \ge \frac{3}{2}\,\alpha\,\hat\mu_i \right)$ and $B\setminus\{g_t\} \neq B^{(t)}_k\setminus\{g_k\}$} \label{alg:single:goods:inner-while-start}
    \Comment{The latter
    implies that some items in $B\setminus\{g_t\}$ have not moved in the current while-loop.}
    \State $g\gets$ the least valuable item in $B\setminus\{g_t\}$. \label{alg:single:goods:get-from-B}
    \If{$|B| > |B^{(t)}_k|$}
    \State $B\gets B\setminus\{g\}$, $B^{(t-1)}_k\gets B^{(t-1)}_k\cup\{g\}$. \label{alg:single:goods:move} \\
    \Comment{Move the item $g$; see \cref{fig:single:goods:move}.}
    \Else
    \Comment{$B$ and $B^{(t)}_k$ have become of equal size.}
    \State $h\gets$ the most valuable item in $B^{(t-1)}_k\setminus\{g_k\}$. \hspace{-5pt} \label{alg:single:goods:get-from-B_k}
    \State $B\gets (B\setminus\{g\})\cup\{h\}$, \\ $B^{(t-1)}_k\gets (B^{(t-1)}_k\setminus\{h\})\cup\{g\}$. \label{alg:single:goods:swap} \\
    \Comment{Swap the items $g$ and $h$; see \cref{fig:single:goods:swap}.}
    \EndIf
    \EndWhile \label{alg:single:goods:inner-while-end}
    \EndFor \label{alg:single:goods:inner-for-end}

    \State Find $i^{(t)}\in N^{(t)}$ s.t. $v_{i^{(t)}}(B) \ge \alpha\,\hat\mu_{i^{(t)}}$.\label{alg:single:goods:before-assign}
    \State $A_{i^{(t)}}\gets B$, $N^{(t - 1)}\gets N^{(t)}\setminus \left\{i^{(t)}\right\}$. \label{alg:single:goods:assign}

    \EndFor \label{alg:single:goods:outer-for-end}

    \State \Return $(A_i)_{i\in N}$. \label{alg:single:goods:allocation-final}
    \EndFunction
  \end{algorithmic}

  \caption{Compute a $\left(\frac{2}{3-(\max\{1,n\})^{-1}}\right)$-MMS allocation for a single-category ordered instance of goods with $n$ agents.}
  \label{alg:single:goods}
\end{algorithm}

\begin{figure}[tb]
  \centering

  \begin{subfigure}{0.45\textwidth}
    \centering
    \includegraphics[scale=0.3]{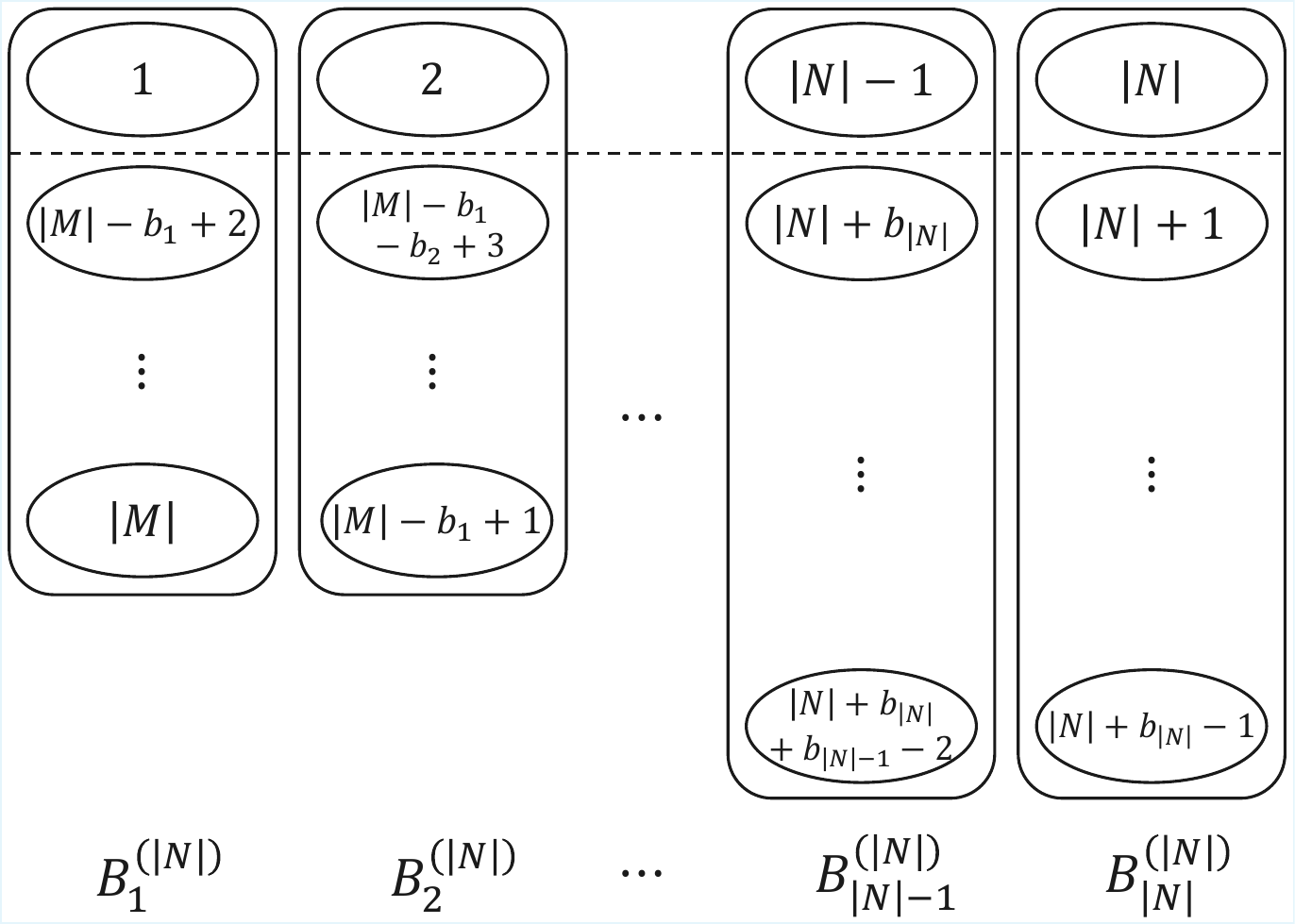}
    \caption{Initialization (Lines~\ref{alg:single:goods:B_k-init-for-start} to \ref{alg:single:goods:B_k-init-for-end}).}
    \label{fig:single:goods:init}
  \end{subfigure} \\ \vspace{10pt}
  \begin{subfigure}{0.22\textwidth}
    \centering
    \includegraphics[scale=0.3]{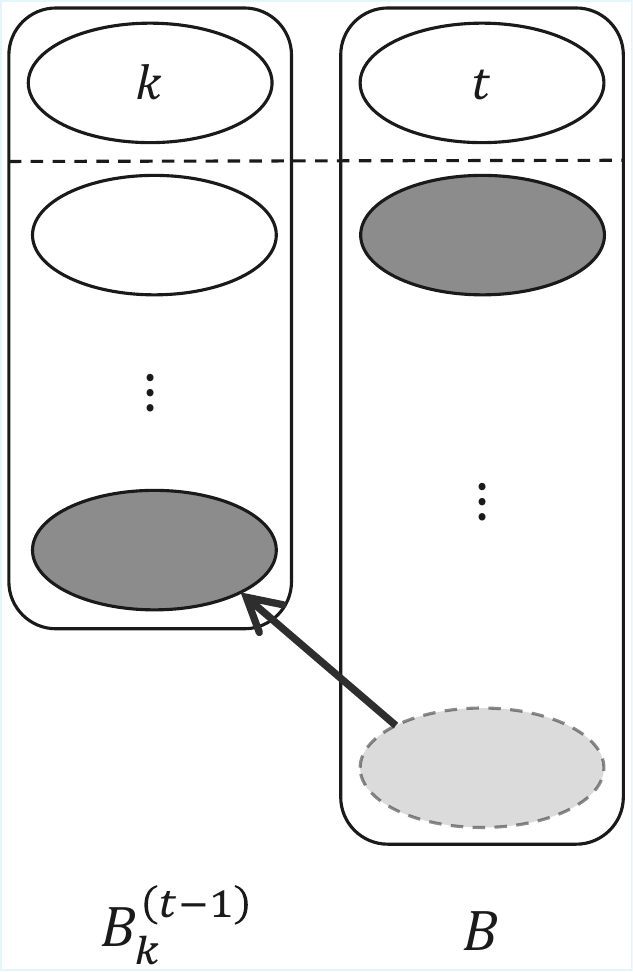}
    \caption{``Move'' (Line~\ref{alg:single:goods:move}).}
    \label{fig:single:goods:move}
  \end{subfigure}
  \begin{subfigure}{0.22\textwidth}
    \centering
    \includegraphics[scale=0.3]{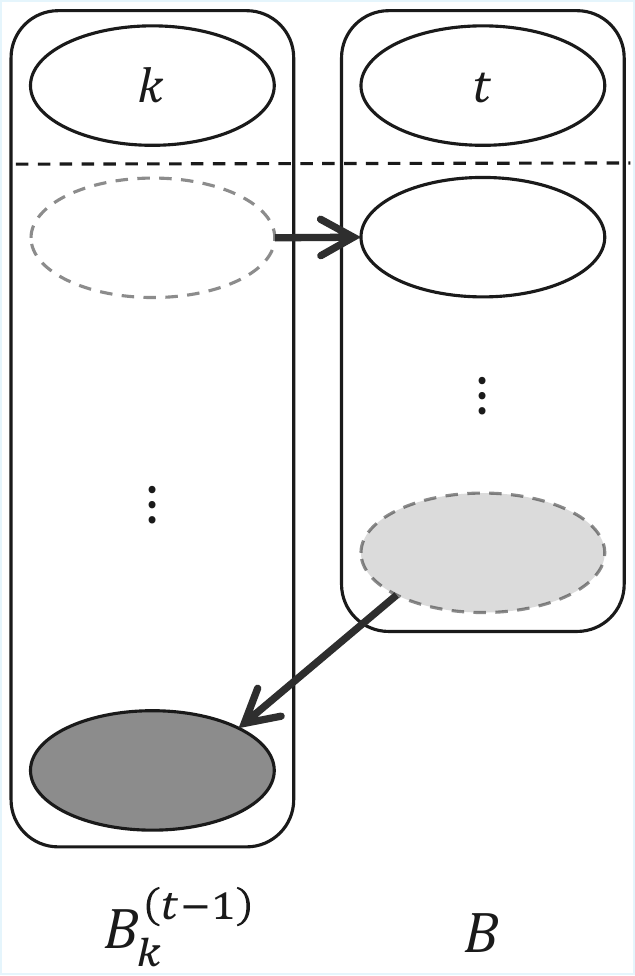}
    \caption{``Swap'' (Line~\ref{alg:single:goods:swap}).}
    \label{fig:single:goods:swap}
  \end{subfigure}

  \caption{%
    An illustration of key steps of \hyperref[alg:single:goods:main-start]{\textproc{ApproxGoods}} in \cref{alg:single:goods}, where each item $g_j\in M$ is denoted by its index~$j$.
    \subref{fig:single:goods:init}
    At Lines~\ref{alg:single:goods:B_k-init-for-start} to \ref{alg:single:goods:B_k-init-for-end}, the item set $M$ is split into the initial bags $B^{(|N|)}_1,B^{(|N|)}_2,\ldots,B^{(|N|)}_{|N|}$ in a greedy manner.
    \subref{fig:single:goods:move}
    At Line~\ref{alg:single:goods:move}, an item is moved from the bag $B$ to $B^{(t-1)}_k$; this continues until either
    the size of these bags are completely exchanged,
    or the bag value $v_i(B)$ falls below a threshold, $\frac{3}{2}\alpha\hat\mu_i$, for every remaining agent $i\in N^{(t)}$.
    \subref{fig:single:goods:swap}
    At Line~\ref{alg:single:goods:swap}, where $|B| = |B^{(t)}_k|$, an item from $B$ is exchanged with another from $B^{(t-1)}_k$; this continues until either
    all items in these bags, except $g_k$ and $g_t$, are completely exchanged from their initial state,
    or $v_i(B)$ falls below $\frac{3}{2}\alpha\hat\mu_i$ for every $i\in N^{(t)}$.
  }
  \label{fig:single:goods}
\end{figure}

We show that \hyperref[alg:single:goods:main-start]{\textproc{ApproxGoods}} defined in \cref{alg:single:goods} certifies \cref{thm:single:goods:two-third} together with the reduction to ordered instances.
To address the above issue, we maintain multiple bags so that at every step, remaining items are left both sufficiently valuable and neither too many nor too few for unassigned agents.
In its main routine (Lines \ref{alg:single:goods:var-init} to \ref{alg:single:goods:assign}), the algorithm incrementally updates a collection of bags to maintain certain invariant conditions, assigning one of them to an unassigned agent in each of the $|N|$ rounds.
Unlike the standard bag-filling method, the bag $B$ assigned in each round starts with items of sufficiently large size and value for any remaining agent and gradually becomes of less size and value through two types of updates: moving an item from $B$ to another bag (\cref{fig:single:goods:move}) and swapping an item in $B$ with one in another bag (\cref{fig:single:goods:swap}).
These updates are stopped when the value of this bag falls below a threshold for every agent.
As formally shown in \cref{sec:single:goods:proof}, valid reduction also helps preserve the invariants, while the polynomial runtime is attained by employing upper bounds on the agents' MMS values.

\paragraph{Initialization of the Bags}
All items in $M$ are partitioned by Lines~\ref{alg:single:goods:B_k-init-for-start} to \ref{alg:single:goods:B_k-init-for-end} into the initial bags $B^{(|N|)}_1,B^{(|N|)}_2,\ldots,B^{(|N|)}_{|N|}$.%
\footnote{Throughout \cref{alg:single:goods}, the superscript indicates the number of agents yet to be assigned a bag, decreasing from $|N|$ to $0$.}
Their respective bag sizes $b_1,b_2,\ldots,b_{|N|} \in \left\{\max\left\{1,q^-\right\},\ldots,q^+ - 1,q^+\right\}$ add up to $|M|$ and are lexicographically minimized among all such integer partitions; this is feasible since $|M| > |N|$ and $q^-|N| \le |M| \le q^+|N|$.
After the most valuable $|N|$ items, $g_1,g_2,\ldots,g_{|N|}$, are placed into distinct bags, the other items are greedily packed into the bags from $B^{(|N|)}_{|N|}$, $B^{(|N|)}_{|N|-1},\ldots,$ to $B^{(|N|)}_1$ in descending order of value, subject to their fixed sizes; see also \cref{fig:single:goods:init} and note that this order is identical across all agents as $\cI$ is ordered.
This configuration satisfies the properties stated in \cref{thm:single:goods:value-sum-MMS}, which serve as the starting point for the invariant conditions formalized in \cref{def:single:goods:invariants}.

\paragraph{Valid Reduction}
After initializing the bags, a valid reduction may occur in Lines~\ref{alg:single:goods:check-two-from-n-th} to \ref{alg:single:goods:valid-reduction-end} if the bundle $\{g_{|N|},g_{|N|+1}\}$ is valued at $\alpha\hat\mu_{i^*}$ or more by some agent $i^*$; recall \cref{thm:common:cor-valid-reduction} and also notice \cref{thm:single:goods:normalize} below.
By repeatedly applying this reduction, we eventually reach an irreducible instance in which every item $g_{|N|+1}, g_{|N|+2}, \ldots, g_{|M|}$ has a value less than $\frac{1}{2} \alpha \hat{\mu}_i$ for each agent $i$ (\cref{thm:single:goods:valid-reduce}), thereby providing a crucial setup for the main routine that follows.

\paragraph{Iterative Updates and Assignment of the Bags}
Once the instance becomes irreducible, the bags are iteratively updated and assigned to agents through the outer for-loop (Lines~\ref{alg:single:goods:outer-for-start} to \ref{alg:single:goods:outer-for-end}).
For each $t\in\{|N|,|N|-1,\ldots,1\}$, the inner for-loop (Lines~\ref{alg:single:goods:inner-for-start} to \ref{alg:single:goods:inner-for-end}) transforms the remaining $t$ bags $B^{(t)}_1, B^{(t)}_2, \ldots, B^{(t)}_{t}$ as a whole into the new $t - 1$ bags $B^{(t-1)}_1, B^{(t-1)}_2, \ldots, B^{(t-1)}_{t-1}$ plus the bundle $B$ that is finally assigned to a remaining agent $i^{(t)}$.
Let's take a look at its first iteration (with $k \gets t - 1$), where only $B$ and $B^{(t-1)}_{t - 1}$ are updated in two phases.
We first move items from $B$ to $B^{(t-1)}_{t - 1}$ until their sizes match each other's initial sizes; see \cref{fig:single:goods:move}.
Then we switch to swapping items between them until $B = (B^{(t)}_{t-1}\setminus\{g_{t-1}\})\cup\{g_t\}$ (and equivalently until $B^{(t-1)}_{t-1} = (B^{(t)}_{t}\setminus\{g_t\})\cup\{g_{t-1}\}$); see \cref{fig:single:goods:swap}.
Similar operations are performed with $k \gets t - 2, t - 3,\ldots,1$, gradually sliding the remaining items (except $g_1,g_2\ldots,g_t$) between the bags, and monotonically reducing both the value and the size of $B$, as long as $v_i(B) \ge \frac{3}{2}\alpha\hat\mu_i$ for some remaining agent $i$.
The validity of this entire routine relies on the invariant conditions and carefully examined in the proof of \cref{thm:single:goods:invariant}.

\subsection{Proof of \cref{thm:single:goods:two-third}} \label{sec:single:goods:proof}

In what follows, we show by induction on $|N|$ that \hyperref[alg:single:goods:main-start]{\textproc{ApproxGoods}} returns an $\alpha$-MMS allocation for any single-category ordered instance of goods $\cI = \left(N, M, (v_i)_{i\in N}, (q^-,q^+)\right)$ and any real constant $\alpha \in \left[\frac{2}{3}, \frac{2}{3 - (\max\left\{|N|, 1\right\})^{-1}}\right]$, in time $O(|N||M|)$.
This induction hypothesis is used to obtain \cref{thm:single:goods:valid-reduce}, which concerns the recursion in Line~\ref{alg:single:goods:valid-reduce}.
We eventually derive \cref{thm:single:goods:two-third} by combining \cref{thm:common:ordered,thm:single:goods:few-items,thm:single:goods:valid-reduce,thm:single:goods:invariant-all,thm:single:goods:runtime}; refer to \cref{thm:common:ordered} for the complexity of the reduction to an ordered instance, which proves dominant.

First, a trivial MMS allocation can be obtained when the number of items is at most that of agents, i.e., when $|M| \le |N|$.

\begin{restatable}{lemma}{SingleGoodsFewItems} \label{thm:single:goods:few-items}
  If Line~\ref{alg:single:goods:bundle-few-items} is reached, an MMS allocation $(A_i)_{i\in N}$ for $\cI$ is returned at Line~\ref{alg:single:goods:allocation-few-items}.
\end{restatable}

\begin{proof}
  The allocation $(A_i)_{i\in N}$ at Line~\ref{alg:single:goods:allocation-few-items} is feasible for $\cI$ due to \cref{eq:def:quotas}.
  If $|M| < |N|$, $\mu_i(\cI) = 0$ holds for each $i\in N$; otherwise, $\mu_i(\cI) = v_i\!\left(g_{|N|}\right)$ holds for each $i\in N$.
  Therefore, it is clear in either case that $(A_i)_{i\in N}$ is an MMS allocation for $\cI$.
\end{proof}

As described before, the bags are initialized by Lines~\ref{alg:single:goods:B_k-init-for-start} to \ref{alg:single:goods:B_k-init-for-end} so that we can ensure the properties in \cref{thm:single:goods:value-sum-MMS}, setting the ground for the invariant conditions defined later in \cref{def:single:goods:invariants}.

\begin{restatable}{lemma}{SingleGoodsValueSumMMS} \label{thm:single:goods:value-sum-MMS}
  The bundles $B^{(|N|)}_1,B^{(|N|)}_2,\ldots,B^{(|N|)}_{|N|}$ defined by Lines~\ref{alg:single:goods:B_k-init-for-start} to \ref{alg:single:goods:B_k-init-for-end} are mutually disjoint and satisfy the following:
  \begin{align}
    &B^{(|N|)}_1 \cup B^{(|N|)}_2 \cup \cdots \cup B^{(|N|)}_{|N|} = M, \label{eq:single:goods:union-init} \\
    &q^- \le |B^{(|N|)}_1| \le |B^{(|N|)}_2| \le \cdots \le |B^{(|N|)}_{|N|}| \le q^+, \label{eq:single:goods:size-init} \\
    \begin{split}
    &v_i\left(B^{(|N|)}_r \cup B^{(|N|)}_{r+1} \cup \cdots \cup B^{(|N|)}_{|N|}\right) \ge \left(|N| - r + 1\right)\mu_i(\cI) \\
    &\quad\forall r\in\{1,2,\ldots,|N|\},\forall i\in N.
    \end{split} \label{eq:single:goods:value-sum-MMS}
  \end{align}
\end{restatable}

\begin{proof}
  Given \cref{eq:def:quotas} and $|M| > |N|$, these bundles are well-defined by Lines~\ref{alg:single:goods:B_k-init-for-start} to \ref{alg:single:goods:B_k-init-for-end}, mutually disjoint, and satisfying \cref{eq:single:goods:union-init,eq:single:goods:size-init}.
  We fix any pair $i\in N$ and $r\in\{1,2,\ldots,|N|\}$, for which the inequality in \cref{eq:single:goods:value-sum-MMS} is shown below.
  Let $(P_k)_{k=1}^{|N|}$ be agent $i$'s MMS partition, which can be assumed, without loss of generality, to satisfy that
  \begin{align*}
    P_1\cup P_2\cup\cdots\cup P_{r-1} &\supseteq \{g_1,g_2,\ldots,g_{r-1}\} &\forall r\in\{1,2,\ldots,|N|\}.
  \end{align*}
  Combining this with $(P_k)_{k=1}^{|N|}\in\cF(\cI)$ and the construction of $B^{(|N|)}_1,B^{(|N|)}_2,\ldots,B^{(|N|)}_{|N|}$ yields that
  \begin{align*}
    v_i\left(B^{(|N|)}_r \cup B^{(|N|)}_{r+1} \cup \cdots \cup B^{(|N|)}_{|N|}\right)
    &\ge v_i\left(P_r \cup P_{r+1} \cup \cdots \cup P_{|N|}\right) \\
    &\ge \left(|N| - r + 1\right) \mu_i(\cI),
  \end{align*}
  which concludes the proof.
\end{proof}



\cref{eq:single:goods:value-sum-MMS} of \cref{thm:single:goods:value-sum-MMS} also provides useful upper bounds on MMS values, $(\hat\mu_i)_{i\in N}$, defined by Line~\ref{alg:single:goods:mu-hat}.

\begin{restatable}{corollary}{SingleGoodsNormalize} \label{thm:single:goods:normalize}
  At Line~\ref{alg:single:goods:mu-hat}, it holds for every $i\in N$ that $\hat\mu_i \ge \mu_i(\cI)$.
\end{restatable}

\begin{proof}
  This follows from Line~\ref{alg:single:goods:mu-hat} and \cref{eq:single:goods:value-sum-MMS} of \cref{thm:single:goods:value-sum-MMS}.
\end{proof}

As $\frac{2}{3 - \max\left\{|N|,1\right\}^{-1}}$ is non-increasing in $|N|$, \cref{thm:common:cor-valid-reduction} can be applied together with the induction hypothesis and \cref{thm:single:goods:normalize}.

\begin{restatable}{lemma}{SingleGoodsValidReduce} \label{thm:single:goods:valid-reduce}
  If Line \ref{alg:single:goods:valid-reduce-assign} is reached, an $\alpha$-MMS allocation $(A_i)_{i\in N}$ for $\cI$ is returned at Line~\ref{alg:single:goods:allocation-valid-reduction}.
  Otherwise, it follows at Line~\ref{alg:single:goods:var-init} that \hspace{-2pt} $v_i(g) < \frac{1}{2}\,\alpha\,\hat\mu_i$ \hspace{-2pt} for every $g\in M\setminus\left\{g_1,g_2,\ldots,g_{|N|}\right\}$ and $i\in N$.
\end{restatable}

\begin{proof}
  By applying \cref{thm:common:valid-reduction} for $C^* = M$ and $d = 1$, we see that the instance $\cI'$ is well-defined by Line~\ref{alg:single:goods:valid-reduced-instance} as a single-category ordered instance of goods.
  Given the induction hypothesis and that
  \begin{align*}
    \frac{2}{3} \le \alpha \le \frac{2}{3 - \max\left\{|N|,1\right\}^{-1}} \le \frac{2}{3 - \max\left\{|N\setminus\{i^*\}|,1\right\}^{-1}},
  \end{align*}
  an $\alpha$-MMS allocation $(A_i)_{i\in N\setminus\{i^*\}}$ for $\cI'$ is obtained at Line~\ref{alg:single:goods:valid-reduce}.
  Because Line~\ref{alg:single:goods:check-two-from-n-th}, \ref{alg:single:goods:valid-reduce-assign}, and \cref{thm:single:goods:normalize} together ensure that
  \begin{align*}
    v_{i^*}(A_{i^*}) \ge \alpha\,\hat\mu_{i^*} \ge \alpha\,\mu_{i^*}(\cI),
  \end{align*}
  \cref{thm:common:cor-valid-reduction} allows $(A_i)_{i\in N}$ returned at Line~\ref{alg:single:goods:allocation-valid-reduction} to be an $\alpha$-MMS allocation for $\cI$.
  As the instance $\cI$ is ordered, Line~\ref{alg:single:goods:check-two-from-n-th} immediately gives the second statement of the \cref{thm:single:goods:valid-reduce}.
\end{proof}

We formally define the invariant conditions maintained throughout the iterations over Lines \ref{alg:single:goods:outer-for-start} to \ref{alg:single:goods:outer-for-end}.

\begin{definition} \label{def:single:goods:invariants}
  Along with the outer for-loop over Lines \ref{alg:single:goods:outer-for-start} to \ref{alg:single:goods:outer-for-end}, we consider the following conditions for each $t\in\{|N|, |N| - 1, \ldots, 1, 0\}$:
  \leqnomode
  \begin{align}
    \tag{C1} \mspace{0mu}& \mathrlap{B^{(t)}_1,B^{(t)}_2,\ldots,B^{(t)}_t, A_{i^{(t + 1)}}, A_{i^{(t + 2)}}, \ldots, A_{i^{(|N|)}} \mbox{ are disjoint.}} \label{cond:single:goods:disjoint} \\[0pt]
    \tag{C2} \mspace{0mu}& \mathrlap{B^{(t)}_1\cup B^{(t)}_2\cup\cdots\cup B^{(t)}_t\cup A_{i^{(t + 1)}}\cup A_{i^{(t + 2)}}\cup\cdots\cup A_{i^{(|N|)}} = M.} \label{cond:single:goods:union} \\[0pt]
    \tag{C3} \mspace{0mu}& \mathrlap{q^- \le |B^{(t)}_1| \le |B^{(t)}_2| \le \cdots \le |B^{(t)}_t| \le q^+.} \label{cond:single:goods:size} \\[0pt]
    \tag{C4}
    \mspace{0mu}&
    \begin{aligned}
      & B^{(t)}_k \cap \left\{g_1,g_2,\ldots,g_{|N|}\right\} = \{g_k\}
      & \forall k\in\{1,2,\ldots,t\}.
    \end{aligned}
    \label{cond:single:goods:top-n} \\[1pt]
    \tag{C5}
    \mspace{0mu}&
    \mathrlap{
      \begin{aligned}
        & v_i(h_1) \le v_i(h_2) \le \cdots \le v_i(h_t) & \forall i\in N^{(t)},\\[-2pt]
        & \mathrlap{\forall h_1\in B^{(t)}_1\setminus\{g_1\},\forall h_2\in B^{(t)}_2\setminus\{g_2\}, \ldots, \forall h_t\in B^{(t)}_t\setminus\{g_t\}.}
    \end{aligned}} \label{cond:single:goods:value-order} \\[1pt]
    \tag{C6}
    \mspace{0mu} &
    \begin{aligned}
    & v_i\left(\bigcup_{k=r}^{t}B^{(t)}_k\right) \ge \left(t - r + 1 - (|N| - t)\left(\frac{3}{2}\,\alpha - 1\right)\right) \hat\mu_i \\[-2pt]
    & \forall r\in\{1,2,\ldots,t\},\forall i\in N^{(t)}.
    \end{aligned}
    \label{cond:single:goods:value-sum}
  \end{align}
  \reqnomode
\end{definition}

\cref{thm:single:goods:value-sum-MMS,thm:single:goods:normalize} together guarantee that these conditions initially hold for $t = |N|$.

\begin{restatable}{lemma}{SingleGoodsInvariantInit} \label{thm:single:goods:invariant-init}
  At Line~\ref{alg:single:goods:var-init}, Conditions \eqref{cond:single:goods:disjoint} to \eqref{cond:single:goods:value-sum} hold for $t = |N|$.
\end{restatable}

\begin{proof}
  Conditions \eqref{cond:single:goods:disjoint} to \eqref{cond:single:goods:value-order} for $t = |N|$ follow from \cref{thm:single:goods:value-sum-MMS}.
  Condition \eqref{cond:single:goods:value-sum} for $t = |N|$ follows from Line~\ref{alg:single:goods:mu-hat}.
\end{proof}

Key to the overall correctness is \cref{thm:single:goods:invariant}: each iteration assigns an agent a sufficiently valuable and feasible bundle while ensuring that all the invariants are properly inherited by the next iteration.

\begin{lemma} \label{thm:single:goods:invariant}
  Let $s\in\{|N|, |N| - 1 \ldots, 1\}$ be arbitrary.
  Suppose that the for-loop over Lines \ref{alg:single:goods:outer-for-start} to \ref{alg:single:goods:outer-for-end} has successfully iterated for $t \in \{|N|,|N|-1,\ldots,s+1\}$, and that Conditions \eqref{cond:single:goods:disjoint} to \eqref{cond:single:goods:value-sum} now hold for $t = s$.
  Then the next iteration with $t = s$ succeeds, where the following hold at Line~\ref{alg:single:goods:before-assign}:
  \begin{align}
    & q^- \le |B| \le q^+. \label{eq:single:goods:bag-size} \\
    & \exists i^{(s)}\in N^{(s)} \ \ \mathrm{s.t.} \ \ v_{i^{(s)}}(B) \ge \alpha\,\hat\mu_{i^{(s)}}. \label{eq:single:goods:bag-value}
  \end{align}
  Furthermore, Conditions \eqref{cond:single:goods:disjoint} to \eqref{cond:single:goods:value-sum} hold for $t = s - 1$ at Line~\ref{alg:single:goods:assign} of the same iteration.
\end{lemma}

\begin{proof}
  First, we verify that Line~\ref{alg:single:goods:before-assign} is indeed reached, with all preceding operations done successfully; namely we show that the desired item exists at Line~\ref{alg:single:goods:get-from-B} and \ref{alg:single:goods:get-from-B_k}, as well as that the while-loop terminates.
  Given Condition \eqref{cond:single:chores:value-order} for $t = s$, neither $|B|$ nor $v_i(B)$ for any $i\in N$ increases after Line~\ref{alg:single:goods:bag-init} until \ref{alg:single:goods:before-assign}.
  At Line~\ref{alg:single:goods:bag-init}, $B\ni g_s$ and $B^{(s-1)}_k\ni g_k$ for each $k\in\{1,2,\ldots,s-1\}$ hold due to Condition \eqref{cond:single:goods:top-n} for $t = s$, remaining true until Line~\ref{alg:single:goods:assign} by the definition of Lines~\ref{alg:single:goods:get-from-B} and \ref{alg:single:goods:get-from-B_k}.
  Let us go through a single iteration of the inner for-loop (Lines \ref{alg:single:goods:inner-for-start} to \ref{alg:single:goods:inner-for-end}) with an arbitrary $k\in\{1,2,\ldots,s-1\}$.
  Clearly, nothing occurs if $k < s - 1$ and the previous iteration has stopped with $v_i(B) < \frac{3}{2}\alpha\hat\mu_i$ for every $i\in N^{(s)}$.
  Otherwise, at the beginning of this iteration, we have that $B\setminus\{g_s\} = B^{(s)}_{k+1}\setminus\{g_{k+1}\}$, that $B^{(s-1)}_{k} = B^{(s)}_k$, and hence that $|B| = |B^{(s)}_{k+1}| \ge |B^{(s)}_k| = |B^{(s-1)}_k|$ due to Condition \eqref{cond:single:goods:size} for $t = s$.
  Any single iteration of the while-loop (Lines \ref{alg:single:goods:inner-while-start} to \ref{alg:single:goods:inner-while-end}) keeps $B \cup B^{(s-1)}_{k}$ unchanged and decreases $|B|$ by at most $1$, while maintaining that $g_s \in B$, $g_k \in B^{(s)}_k$, and $|B| \ge |B^{(s)}_k|$.
  Therefore, $B\setminus\{g_{s}\}$ is never be empty at Line~\ref{alg:single:goods:get-from-B}, until which it has been maintained that $B\setminus\{g_s\}\neq B^{(s)}_k\setminus\{g_k\}$.
  We then obtain that $B^{(s-1)}_k\setminus\{g_{k}\} \neq \emptyset$ at Line~\ref{alg:single:goods:get-from-B_k}, where $|B^{(s-1)}_{k}| = |B^{(s)}_{k+1}| \ge |B^{(s)}_{k}| = |B|\ge 2$.
  Thanks to the choice of the items $g$ and $h$ at Lines \ref{alg:single:goods:get-from-B} and \ref{alg:single:goods:get-from-B_k}, respectively, the while-loop successfully terminates.

  Next, we prove that \cref{eq:single:goods:bag-size,eq:single:goods:bag-value} hold at Line~\ref{alg:single:goods:before-assign}, where we define the following:
  \begin{align}
    k^* \coloneqq
    \begin{cases}
      0 & \text{\hspace{16pt} if }\ \exists i\in N^{(s)} \mbox{ \ s.t. \ } v_i(B) \ge \frac{3}{2}\,\alpha\,\hat\mu_i; \\[4pt]
      \mathrlap{\min \left(\left\{k\in\{1,2,\ldots,s-1\} \mid B^{(s-1)}_k \neq B^{(s)}_k\right\} \cup \{s\}\right)} & \\
       & \text{\hspace{16pt} otherwise.\hspace{100pt}}
    \end{cases}
    \label{eq:single:goods:k*}
  \end{align}
  The value of $k^*$ represents at which point the updates on $B$ have stopped:
  $k^* = 0$ when all possible updates on $B$ run out, i.e., when it holds at Line~\ref{alg:single:goods:before-assign} that $B = (B^{(s)}_1\setminus\{g_1\})\cup\{g_s\}$, $B^{(s-1)}_1 = (B^{(s)}_2\setminus\{g_2\})\cup\{g_1\}$, \ldots, and $B^{(s-1)}_{s-1} = (B^{(s)}_{s}\setminus\{g_s\})\cup\{g_{s-1}\}$;
  $k^* = s$ if $B$ has not been updated at Line~\ref{alg:single:goods:before-assign} since Line~\ref{alg:single:goods:bag-init};
  otherwise (if $0 < k < s$), $B^{(s-1)}_{k^*}$ denotes the bag that has been last updated at either Line~\ref{alg:single:goods:move} or \ref{alg:single:goods:swap} (as $B^{(t-1)}_k$).
  Based on the observations in the previous paragraph and Conditions \eqref{cond:single:goods:disjoint} to \eqref{cond:single:goods:top-n} for $t = s$, we see that the following hold at Line~\ref{alg:single:goods:before-assign}:
  \begin{align}
    & B^{(s-1)}_{k} = B^{(s)}_{k} & & \forall k\in\{1, 2, \ldots, k^*-1\}. \label{eq:single:goods:B_k-same} \\
    & B^{(s-1)}_{k} = (B^{(s)}_{k+1}\setminus\{g_{k+1}\})\cup\{g_k\} & & \forall k\in\{k^* + 1, \ldots, s-1\}. \label{eq:single:goods:B_k-shift} \\
    &\mathrlap{
      \begin{cases}
        B = \left(B^{(s)}_{\max\left\{k^*,1\right\}} \setminus \left\{g_{\max\left\{k^*,1\right\}}\right\}\right) \cup \{g_s\} & \text{\hspace{-12pt}if }\ k^* \in \{0, s\}; \\
        \left(B\setminus\left\{g_{s}\right\}\right) \cup B^{(s-1)}_{k^*} = B^{(s)}_{k^*} \cup \left(B^{(s)}_{k^*+1}\setminus\left\{g_{k^*+1}\right\}\right) & \text{\hspace{-1pt}otherwise.}
    \end{cases}} \label{eq:single:goods:union} \\
    &B \cap \left\{g_1,g_2,\ldots,g_{|N|}\right\} = \{g_s\}. \label{eq:single:goods:special-item} \\
    & \mathrlap{q^- \le \left|B^{(s)}_{\max\left\{k^*,1\right\}}\right| \le |B| \le \left|B^{(s)}_{\min\left\{k^*+1,s\right\}}\right| \le q^+.} \label{eq:single:goods:size-ineq} \\
    &\mathrlap{B^{(s - 1)}_1,B^{(s - 1)}_2,\ldots,B^{(s - 1)}_{s - 1}, \mbox{ and $B$ are disjoint.}} \label{eq:single:goods:disjoint}
  \end{align}
  In particular, \cref{eq:single:goods:size-ineq} immediately gives \cref{eq:single:goods:bag-size}.
  Moreover, applying the latter statement of \cref{thm:single:goods:valid-reduce} to the item $g\in B\setminus\{g_s\}$ at Line~\ref{alg:single:goods:get-from-B}, along with the first condition of the while-loop, guarantees that
  \begin{align*}
    &\exists i\in N^{(s)} \quad \mathrm{s.t.} \quad v_i(B) > \frac{3}{2}\,\alpha\,\hat\mu_i - \frac{1}{2}\,\alpha\,\hat\mu_i = \alpha\,\hat\mu_i,
  \end{align*}
  whenever $B$ gets updated at either Line~\ref{alg:single:goods:move} or \ref{alg:single:goods:swap}.
  Condition \eqref{cond:single:goods:value-sum} for $t = s$ also gives that at Line~\ref{alg:single:goods:bag-init},
  \begin{align*}
    v_i(B) = v_i\!\left(B^{(s)}_{s}\right)
    &\ge \left(1 - (|N| - s)\left(\frac{3}{2}\,\alpha - 1\right)\right)\hat\mu_i \\
    &\ge \left(1 - (|N| - 1)\left(\frac{3}{2}\,\alpha - 1\right)\right)\hat\mu_i \\
    &\ge \frac{2}{3 - |N|^{-1}}\,\hat\mu_i \\
    &\ge \alpha\,\hat\mu_i
    &\forall i\in N^{(s)},
  \end{align*}
  because $\alpha \in \left[\frac{2}{3},\frac{2}{3 - |N|^{-1}}\right]$.
  Therefore, we obtain that \cref{eq:single:goods:bag-value} always holds at Line~\ref{alg:single:goods:before-assign}.

  Finally, we show that Conditions \eqref{cond:single:goods:disjoint} to \eqref{cond:single:goods:value-sum} hold for $t = s - 1$ at Line~\ref{alg:single:goods:assign}.
  Conditions \eqref{cond:single:goods:disjoint} to \eqref{cond:single:goods:value-order} for $t = s - 1$ follow from those for $t = s$ and \crefrange{eq:single:goods:B_k-same}{eq:single:goods:disjoint}.
  We fix any pair $r\in\{1,2,\ldots,s-1\}$ and $i\in N^{(s-1)}\subset N^{(s)}$, for which the inequality in Condition \eqref{cond:single:goods:value-sum} for $t = s - 1$ is obtained in either of the following cases.

  \begin{case} \label{case:single:goods:1}
    Suppose $r > k^*$.
    Then given \cref{eq:single:goods:B_k-shift}, Conditions \eqref{cond:single:goods:top-n} and \eqref{cond:single:goods:value-sum} for $t = s$, and that $\alpha\ge\frac{2}{3}$, we obtain that
    \begin{align*}
      \begin{split}
        v_i\left(\bigcup_{k=r}^{s-1}B^{(s-1)}_k\right)
        &= v_i\left(\bigcup_{k=r}^{s-1}\left(\left(B^{(s)}_{k+1}\setminus\{g_{k+1}\}\right)\cup\{g_k\}\right)\right) \\
        &\ge v_i\left(\bigcup_{k=r+1}^{s}B^{(s)}_{k}\right) \\
        &\ge \left(s - r - (|N| - s)\left(\frac{3}{2}\,\alpha - 1\right)\right) \hat\mu_i \\
        &\ge \left(s - r - (|N| - s + 1)\left(\frac{3}{2}\,\alpha - 1\right)\right) \hat\mu_i.
      \end{split}
    \end{align*}
  \end{case}

  \begin{case} \label{case:single:goods:2}
    Suppose $1 \le r \le k^*$.
    Then given the definition of $k^*$ in \cref{eq:single:goods:k*}, we must have that $v_i(B) < \frac{3}{2}\alpha\hat\mu_i$
    at Line~\ref{alg:single:goods:before-assign}.
    It also follows from that $r \le k^*$ and \crefrange{eq:single:goods:B_k-same}{eq:single:goods:special-item} that
    \begin{align*}
      B \cup B^{(s-1)}_r \cup B^{(s-1)}_{r + 1} \cup\cdots\cup B^{(s-1)}_{s-1} =  B^{(s)}_r \cup B^{(s)}_{r + 1} \cup\cdots\cup B^{(s)}_{s}.
    \end{align*}
    These observations and Condition \eqref{cond:single:goods:value-sum} for $t = s$ together establish that
    \begin{align*}
      v_i\left(\bigcup_{k=r}^{s-1}B^{(s-1)}_k\right)
      &\ge v_i\left(\bigcup_{k=r}^{s}B^{(s)}_{k}\right) - v_i(B) \\
      &> \left(s - r + 1 - (|N| - s)\left(\frac{3}{2}\,\alpha - 1\right)\right) \hat\mu_i - \frac{3}{2}\,\alpha\,\hat\mu_i \\
      &= \left(s - r - (|N| - s + 1)\left(\frac{3}{2}\,\alpha - 1\right)\right) \hat\mu_i.
    \end{align*}
  \end{case}

  The proof of \cref{thm:single:goods:invariant} is now completed.
\end{proof}

\cref{thm:single:goods:invariant-init,thm:single:goods:invariant,thm:single:goods:normalize} together lead to \cref{thm:single:goods:invariant-all}.

\begin{restatable}{corollary}{SingleGoodsInvariantAll} \label{thm:single:goods:invariant-all}
  If Line~\ref{alg:single:goods:var-init} is reached, an $\alpha$-MMS allocation $(A_i)_{i\in N}$ for $\cI$ is returned at Line~\ref{alg:single:goods:allocation-final}.
\end{restatable}

\begin{proof}
  \cref{thm:single:goods:invariant-init,thm:single:goods:invariant} together ensure that Line~\ref{alg:single:goods:allocation-final} is successfully reached after Line~\ref{alg:single:goods:var-init}, with Conditions \eqref{cond:single:goods:disjoint} to \eqref{cond:single:goods:value-sum} fulfilled for every $t\in\left\{|N|,|N|-1,\ldots,1,0\right\}$.
  Also given Conditions \eqref{cond:single:goods:disjoint} to \eqref{cond:single:goods:union} for $t = 0$ and \cref{eq:single:goods:bag-size} of \cref{thm:single:goods:invariant}, it is implied that $(A_i)_{i\in N} = \left(A_{i^{(t)}}\right)_{t\in\{|N|,|N|-1,\ldots,1\}}$ is a feasible allocation for $\cI$.
  \cref{thm:single:goods:normalize} and \cref{eq:single:goods:bag-value} of \cref{thm:single:goods:invariant} also yield that
  \begin{align*}
    &v_{i^{(t)}}\!\left(A_{i^{(t)}}\right) \ge \alpha\,\hat\mu_{i^{(t)}} \ge \alpha\,\mu_{i^{(t)}}(\cI) &\forall t\in\left\{|N|,|N|-1,\ldots,1\right\},
  \end{align*}
  establishing that $(A_i)_{i\in N}$ is an $\alpha$-MMS allocation for $\cI$.
\end{proof}



\begin{lemma} \label{thm:single:goods:runtime}
  \hyperref[alg:single:goods:main-start]{\textproc{ApproxGoods}} runs in time $O(|N||M|)$.
\end{lemma}

\begin{proof}
  The recursion in Line~\ref{alg:single:goods:valid-reduce} reduces $|N|$ by one, while Lines \ref{alg:single:goods:B_k-init-for-start} to \ref{alg:single:goods:valid-reduced-instance} run in time $O(|M|)$.
  In each of the $|N|$ iterations of the outer for-loop (Lines \ref{alg:single:goods:outer-for-start} to \ref{alg:single:goods:outer-for-end}), the inner for-loop (Lines \ref{alg:single:goods:inner-for-start} to \ref{alg:single:goods:inner-for-end}) iterates $t - 1 < |N| < |M|$ times, during which the while-loop (Lines \ref{alg:single:goods:inner-while-start} to \ref{alg:single:goods:inner-while-end}) iterates at most $\sum_{k=2}^{t}(|B^{(t)}_k|-1) \le |M|$ times in total.
\end{proof}



\subsection{Tight Instances for the Algorithm} \label{sec:single:goods:tight}

Our analysis above is tight for any number of agents, even when they have identical valuations.

\begin{restatable}{theorem}{SingleGoodsTight} \label{thm:single:goods:tight}
  For any $n\in\{1,2,\ldots\}$, there is a single-category ordered instance $\cI_n$ with $n$ agents, $3n$ goods, and identical valuations, which satisfies the following: for any $\beta \in \left(\frac{2n}{3n - 1},\infty\right)$, no $\beta$-MMS allocation for $\cI_n$ is obtained by \hyperref[alg:single:goods:main-start]{\textproc{ApproxGoods}} given $\cI_{n}$ and any $\alpha\in\mathbb{R}$ as input.
\end{restatable}


\begin{proof}
  Let $n$ be an arbitrary positive integer.
  We prove the claimed property of the single-category ordered instance of goods $\cI_n = \left(N, M, (v_i)_{i\in N}, (q^-, q^+)\right)$ defined as follows:
  \begin{align*}
    N & \coloneqq \{1,2,\ldots,n\}, \\
    M & \coloneqq  \{g_1,g_2,\ldots,g_{3n}\},\\
    \left(q^-,q^+\right) & \coloneqq (3,3); \\
    v_i(g_j) & \coloneqq
    \begin{cases}
      \frac{2n - j}{3n - 1} & \text{if } j \le n, \\[3pt]
      \frac{1}{6n - 2}\left\lceil\frac{5n + 1 - j}{2}\right\rceil & \text{if } n < j \le 3n - 2, \\[3pt]
      \frac{3n - j}{3n - 1} & \text{otherwise}
    \end{cases} & \forall g_j\in M,\forall i\in N.
  \end{align*}
  Given the partition $(P_k)_{k=1}^{n}\in\cF(\cI_n)$ of $M$ defined as
  \begin{align*}
    P_k &\coloneqq
    \begin{cases}
      \{g_k, g_{n + k}, g_{3n + 1 - k}\} & \text{ if } k \le 2, \\
      \{g_{k}, g_{3n - 2k + 3}, g_{3n - 2k + 4}\} & \text{ otherwise}
    \end{cases}
    &\forall k\in \{1,2,\ldots,n\},
  \end{align*}
  we have that $\mu_i(\cI_n) = v_i(P_k) = 1$ for any $k\in\{1,2,\ldots,n\}$ and $i\in N$.
  Let us fix any $\beta \in \left(\frac{2n}{3n - 1},\infty\right)$, and suppose that \hyperref[alg:single:goods:main-start]{\textproc{ApproxGoods}} runs given $\cI_{n}$ and some $\alpha\in\bR$ as input.
  Because $|M| > |N|$ at Line~\ref{alg:single:goods:check-few-items}, and it also follows at Line~\ref{alg:single:goods:mu-hat} that
  \begin{align*}
    v_i\left(B^{(|N|)}_k\right) &= \mathrlap{v_i(\{g_k, g_{3n - 2k + 1}, g_{3n - 2k + 2}\})} \\
    & =
    \begin{cases}
      \frac{2n}{3n - 1} &\text{if }\, k = 1, \\[3pt]
      \frac{3n}{3n - 1} &\text{otherwise}
    \end{cases}
    &\forall k\in\{1,2,\ldots,n\},\forall i\in N,
  \end{align*}
  Line~\ref{alg:single:goods:check-two-from-n-th} is reached with $\hat\mu_i = 1$ for every $i\in N$.
  The rest of the analysis depends on the regime of $\alpha\in\bR$ and confirms that the algorithm never obtains a $\beta$-MMS allocation in any case.



  \begin{case}
    Suppose $\alpha \le \frac{2n}{3n - 1}$.
    Because Line~\ref{alg:single:goods:allocation-valid-reduction} is reached with
    \begin{align*}
      &v_{i^*}(A_{i^*}) = v_{i^*}(\{g_n, g_{n + 1}, g_{3n}\}) = \frac{2n}{3n - 1} < \beta = \beta\, \mu_{i^*}(\cI_n),
    \end{align*}
    $(A_i)_{i\in N}$ is then not a $\beta$-MMS allocation for $\cI$.
  \end{case}

  \begin{case}
    Suppose $\frac{2n}{3n - 1} < \alpha \le 1$.
    Then $B = B^{(|N|)}_t$ for each $t\in\{n,n-1\ldots,2\}$ at Line~\ref{alg:single:goods:assign}.
    In the final iteration with $t = 1$, Line~\ref{alg:single:goods:bag-init} is thus reached with
    \begin{align*}
      &v_i(B) = v_i\!\left(B^{(1)}_1\right) = v_i\!\left(B^{(|N|)}_1\right) = \frac{2n}{3n - 1} < \alpha = \alpha\,\hat\mu_i &\forall i\in N^{(1)},
    \end{align*}
    which immediately makes Line~\ref{alg:single:goods:before-assign} unsuccessful.
  \end{case}

  \begin{case}
    Suppose $\alpha > 1$.
    Then \hyperref[alg:single:goods:main-start]{\textproc{ApproxGoods}} never successfully terminates; otherwise, Lines~\ref{alg:single:goods:before-assign} and \ref{alg:single:goods:assign} would imply that
    \begin{align*}
      v_i(M) = \sum_{t=1}^{n} v_{i}\!\left(A_{i^{(t)}}\right) & = \sum_{t=1}^{n} v_{i^{(t)}}\!\left(A_{i^{(t)}}\right) \\
      & \ge \sum_{t=1}^n \alpha\,\hat\mu_{i^{(t)}} = n\alpha > n & \forall i\in N,
    \end{align*}
    which contradicts the definition of $\cI_n$.
  \end{case}

  The proof of \cref{thm:single:goods:tight} is now completed.
\end{proof}

\section{Discussion} \label{sec:discuss}

A promising future work is to improve our MMS approximations, given their better counterparts in the unconstrained setting~\citep{huang2025fptas,huang2023reduction}.
A tantalizing question is whether imposing lower quotas strictly decreases the best MMS approximation achievable.
Allowing for non-additive valuations would also be of interest.

Beyond MMS fairness, it may be even more desirable to obtain simultaneous guarantees across different fairness criteria \citep{garg2026exploring}, such as a pair of envy-based and share-based notions \citep{amanatidis2020multiple,chaudhury2021little,akrami2025achieving,ashuri2025simultaneously}.
Moreover, extending our results to randomized allocations \citep{aziz2019probabilistic,aziz2024best,babaioff2022best,akrami2023randomized} and/or online arrivals \citep{aleksandrov2020online,zhou2023multi,procaccia2024honor,schiffer2025improved,kulkarni2025online} would be of both theoretical and practical significance.


In addition to fairness, allocation efficiency is another central issue.
A large body of work has sought to simultaneously achieve fairness and efficiency, both in the unconstrained setting \citep{caragiannis2019unreasonable,barman2018finding,bei2021price,mahara2026existence,barman2025introspectively} and under constraints \citep{shoshan2023efficient,wu2025approximate,cookson2025constrained,igarashi2025fair}, motivating us to explore the trade-off between MMS approximation and various efficiency measures both under quota constraints and beyond.
It would also be relevant to quantify the \emph{price} of (i.e., the welfare loss due to) lower quotas, as suggested by \citet{lam2025exact}.



\begin{acks}
  This work was partially supported by JST FOREST Grant Number JPMJFR226O.
\end{acks}



\bibliographystyle{ACM-Reference-Format}
\bibliography{bib/fair-division, bib/matching, bib/complexity, bib/job-scheduling, bib/knapsack, bib/optimization, bib/online}

@inproceedings{akrami2023randomized,
  title     = {Randomized and deterministic maximin-share approximations for fractionally subadditive valuations},
  author    = {Akrami, Hannaneh and Mehlhorn, Kurt and Seddighin, Masoud and Shahkarami, Golnoosh},
  booktitle = {Proceedings of the 37th International Conference on Neural Information Processing Systems},
  pages     = {58821--58832},
  paper     = {https://arxiv.org/pdf/2308.14545},
  review    = {https://openreview.net/forum?id=I3k2NHt1zu},
  year      = {2023}
}

@inproceedings{akrami2023simplification,
  title     = {Simplification and improvement of {MMS} approximation},
  author    = {Akrami, Hannaneh and Garg, Jugal and Sharma, Eklavya and Taki, Setareh},
  booktitle = {Proceedings of the 32nd International Joint Conference on Artificial Intelligence},
  pages     = {2485--2493},
  year      = {2023},
  memo      = {3/4-MMS poly simpler analysis, O(1/n) improvement},
  paper     = {https://arxiv.org/pdf/2303.16788}
}

@inproceedings{akrami2024breaking,
  title     = {Breaking the $3/4$ barrier for approximate maximin share},
  author    = {Akrami, Hannaneh and Garg, Jugal},
  booktitle = {Proceedings of the 35th Annual ACM-SIAM Symposium on Discrete Algorithms},
  pages     = {74--91},
  year      = {2024},
  memo      = {(3/4 + 3/3836)-MMS},
  paper     = {https://epubs.siam.org/doi/pdf/10.1137/1.9781611977912.4}
}

@inproceedings{akrami2025achieving,
  title     = {Achieving maximin share and {EFX/EF1} guarantees simultaneously},
  author    = {Akrami, Hannaneh and Rathi, Nidhi},
  booktitle = {Proceedings of the 39th AAAI Conference on Artificial Intelligence},
  pages     = {13529--13537},
  year      = {2025}
}

@inproceedings{akrami2026matroids,
  title     = {Matroids are equitable},
  author    = {Akrami, Hannaneh and Raj, Roshan and V{\'e}gh, L{\'a}szl{\'o} A.},
  booktitle = {Proceedings of the 37th Annual ACM-SIAM Symposium on Discrete Algorithms},
  pages     = {5843--5860},
  year      = {2026},
  paper     = {https://arxiv.org/pdf/2507.12100}
}

@article{amanatidis2017approximation,
  title     = {Approximation algorithms for computing maximin share allocations},
  author    = {Amanatidis, Georgios and Markakis, Evangelos and Nikzad, Afshin and Saberi, Amin},
  journal   = {ACM Transactions on Algorithms},
  volume    = {13},
  number    = {4},
  pages     = {1--28},
  year      = {2017},
  publisher = {ACM New York, NY, USA},
  paper     = {https://dl.acm.org/doi/pdf/10.1145/3147173},
  memo      = {Journal ver. of `amanatidis2015approximation`}
}

@article{amanatidis2020multiple,
  title     = {Multiple birds with one stone: Beating $1/2$ for {EFX and GMMS} via envy cycle elimination},
  author    = {Amanatidis, Georgios and Markakis, Evangelos and Ntokos, Apostolos},
  journal   = {Theoretical Computer Science},
  volume    = {841},
  pages     = {94--109},
  year      = {2020},
  publisher = {Elsevier},
  paper     = {https://arxiv.org/pdf/1909.07650}
}

@article{amanatidis2023fair,
  title     = {Fair division of indivisible goods: Recent progress and open questions},
  author    = {Amanatidis, Georgios and Aziz, Haris and Birmpas, Georgios and Filos-Ratsikas, Aris and Li, Bo and Moulin, Herv{\'e} and Voudouris, Alexandros A. and Wu, Xiaowei},
  journal   = {Artificial Intelligence},
  volume    = {322},
  pages     = {103965},
  year      = {2023},
  publisher = {Elsevier},
  paper     = {https://www.sciencedirect.com/science/article/pii/S000437022300111X/pdfft?md5=4754f97775b975745cbace8fb15bc74d&pid=1-s2.0-S000437022300111X-main.pdf}
}

@inproceedings{ashuri2025simultaneously,
  title     = {Simultaneously satisfying {MXS and EFL}},
  author    = {Ashuri, Arash and Gkatzelis, Vasilis},
  booktitle = {Proceedings of the 26th ACM Conference on Economics and Computation},
  pages     = {689--718},
  year      = {2025},
  paper     = {https://arxiv.org/pdf/2412.00358}
}

@inproceedings{aziz2017algorithms,
  title     = {Algorithms for max-min share fair allocation of indivisible chores},
  author    = {Aziz, Haris and Rauchecker, Gerhard and Schryen, Guido and Walsh, Toby},
  booktitle = {Proceedings of the 31st AAAI Conference on Artificial Intelligence},
  pages     = {335--341},
  year      = {2017},
  paper     = {https://cdn.aaai.org/ojs/10582/10582-13-14110-1-2-20201228.pdf}
}

@article{aziz2019probabilistic,
  title     = {A probabilistic approach to voting, allocation, matching, and coalition formation},
  author    = {Aziz, Haris},
  journal   = {The Future of Economic Design: The Continuing Development of a Field as Envisioned by Its Researchers},
  pages     = {45--50},
  year      = {2019},
  publisher = {Springer},
  paper     = {https://link.springer.com/content/pdf/10.1007/978-3-030-18050-8.pdf}
}

@article{aziz2024best,
  title     = {Best of both worlds: Ex ante and ex post fairness in resource allocation},
  author    = {Aziz, Haris and Freeman, Rupert and Shah, Nisarg and Vaish, Rohit},
  journal   = {Operations Research},
  volume    = {72},
  number    = {4},
  pages     = {1674--1688},
  year      = {2024},
  publisher = {INFORMS},
  paper     = {https://pubsonline.informs.org/doi/epdf/10.1287/opre.2022.2432}
}

@inproceedings{babaioff2022best,
  title     = {On best-of-both-worlds fair-share allocations},
  author    = {Babaioff, Moshe and Ezra, Tomer and Feige, Uriel},
  booktitle = {Proceedings of the 18th International Conference on Web and Internet Economics},
  pages     = {237--255},
  year      = {2022},
  paper     = {https://arxiv.org/pdf/2102.04909}
}

@inproceedings{barman2018finding,
  title     = {Finding fair and efficient allocations},
  author    = {Barman, Siddharth and Krishnamurthy, Sanath Kumar and Vaish, Rohit},
  booktitle = {Proceedings of the 19th ACM Conference on Economics and Computation},
  pages     = {557--574},
  year      = {2018},
  paper     = {https://dl.acm.org/doi/pdf/10.1145/3219166.3219176}
}

@article{barman2020approximation,
  title     = {Approximation algorithms for maximin fair division},
  author    = {Barman, Siddharth and Krishnamurthy, Sanath Kumar},
  journal   = {ACM Transactions on Economics and Computation},
  volume    = {8},
  number    = {1},
  pages     = {1--28},
  year      = {2020},
  publisher = {ACM New York, NY, USA}
}

@inproceedings{barman2023finding,
  title     = {Finding fair allocations under budget constraints},
  author    = {Barman, Siddharth and Khan, Arindam and Shyam, Sudarshan and Sreenivas, K. V. N.},
  booktitle = {Proceedings of the 37th AAAI Conference on Artificial Intelligence},
  pages     = {5481--5489},
  year      = {2023},
  paper     = {https://arxiv.org/pdf/2208.08168}
}

@inproceedings{barman2023guaranteeing,
  title     = {Guaranteeing envy-freeness under generalized assignment constraints},
  author    = {Barman, Siddharth and Khan, Arindam and Shyam, Sudarshan and Sreenivas, K. V. N.},
  booktitle = {Proceedings of the 24th ACM Conference on Economics and Computation},
  pages     = {242--269},
  year      = {2023},
  paper     = {https://dl.acm.org/doi/pdf/10.1145/3580507.3597698}
}

@misc{barman2025introspectively,
  title        = {Introspectively envy-free and efficient allocation of indivisible mixed manna},
  author       = {Barman, Siddharth and Verma, Paritosh},
  howpublished = {arXiv:2509.18673},
  year         = {2025},
  paper        = {https://arxiv.org/pdf/2509.18673}
}

@article{bei2021price,
  title     = {The price of fairness for indivisible goods},
  author    = {Bei, Xiaohui and Lu, Xinhang and Manurangsi, Pasin and Suksompong, Warut},
  journal   = {Theory of Computing Systems},
  volume    = {65},
  pages     = {1069--1093},
  year      = {2021},
  publisher = {Springer},
  paper     = {https://arxiv.org/pdf/1905.04910}
}

@inproceedings{biswas2018fair,
  title     = {Fair division under cardinality constraints},
  author    = {Arpita Biswas and Siddharth Barman},
  booktitle = {Proceedings of the 27th International Joint Conference on Artificial Intelligence},
  pages     = {91--97},
  year      = {2018},
  paper     = {https://arxiv.org/abs/1804.09521},
  memo      = {The proof of Lemma 5 in the IJCAI version (https://www.ijcai.org/proceedings/2018/0013.pdf) was wrong; the corrected lemma in the arXiv version is limited to laminar matroids.}
}

@inproceedings{biswas2019matroid,
  title     = {Matroid constrained fair allocation problem},
  author    = {Biswas, Arpita and Barman, Siddharth},
  booktitle = {Proceedings of the 33rd AAAI Conference on Artificial Intelligence},
  pages     = {9921--9922},
  year      = {2019},
  paper     = {https://cdn.aaai.org/ojs/5097/5097-13-8160-1-10-20190710.pdf}
}

@article{bouveret2016characterizing,
  title     = {Characterizing conflicts in fair division of indivisible goods using a scale of criteria},
  author    = {Bouveret, Sylvain and Lema{\^{i}}tre, Michel},
  journal   = {Autonomous Agents and Multi-Agent Systems},
  volume    = {30},
  number    = {2},
  pages     = {259--290},
  year      = {2016},
  publisher = {Springer},
  paper     = {https://link.springer.com/content/pdf/10.1007/s10458-015-9287-3.pdf}
}

@book{brams1996fair,
  title     = {Fair Division: From Cake-Cutting to Dispute Resolution},
  author    = {Brams, Steven J. and Taylor, Alan D.},
  year      = {1996},
  publisher = {Cambridge University Press}
}

@article{budish2011combinatorial,
  title     = {The combinatorial assignment problem: Approximate competitive equilibrium from equal incomes},
  author    = {Budish, Eric},
  journal   = {Journal of Political Economy},
  volume    = {119},
  number    = {6},
  pages     = {1061--1103},
  year      = {2011},
  publisher = {University of Chicago Press Chicago, IL},
  paper     = {https://www.jstor.org/stable/pdf/10.1086/664613.pdf}
}

@article{budish2017course,
  title     = {Course match: A large-scale implementation of approximate competitive equilibrium from equal incomes for combinatorial allocation},
  author    = {Budish, Eric and Cachon, G{\'e}rard P. and Kessler, Judd B. and Othman, Abraham},
  journal   = {Operations Research},
  volume    = {65},
  number    = {2},
  pages     = {314--336},
  year      = {2017},
  publisher = {INFORMS}
}

@article{caragiannis2019unreasonable,
  title     = {The unreasonable fairness of maximum {Nash} welfare},
  author    = {Caragiannis, Ioannis and Kurokawa, David and Moulin, Herv{\'e} and Procaccia, Ariel D. and Shah, Nisarg and Wang, Junxing},
  journal   = {ACM Transactions on Economics and Computation},
  volume    = {7},
  number    = {3},
  pages     = {1--32},
  year      = {2019},
  publisher = {ACM New York, NY, USA},
  paper     = {https://www.cs.toronto.edu/~nisarg/papers/mnw.ec16.pdf}
}

@article{chaudhury2021little,
  title     = {A little charity guarantees almost envy-freeness},
  author    = {Chaudhury, Bhaskar R. and Kavitha, Telikepalli and Mehlhorn, Kurt and Sgouritsa, Alkmini},
  journal   = {SIAM Journal on Computing},
  volume    = {50},
  number    = {4},
  pages     = {1336--1358},
  year      = {2021},
  publisher = {SIAM}
}

@inproceedings{cookson2025constrained,
  title     = {Constrained fair and efficient allocations},
  author    = {Cookson, Benjamin and Ebadian, Soroush and Shah, Nisarg},
  booktitle = {Proceedings of the 39th AAAI Conference on Artificial Intelligence},
  year      = {2025},
  pages     = {13718--13726},
  paper     = {https://arxiv.org/pdf/2411.00133}
}

@article{dai2023maximum,
  title     = {Maximum {Nash} social welfare under budget-feasible {EFX}},
  author    = {Dai, Sijia and Gao, Guichen and Liu, Shengxin and Lim, Boon Han and Ning, Li and Xu, Yicheng and Zhang, Yong},
  journal   = {IEEE Transactions on Network Science and Engineering},
  volume    = {11},
  number    = {2},
  pages     = {1810--1820},
  year      = {2023},
  publisher = {IEEE},
  paper     = {https://ieeexplore.ieee.org/stamp/stamp.jsp?arnumber=10315179}
}

@article{deng2024budgeted,
  title     = {The budgeted maximin share allocation problem},
  author    = {Deng, Bin and Li, Weidong},
  journal   = {Optimization Letters},
  volume    = {19},
  number    = {5},
  pages     = {955--968},
  year      = {2024},
  publisher = {Springer},
  paper     = {https://link.springer.com/content/pdf/10.1007/s11590-024-02145-6.pdf}
}

@article{dror2023fair,
  title     = {On fair division under heterogeneous matroid constraints},
  author    = {Dror, Amitay and Feldman, Michal and Segal-Halevi, Erel},
  journal   = {Journal of Artificial Intelligence Research},
  volume    = {76},
  pages     = {567--611},
  year      = {2023},
  publisher = {AI Access Foundation},
  paper     = {https://dl.acm.org/doi/pdf/10.1613/jair.1.13779}
}

@inproceedings{elkind2024fair,
  title     = {Fair division of chores with budget constraints},
  author    = {Elkind, Edith and Igarashi, Ayumi and Teh, Nicholas},
  booktitle = {Proceedings of the 17th International Symposium on Algorithmic Game Theory},
  pages     = {55--71},
  year      = {2024},
  paper     = {https://arxiv.org/pdf/2410.23979}
}

@inproceedings{feige2021tight,
  title     = {A tight negative example for {MMS} fair allocations},
  author    = {Feige, Uriel and Sapir, Ariel and Tauber, Laliv},
  booktitle = {Proceedings of the 17th International Conference on Web and Internet Economics},
  pages     = {355--372},
  year      = {2021},
  paper     = {https://arxiv.org/pdf/2104.04977}
}

@misc{feige2022improved,
  title        = {Improved maximin fair allocation of indivisible items to three agents},
  author       = {Feige, Uriel and Norkin, Alexey},
  howpublished = {arXiv:2205.05363},
  year         = {2022},
  paper        = {https://arxiv.org/pdf/2205.05363}
}

@misc{feige2022maximin,
  author       = {Uriel Feige},
  title        = {Maximin fair allocations with two item values},
  year         = {2022},
  howpublished = {\url{https://www.wisdom.weizmann.ac.il/~feige/mypapers/MMSab.pdf}},
  note         = {Last Accessed: April 8th, 2025}
}

@inproceedings{ferraioli2014regular,
  title     = {On regular and approximately fair allocations of indivisible goods},
  author    = {Ferraioli, Diodato and Gourv{\`e}s, Laurent and Monnot, J{\'e}r{\^o}me},
  booktitle = {Proceedings of the 13th International Conference on Autonomous Agents and Multi-agent Systems},
  pages     = {997--1004},
  year      = {2014},
  paper     = {https://ifaamas.org/Proceedings/aamas2014/aamas/p997.pdf}
}

@article{foley1967resource,
  title   = {Resource allocation and the public sector},
  author  = {Foley, Duncan K.},
  journal = {Yale Economic Essays},
  volume  = {7},
  pages   = {45--98},
  year    = {1967},
  memo    = {Envy-freeness (EF)}
}

@inproceedings{garbea2023efx,
  title     = {{EFx} budget-feasible allocations with high {Nash} welfare},
  author    = {Garbea, Marius and Gkatzelis, Vasilis and Tan, Xizhi},
  booktitle = {Proceedings of the 26th European Conference on Artificial Intelligence},
  year      = {2023},
  pages     = {795--802},
  paper     = {https://arxiv.org/pdf/2305.02280}
}

@article{garg2010assigning,
  title     = {Assigning papers to referees},
  author    = {Garg, Naveen and Kavitha, Telikepalli and Kumar, Amit and Mehlhorn, Kurt and Mestre, Juli{\'a}n},
  journal   = {Algorithmica},
  volume    = {58},
  number    = {1},
  pages     = {119--136},
  year      = {2010},
  publisher = {Springer-Verlag Berlin, Heidelberg},
  paper     = {https://link.springer.com/content/pdf/10.1007/s00453-009-9386-0.pdf}
}

@inproceedings{garg2019approximating,
  title     = {Approximating maximin share allocations},
  author    = {Garg, Jugal and McGlaughlin, Peter and Taki, Setareh},
  booktitle = {Proceedings of the 2nd Symposium on Simplicity in Algorithms},
  pages     = {20:1--20:11},
  year      = {2019},
  paper     = {https://drops.dagstuhl.de/storage/01oasics/oasics-vol069-sosa2019/OASIcs.SOSA.2019.20/OASIcs.SOSA.2019.20.pdf}
}

@article{garg2021improved,
  title   = {An improved approximation algorithm for maximin shares},
  author  = {Garg, Jugal and Taki, Setareh},
  journal = {Artificial Intelligence},
  volume  = {300},
  pages   = {103547},
  number  = {C},
  year    = {2021},
  paper   = {https://arxiv.org/pdf/1903.00029}
}

@inproceedings{garg2026exploring,
  title     = {Exploring relations among fairness notions in discrete fair division},
  author    = {Garg, Jugal and Sharma, Eklavya},
  booktitle = {Proceedings of the 25th International Conference on Autonomous Agents and Multi-Agent Systems},
  year      = {2026},
  paper     = {https://arxiv.org/pdf/2502.02815}
}

@article{ghodsi2021fair,
  title     = {Fair allocation of indivisible goods: Improvement},
  author    = {Ghodsi, Mohammad and Hajiaghayi, Mohammad T. and Seddighin, Masoud and Seddighin, Saeed and Yami, Hadi},
  journal   = {Mathematics of Operations Research},
  volume    = {46},
  number    = {3},
  pages     = {1038--1053},
  year      = {2021},
  publisher = {INFORMS},
  paper     = {https://pubsonline.informs.org/doi/epdf/10.1287/moor.2020.1096},
  memo      = {Supersedes ghodsi2018fair, except 4/5-MMS for four agents}
}

@article{gourves2019maximin,
  title     = {On maximin share allocations in matroids},
  author    = {Gourv{\`e}s, Laurent and Monnot, J{\'e}r{\^o}me},
  journal   = {Theoretical Computer Science},
  volume    = {754},
  pages     = {50--64},
  year      = {2019},
  publisher = {Elsevier},
  paper     = {https://www.sciencedirect.com/science/article/pii/S0304397518303384/pdfft?md5=e5ffec82cc73b749959d3d60c31f2b94&pid=1-s2.0-S0304397518303384-main.pdf},
  memo      = {Restricting the union of bundles}
}

@inproceedings{heidari2026improved,
  title     = {Improved maximin share guarantee for additive valuations},
  author    = {Heidari, Ehsan and Kaviani, Alireza and Seddighin, Masoud and Shahrezaei, AmirMohammad},
  booktitle = {Proceedings of the 37th Annual ACM-SIAM Symposium on Discrete Algorithms},
  pages     = {2239--2290},
  year      = {2026},
  paper     = {https://arxiv.org/pdf/2510.10423}
}

@inproceedings{huang2021algorithmic,
  title     = {An algorithmic framework for approximating maximin share allocation of chores},
  author    = {Huang, Xin and Lu, Pinyan},
  booktitle = {Proceedings of the 22nd ACM Conference on Economics and Computation},
  pages     = {630--631},
  year      = {2021},
  memo      = {11/9-MMS},
  paper     = {https://arxiv.org/pdf/1907.04505}
}

@inproceedings{huang2023reduction,
  title     = {A reduction from chores allocation to job scheduling},
  author    = {Huang, Xin and Segal-Halevi, Erel},
  booktitle = {Proceedings of the 24th ACM Conference on Economics and Computation},
  pages     = {908--908},
  year      = {2023},
  memo      = {13/11-MMS},
  paper     = {https://arxiv.org/pdf/2302.04581}
}

@misc{huang2025fptas,
  title        = {An {FPTAS} for $7/9$-approximation to maximin share allocations},
  author       = {Huang, Xin and Zhou, Shengwei},
  howpublished = {arXiv:2511.13056},
  year         = {2025},
  paper        = {https://arxiv.org/pdf/2511.13056}
}

@inproceedings{hummel2022maximin,
  title     = {Maximin shares under cardinality constraints},
  author    = {Hummel, Halvard and Hetland, Magnus L.},
  booktitle = {Proceedings of the 19th European Conference on Multi-Agent Systems},
  pages     = {188--206},
  year      = {2022},
  note      = {A full version is available at \url{https://arxiv.org/pdf/2106.07300}.},
  memo      = {2/3-cardinality-MMS},
  cited_by  = {https://scholar.google.com/scholar?cites=12421044056997004364&as_sdt=2005&sciodt=0,5&hl=en},
  paper     = {https://arxiv.org/pdf/2106.07300}
}

@article{hummel2025maximin,
  title   = {Maximin shares in hereditary set systems},
  author  = {Hummel, Halvard},
  journal = {ACM Transactions on Economics and Computation},
  volume  = {13},
  number  = {3},
  pages   = {1-33},
  year    = {2025},
  paper   = {https://arxiv.org/pdf/2404.11582}
}

@misc{igarashi2025fair,
  title        = {Fair and efficient allocation of indivisible items under category constraints},
  author       = {Igarashi, Ayumi and Meunier, Fr{\'e}d{\'e}ric},
  howpublished = {arXiv:2503.20260},
  paper        = {https://arxiv.org/pdf/2503.20260},
  year         = {2025}
}

@article{kurokawa2018fair,
  title     = {Fair enough: Guaranteeing approximate maximin shares},
  author    = {Kurokawa, David and Procaccia, Ariel D. and Wang, Junxing},
  journal   = {Journal of the ACM},
  volume    = {65},
  number    = {2},
  pages     = {1--27},
  year      = {2018},
  publisher = {ACM New York, NY, USA},
  paper     = {https://dl.acm.org/doi/pdf/10.1145/3140756}
}

@inproceedings{lam2025exact,
  title     = {The (exact) price of cardinality for indivisible goods: A parametric perspective},
  author    = {Lam, Alexander and Li, Bo and Sun, Ankang},
  booktitle = {Proceedings of the 39th AAAI Conference on Artificial Intelligence},
  year      = {2025},
  pages     = {13985--13992},
  paper     = {https://arxiv.org/pdf/2501.01660}
}

@article{li2021fair-division,
  title     = {The fair division of hereditary set systems},
  author    = {Li, Zhentao and Vetta, Adrian},
  journal   = {ACM Transactions on Economics and Computation},
  volume    = {9},
  number    = {2},
  pages     = {1--19},
  year      = {2021},
  publisher = {ACM New York, NY, USA},
  paper     = {https://dl.acm.org/doi/pdf/10.1145/3434410}
}

@inproceedings{li2021fair-scheduling,
  title     = {Fair scheduling for time-dependent resources},
  author    = {Li, Bo and Li, Minming and Zhang, Ruilong},
  booktitle = {Proceedings of the 35th International Conference on Neural Information Processing Systems},
  pages     = {21744--21756},
  year      = {2021},
  paper     = {https://proceedings.neurips.cc/paper_files/paper/2021/file/b5b1d9ada94bb80609d21eecf7a2ce7a-Paper.pdf},
  review    = {https://openreview.net/forum?id=QMJb9BvLqXU}
}

@inproceedings{li2023fair,
  title     = {Fair allocation of indivisible chores: Beyond additive costs},
  author    = {Li, Bo and Wang, Fangxiao and Zhou, Yu},
  booktitle = {Proceedings of the 37th International Conference on Neural Information Processing Systems},
  pages     = {54366--54385},
  year      = {2023},
  paper     = {https://arxiv.org/pdf/2205.10520},
  review    = {https://openreview.net/forum?id=uJmsYZiu3E}
}

@inproceedings{mackin2016allocating,
  title     = {Allocating indivisible items in categorized domains},
  author    = {Mackin, Erika and Xia, Lirong},
  booktitle = {Proceedings of the 25th International Joint Conference on Artificial Intelligence},
  pages     = {359--365},
  year      = {2016},
  paper     = {https://www.ijcai.org/Proceedings/16/Papers/058.pdf}
}

@inproceedings{mahara2026existence,
  title     = {Existence of fair and efficient allocation of indivisible chores},
  author    = {Mahara, Ryoga},
  booktitle = {Proceedings of the 37th Annual ACM-SIAM Symposium on Discrete Algorithms},
  pages     = {6742--6766},
  year      = {2026},
  paper     = {https://arxiv.org/pdf/2507.09544}
}

@book{moulin2004fair,
  title     = {Fair Division and Collective Welfare},
  author    = {Moulin, Herv{\'e}},
  year      = {2004},
  publisher = {MIT Press}
}

@article{moulin2019fair,
  title     = {Fair division in the internet age},
  author    = {Moulin, Herv{\'e}},
  journal   = {Annual Review of Economics},
  volume    = {11},
  number    = {1},
  pages     = {407--441},
  year      = {2019},
  publisher = {Annual Reviews},
  paper     = {https://www.annualreviews.org/docserver/fulltext/economics/11/1/annurev-economics-080218-025559.pdf}
}

@inproceedings{procaccia2014fair,
  title     = {Fair enough: Guaranteeing approximate maximin shares},
  author    = {Procaccia, Ariel D. and Wang, Junxing},
  booktitle = {Proceedings of the 15th ACM Conference on Economics and Computation},
  pages     = {675--692},
  year      = {2014},
  memo      = {2/3-MMS poly (only when the number of agents is constant.)},
  paper     = {https://dl.acm.org/doi/pdf/10.1145/2600057.2602835}
}

@inproceedings{shoshan2023efficient,
  title     = {Efficient nearly-fair division with capacity constraints},
  author    = {Shoshan, Hila and Hazon, Noam and Segal-Halevi, Erel},
  booktitle = {Proceedings of the 22nd International Conference on Autonomous Agents and Multiagent Systems},
  pages     = {206--214},
  year      = {2023},
  paper     = {https://arxiv.org/pdf/2205.07779}
}

@article{steinhaus1948problem,
  title   = {The problem of fair division},
  author  = {Steinhaus, Hugo},
  journal = {Econometrica},
  volume  = {16},
  pages   = {101--104},
  year    = {1948},
  memo    = {Proportionality (PROP)}
}

@article{suksompong2021constraints,
  title     = {Constraints in fair division},
  author    = {Suksompong, Warut},
  journal   = {ACM SIGecom Exchanges},
  volume    = {19},
  number    = {2},
  pages     = {46--61},
  year      = {2021},
  publisher = {ACM New York, NY, USA},
  paper     = {https://dl.acm.org/doi/pdf/10.1145/3505156.3505162}
}

@inproceedings{wang2024fairness,
  title     = {The fairness of maximum {Nash} social welfare under matroid constraints and beyond},
  author    = {Wang, Yuanyuan and Chen, Xin and Nong, Qingqin},
  booktitle = {Proceedings of the 20th International Conference on Web and Internet Economics},
  pages     = {172--189},
  year      = {2024},
  paper     = {https://arxiv.org/pdf/2411.01462}
}

@article{wang2025guaranteeing,
  title     = {Guaranteeing fairness and efficiency under budget constraints},
  author    = {Wang, Yuanyuan and Chen, Xin and Fang, Qizhi and Nong, Qingqin and Liu, Wenjing},
  journal   = {Journal of Combinatorial Optimization},
  volume    = {49},
  number    = {3},
  pages     = {1--21},
  year      = {2025},
  publisher = {Springer US New York},
  paper     = {https://link.springer.com/content/pdf/10.1007/s10878-025-01275-6.pdf}
}

@article{wu2025approximate,
  title     = {Approximate envy-freeness in indivisible resource allocation with budget constraints},
  author    = {Wu, Xiaowei and Li, Bo and Gan, Jiarui},
  journal   = {Information and Computation},
  volume    = {303},
  pages     = {105264},
  year      = {2025},
  publisher = {Elsevier},
  paper     = {https://www.sciencedirect.com/science/article/pii/S0890540124001299/pdfft?md5=670546b6b80f5e559b79dfe4c9074274&pid=1-s2.0-S0890540124001299-main.pdf},
  memo      = {Supersedes wu2021budget and gan2023approximation}
}

@article{hochbaum1987using,
  title     = {Using dual approximation algorithms for scheduling problems: Theoretical and practical results},
  author    = {Hochbaum, Dorit S. and Shmoys, David B.},
  journal   = {Journal of the ACM},
  volume    = {34},
  number    = {1},
  pages     = {144--162},
  year      = {1987},
  publisher = {ACM New York, NY, USA},
  cited_by  = {https://scholar.google.com/scholar?cites=14812103578878569073&as_sdt=2005&sciodt=0,5&hl=en&oi=gsb},
  paper     = {https://dl.acm.org/doi/pdf/10.1145/7531.7535}
}

@article{woeginger1997polynomial,
  title     = {A polynomial-time approximation scheme for maximizing the minimum machine completion time},
  author    = {Woeginger, Gerhard J.},
  journal   = {Operations Research Letters},
  volume    = {20},
  number    = {4},
  pages     = {149--154},
  year      = {1997},
  publisher = {Elsevier},
  paper     = {https://www-sciencedirect-com.utokyo.idm.oclc.org/science/article/pii/S0167637796000557/pdf?md5=ba4b0cf26eab024c2e8ac35f3e845017&pid=1-s2.0-S0167637796000557-main.pdf}
}

@article{woeginger2000does,
  title     = {When does a dynamic programming formulation guarantee the existence of a fully polynomial time approximation scheme ({FPTAS})?},
  author    = {Woeginger, Gerhard J.},
  journal   = {INFORMS Journal on Computing},
  volume    = {12},
  number    = {1},
  pages     = {57--74},
  year      = {2000},
  publisher = {INFORMS},
  paper     = {https://citeseerx.ist.psu.edu/document?repid=rep1&type=pdf&doi=5858a4646ba18372be4b3c63c31fdd55f7807e5d}
}

@article{arulselvan2018matchings,
  title     = {Matchings with lower quotas: Algorithms and complexity},
  author    = {Arulselvan, Ashwin and Cseh, {\'A}gnes and Gro$\beta$, Martin and Manlove, David F. and Matuschke, Jannik},
  journal   = {Algorithmica},
  volume    = {80},
  number    = {1},
  pages     = {185--208},
  year      = {2018},
  publisher = {Springer-Verlag Berlin, Heidelberg},
  paper     = {https://link.springer.com/content/pdf/10.1007/s00453-016-0252-6.pdf}
}

@inproceedings{aziz2022matching,
  title     = {Matching market design with constraints},
  author    = {Aziz, Haris and Bir{\'o}, P{\'e}ter and Yokoo, Makoto},
  booktitle = {Proceedings of the 36th AAAI Conference on Artificial Intelligence},
  pages     = {12308--12316},
  year      = {2022},
  paper     = {https://cdn.aaai.org/ojs/21495/21495-13-25508-1-2-20220628.pdf}
}

@article{biro2010college,
  title     = {The college admissions problem with lower and common quotas},
  author    = {Bir{\'o}, P{\'e}ter and Fleiner, Tam{\'a}s and Irving, Robert W. and Manlove, David F.},
  journal   = {Theoretical Computer Science},
  volume    = {411},
  number    = {34-36},
  pages     = {3136--3153},
  year      = {2010},
  publisher = {Elsevier Science Publishers Ltd. Essex, UK},
  paper     = {https://www.sciencedirect.com/science/article/pii/S0304397510002860/pdf?md5=e5debce03bea7e71cae64a53eba9443d&pid=1-s2.0-S0304397510002860-main.pdf}
}

@article{gale1962college,
  title     = {College admissions and the stability of marriage},
  author    = {Gale, David and Shapley, Lloyd S.},
  journal   = {The American Mathematical Monthly},
  volume    = {69},
  number    = {1},
  pages     = {9--15},
  year      = {1962},
  publisher = {Taylor \& Francis}
}

@book{gusfield1989stable,
  title     = {The Stable Marriage Problem: Structure and Algorithms},
  author    = {Gusfield, Dan and Irving, Robert W.},
  year      = {1989},
  publisher = {MIT Press}
}

@article{hamada2016hospitals,
  title     = {The hospitals/residents problem with lower quotas},
  author    = {Hamada, Koki and Iwama, Kazuo and Miyazaki, Shuichi},
  journal   = {Algorithmica},
  volume    = {74},
  pages     = {440--465},
  year      = {2016},
  publisher = {Springer},
  paper     = {https://link.springer.com/content/pdf/10.1007/s00453-014-9951-z.pdf}
}

@inproceedings{huang2010classified,
  title     = {Classified stable matching},
  author    = {Huang, Chien-Chung},
  booktitle = {Proceedings of the 21st Annual ACM-SIAM Symposium on Discrete Algorithms},
  pages     = {1235--1253},
  year      = {2010},
  paper     = {https://epubs.siam.org/doi/epdf/10.1137/1.9781611973075.99}
}

@article{kamada2017stability,
  title     = {Stability concepts in matching under distributional constraints},
  author    = {Kamada, Yuichiro and Kojima, Fuhito},
  journal   = {Journal of Economic Theory},
  volume    = {168},
  pages     = {107--142},
  year      = {2017},
  publisher = {Elsevier},
  paper     = {https://www.sciencedirect.com/science/article/pii/S0022053116301156/pdfft?casa_token=t9oHN6CnjzMAAAAA:pHreV_QHdaObX1hc7ETTrxPj_LubQiSy7iPkj1sibvMWAZz-AHr_cpigLGIMD52mfJP10RCEEbhB&md5=a4ae13ff49a8e71b695dd7154ff1d3d5&pid=1-s2.0-S0022053116301156-main.pdf}
}

@inproceedings{nasre2017popular,
  title     = {Popular matchings with lower quotas},
  author    = {Nasre, Meghana and Nimbhorkar, Prajakta},
  booktitle = {Proceedings of the 37th IARCS Annual Conference on Foundations of Software Technology and Theoretical Computer Science},
  year      = {2017},
  pages     = {44:1--44:15},
  paper     = {https://d-nb.info/1153079542/34#page=611}
}

@article{yokoi2020envy,
  title     = {Envy-free matchings with lower quotas},
  author    = {Yokoi, Yu},
  journal   = {Algorithmica},
  volume    = {82},
  number    = {2},
  pages     = {188--211},
  year      = {2020},
  publisher = {Springer US New York},
  paper     = {https://arxiv.org/pdf/1704.04888}
}

@inproceedings{aleksandrov2015online,
  title     = {Online fair division: Analysing a food bank problem},
  author    = {Aleksandrov, Martin and Aziz, Haris and Gaspers, Serge and Walsh, Toby},
  booktitle = {Proceedings of the 24th International Joint Conference on Artificial Intelligence},
  pages     = {2540--2546},
  year      = {2015},
  paper     = {https://arxiv.org/pdf/1502.07571}
}

@inproceedings{aleksandrov2020online,
  title     = {Online fair division: A survey},
  author    = {Aleksandrov, Martin and Walsh, Toby},
  booktitle = {Proceedings of the 34th AAAI Conference on Artificial Intelligence},
  pages     = {13557--13562},
  year      = {2020},
  paper     = {https://cdn.aaai.org/ojs/7081/7081-13-10310-1-10-20200526.pdf}
}

@inproceedings{kulkarni2025online,
  title     = {Online fair division: Towards ex-post constant {MMS} guarantees},
  author    = {Kulkarni, Pooja and Mehta, Ruta and Shahkar, Parnian},
  booktitle = {Proceedings of the 26th ACM Conference on Economics and Computation},
  pages     = {638--638},
  year      = {2025},
  paper     = {https://arxiv.org/pdf/2503.02088}
}

@inproceedings{mertzanidis2024automating,
  title     = {Automating food drop: The power of two choices for dynamic and fair food allocation},
  author    = {Mertzanidis, Marios and Psomas, Alexandros and Verma, Paritosh},
  booktitle = {Proceedings of the 25th ACM Conference on Economics and Computation},
  pages     = {243--243},
  year      = {2024},
  paper     = {https://arxiv.org/pdf/2406.06363}
}

@inproceedings{procaccia2024honor,
  title     = {Honor among bandits: No-regret learning for online fair division},
  author    = {Procaccia, Ariel D. and Schiffer, Ben and Zhang, Shirley},
  year      = {2024},
  booktitle = {Proceedings of the 38th International Conference on Neural Information Processing Systems},
  pages     = {13183-13227},
  paper     = {https://arxiv.org/pdf/2407.01795},
  review    = {https://openreview.net/forum?id=OCQbC0eDJJ}
}

@inproceedings{schiffer2025improved,
  title     = {Improved regret bounds for online fair division with bandit learning},
  author    = {Schiffer, Benjamin and Zhang, Shirley},
  booktitle = {Proceedings of the 39th AAAI Conference on Artificial Intelligence},
  pages     = {14079-14086},
  year      = {2025},
  paper     = {https://arxiv.org/pdf/2501.07022}
}

@inproceedings{zhou2023multi,
  title     = {Multi-agent online scheduling: {MMS} allocations for indivisible items},
  author    = {Zhou, Shengwei and Bai, Rufan and Wu, Xiaowei},
  booktitle = {Proceedings of the 40th International Conference on Machine Learning},
  pages     = {42506--42516},
  year      = {2023},
  paper     = {https://proceedings.mlr.press/v202/zhou23a/zhou23a.pdf}
}

@book{korte2018combinatorial,
  title     = {Combinatorial Optimization: Theory and Algorithms},
  author    = {Korte, Bernhard and Vygen, Jens},
  year      = {2018},
  publisher = {Springer},
  edition   = {6},
  paper     = {https://link.springer.com/content/pdf/10.1007/978-3-662-56039-6.pdf}
}


\clearpage
\appendix

\section{A $\left(\frac{3n-1}{2n}\right)$-MMS allocation algorithm for single-category chores} \label{sec:single:chores}

\begin{algorithm}[tb]
  \centering

  \begin{algorithmic}[1]
    \Function{ApproxChores}{$\cI = \left(N,M,(v_i)_{i\in N},(q^-,q^+)\right),\alpha$} \label{alg:single:chores:main-start}

    \If{$|M| \le |N|$} \label{alg:single:chores:check-few-items}
    \State $A_i\gets
    \begin{cases}
      \{g_{i}\} & \text{if } i \le |M|, \\
      \emptyset & \text{otherwise}
    \end{cases}$ for each $i \in N = \{1,\ldots,|N|\}$. \label{alg:single:chores:bundle-few-items}
    \State \Return $(A_i)_{i\in N}$ \label{alg:single:chores:allocation-few-items}
    \EndIf \label{alg:single:chores:end-check-few-items}


    \For{$k \gets |N|,|N|-1,\ldots,1$} \label{alg:single:chores:B_k-init-for-start}
    \State $b_k \gets \max\left\{1, q^-, |M| - \sum_{k'=k+1}^{|N|}b_{k'} - q^+(k - 1)\right\}$. \label{alg:single:chores:b_k-init-n}
    \State $\begin{aligned}
      B^{(|N|)}_k \gets & \left\{g_{|M| - |N| + k}\right\} \cup \\
      & \left\{g_{j} \mid \sum_{k'=k+1}^{|N|}(b_{k'} - 1) < j \le \sum_{k'=k}^{|N|}(b_{k'} - 1) \right\}.
    \end{aligned}$ \label{alg:single:chores:B_k-init-n}
    \EndFor \label{alg:single:chores:B_k-init-for-end}

    \State $\begin{aligned}
      \hat\mu_i \gets &\min \left(\left\{2\, v_i\left(g_{|M| - |N|}\right)\right\} \right. \\
      & \left. \cup \left\{ \frac{1}{|N|-r+1}\, v_i\left(\bigcup_{k = r}^{|N|} B^{(|N|)}_k\right) \mid 1\le r\le |N| \right\} \right)
    \end{aligned}$ \\ for each $i\in N$. \label{alg:single:chores:mu-hat}

    \State $N^{(|N|)}\gets N$. \label{alg:single:chores:var-init}

    \For{$t \gets |N|,|N|-1,\ldots,1$} \label{alg:single:chores:outer-for-start}
    \State $(B^{(t-1)}_1, B^{(t-1)}_2, \ldots, B^{(t-1)}_{t-1}, B) \gets (B^{(t)}_1, B^{(t)}_2, \ldots, B^{(t)}_{t})$. \label{alg:single:chores:bag-init}

    \For{$k \gets t - 1,t - 2,\ldots,1$} \label{alg:single:chores:inner-for-start}

    \While{$\left(\exists i\in N^{(t)}~\mbox{s.t.}~v_i(B) \ge \left(\alpha - \frac{1}{2}\right)\hat\mu_i \right)$ and $B\setminus\left\{g_{|M| - |N| + t}\right\} \neq B^{(t)}_k\setminus\left\{g_{|M| - |N| + k}\right\}$} \label{alg:single:chores:inner-while-start}
    \State $g\gets$ the least valuable item in $B^{(t-1)}_k\setminus\left\{g_{|M|-|N| + k}\right\}$. \label{alg:single:chores:get-from-B_k}

    \If{$|B| < |B^{(t)}_k|$}
    \State $B\gets B\cup\{g\}$, $B^{(t-1)}_k\gets B^{(t-1)}_k\setminus\{g\}$. \label{alg:single:chores:move}
    \Else
    \State $h\gets$ the most valuable item in $B\setminus\left\{g_{|M|-|N|+t}\right\}$. \label{alg:single:chores:get-from-B}
    \State $B\gets (B\setminus\{h\})\cup\{g\}$, \\
    $B^{(t-1)}_k\gets \left(B^{(t-1)}_k\setminus\{g\}\right)\cup\{h\}$. \label{alg:single:chores:swap}
    \EndIf
    \EndWhile \label{alg:single:chores:inner-while-end}
    \EndFor \label{alg:single:chores:inner-for-end}

    \State Find $i^{(t)}\in N^{(t)}$ s.t. $v_{i^{(t)}}(B) \ge \alpha\,\hat\mu_{i^{(t)}}$.\label{alg:single:chores:before-assign}
    \State $A_{i^{(t)}}\gets B$, $N^{(t - 1)}\gets N^{(t)}\setminus \left\{i^{(t)}\right\}$. \label{alg:single:chores:assign}

    \EndFor \label{alg:single:chores:outer-for-end}

    \State \Return $(A_i)_{i\in N}$ \label{alg:single:chores:allocation-final}
    \EndFunction

  \end{algorithmic}

  \caption{Compute a $\left(\frac{2}{3-(\max\{1,n\})^{-1}}\right)$-MMS allocation for a single-category ordered instance of chores with $n$ agents.}
  \label{alg:single:chores}

\end{algorithm}

In \cref{alg:single:chores}, we present a polynomial-time algorithm, \hyperref[alg:single:chores:main-start]{\textproc{ApproxChores}}, which computes a $\left(\frac{3n - 1}{2n}\right)$-MMS allocation for a single-category ordered instance of chores with $n$ agents and thus establishes \cref{thm:single:chores:three-halves} together with the reduction to ordered instances.
In contrast to the case of goods, the unassigned bags are maintained in decreasing order of size, with special treatment for the least valuable $n$ items.
Unlike the algorithm for goods discussed in \cref{sec:single:goods}, \hyperref[alg:single:chores:main-start]{\textproc{ApproxChores}} is not recursive due to the absence of valid reduction.
Here, \cref{thm:common:bound-chores} implicitly serves as a useful proxy of valid reduction, similarly exploiting the pigeonhole principle.

\begin{restatable}{lemma}{LemValidReductionChore}
  \label{thm:common:bound-chores}
  Let $\cI = \left(N,M,(v_i)_{i\in N},\cC,(q_{C}^-,q_{C}^+)_{C\in\cC}\right)$ be an ordered instance of chores, $C^*\in\cC$ be non-empty, and $d\in\bZ_{\ge 0}$ satisfy $|C^*| \ge d|N| + 1$.
  Then
  \begin{align}
    B &\coloneqq \left\{g^{C^*}_{|C^*| - j} \mid d(|N| - 1) \le j \le d|N| \right\} \label{eq:common:chores:valid-bag}
  \end{align}
  satisfies $v_i(B) \ge \mu_i(\cI)$ for every $i\in N$.
\end{restatable}

\begin{proof}
  Fix an arbitrary $i\in N$.
  Let $(P_k)_{k=1}^{|N|}$ be agent $i$'s MMS partition, which we assume satisfies
  \begin{align}
    \left|P_1 \cap \left\{g^{C^*}_{|C^*|-j} \mid 0\le j\le d|N| \right\}\right| \ge d + 1 \label{eq:common:chores:pigeon}
  \end{align}
  without loss of generality, by the pigeonhole principle and that $|C^*| \ge d|N| + 1$.
  Now \cref{eq:common:chores:valid-bag,eq:common:chores:pigeon} imply
  \begin{align*}
    \mu_i(\cI) &\le v_i(P_1) \le v_i\left(P_1 \cap \left\{g^{C^*}_{|C^*|-j} \mid 0\le j\le d|N|\right\}\right) \le v_i(B). \qedhere
  \end{align*}
\end{proof}

\cref{sec:single:chores:proof} is dedicated to the proof of the following theorem, restated below from \cref{sec:main_results}.
Then \cref{sec:single:chores:tight} demonstrates that our analysis is tight.

\SingleChoresMain*

\subsection{Proof of \cref{thm:single:chores:three-halves}} \label{sec:single:chores:proof}

Here, we prove that \hyperref[alg:single:chores:main-start]{\textproc{ApproxChores}} returns an $\alpha$-MMS allocation for an arbitrary single-category ordered instance of chores $\cI = \left(N, M, (v_i)_{i\in N}, (q^-,q^+)\right)$ and an arbitrary real constant $\alpha \in \left[\frac{3 - (\max\left\{|N|, 1\right\})^{-1}}{2}, \frac{3}{2}\right]$, in time $O(|N||M|)$.
Notice that all the valuations $(v_i)_{i\in N}$ are non-positive here.
The proof structure is similar to the case of goods, where we eventually derive \cref{thm:single:chores:three-halves} by combining \cref{thm:common:ordered,thm:single:chores:few-items,thm:single:chores:invariant-all,thm:single:chores:runtime}; refer to \cref{thm:common:ordered} for the time complexity of the reduction to an ordered instance, which turns out to be dominant.
All line numbers in this section refer to \cref{alg:single:chores}.

\begin{lemma} \label{thm:single:chores:few-items}
  If Line~\ref{alg:single:chores:bundle-few-items} is reached, an MMS allocation $(A_i)_{i\in N}$ for $\cI$ is returned at Line~\ref{alg:single:chores:allocation-few-items}.
\end{lemma}

\begin{proof}
  The allocation $(A_i)_{i\in N}$ at Line~\ref{alg:single:chores:allocation-few-items} is feasible for $\cI$ due to \cref{eq:def:quotas}.
  Because $\mu_i(\cI) = v_i(g_{|M|})$ for each $i\in N$, it also holds for each $i\in N$ that $v_i(A_i) \ge \mu_i(\cI)$.
  Therefore, $(A_i)_{i\in N}$ is an MMS allocation for $\cI$.
\end{proof}

\begin{lemma} \label{thm:single:chores:value-sum-MMS}
  The bundles $B^{(|N|)}_1,B^{(|N|)}_2,\ldots,B^{(|N|)}_{|N|}$ defined by Lines \ref{alg:single:chores:B_k-init-for-start} to \ref{alg:single:chores:B_k-init-for-end} are mutually disjoint and satisfy the following
  \begin{align}
    &B^{(|N|)}_1 \cup B^{(|N|)}_2 \cup \cdots \cup B^{(|N|)}_{|N|} = M, \label{eq:single:chores:union-init} \\
    &q^+ \ge |B^{(|N|)}_1| \ge |B^{(|N|)}_2| \ge \cdots \ge |B^{(|N|)}_{|N|}| \ge q^-, \label{eq:single:chores:size-init} \\
    \begin{split}
    &v_i\left(B^{(|N|)}_r \cup B^{(|N|)}_{r+1} \cup \cdots \cup B^{(|N|)}_{|N|}\right) \ge \left(|N| - r + 1\right)\mu_i(\cI) \\
    &\quad\forall r\in\{1,2,\ldots,|N|\},\forall i\in N.
    \end{split} \label{eq:single:chores:value-sum-MMS}
  \end{align}
\end{lemma}


\begin{proof}
  Given \cref{eq:def:quotas} and $|M| > |N|$, these bundles are well-defined by Lines~\ref{alg:single:chores:B_k-init-for-start} to \ref{alg:single:chores:B_k-init-for-end}, mutually disjoint, and satisfying \cref{eq:single:chores:union-init,eq:single:chores:size-init}.
  We fix any pair $i\in N$ and $r\in\{1,2,\ldots,|N|\}$, for which the inequality in \cref{eq:single:chores:value-sum-MMS} is shown below.
  Let $(P_k)_{k=1}^{|N|}$ be agent $i$'s MMS partition, which can be assumed, without loss of generality, to satisfy that
  \begin{align*}
    P_r\cup P_{r+1}\cup\cdots\cup P_{|N|} &\supseteq \left\{g_{|M|},g_{|M| - 1},\ldots,g_{|M| - |N| + r}\right\} \\
    &\forall r\in\{1,2,\ldots,|N|\}.
  \end{align*}
  Combining this with $(P_k)_{k=1}^{|N|}\in\cF(\cI)$ and the construction of $B^{(|N|)}_1,B^{(|N|)}_2,\ldots,B^{(|N|)}_{|N|}$ yields that
  \begin{align*}
    v_i\left(B^{(|N|)}_r \cup B^{(|N|)}_{r+1} \cup \cdots \cup B^{(|N|)}_{|N|}\right)
    &\ge v_i\left(P_r \cup P_{r+1} \cup \cdots \cup P_{|N|}\right) \\
    &\ge \left(|N| - r + 1\right) \mu_i(\cI),
  \end{align*}
  which concludes the proof.
\end{proof}

\begin{corollary} \label{thm:single:chores:normalize}
  At Line~\ref{alg:single:chores:var-init}, the following hold:
  \begin{align}
    2\,v_i(g) &\ge \hat\mu_i & & \mspace{-12mu} \forall g\in M\setminus\left\{g_{|M|},g_{|M|-1},\ldots,g_{|M|-|N|+1}\right\},\forall i\in N. \label{eq:single:chores:value-half} \\
    \mu_i(\cI) &\le \hat\mu_i & & \mspace{-12mu} \forall i\in N. \label{eq:single:chores:value-mu-hat}
  \end{align}
\end{corollary}


\begin{proof}
  As the instance $\cI$ is ordered, \cref{eq:single:chores:value-half} follows from Line~\ref{alg:single:chores:mu-hat}.
  Due to \cref{thm:common:bound-chores} for $C^* = M$ and $d = 1$, it also holds that
  \begin{align}
    \begin{split}
      2\,v_i\left(g_{|M|-|N|}\right) &\ge v_i\left(\left\{g_{|M|-|N|},g_{|M|-|N|+1}\right\}\right) \\
      & \ge \mu_i(\cI)
    \end{split} &\forall i\in N. \label{eq:single:chores:normalize-(n+1)-MMS}
  \end{align}
  Line~\ref{alg:single:chores:mu-hat}, \cref{eq:single:chores:normalize-(n+1)-MMS}, and \cref{eq:single:chores:value-sum-MMS} of \cref{thm:single:chores:value-sum-MMS} together imply \cref{eq:single:chores:value-mu-hat}.
\end{proof}


\begin{definition} \label{def:single:chores:invariants}
  Along with the outer for-loop over Lines \ref{alg:single:chores:outer-for-start} to \ref{alg:single:chores:outer-for-end}, we consider the following conditions for each $t\in\{|N|, |N| - 1, \ldots, 1, 0\}$:
  \leqnomode
  \begin{align}
    \tag{C1} \mspace{0mu}& \mathrlap{B^{(t)}_1,B^{(t)}_2,\ldots,B^{(t)}_t, A_{i^{(t + 1)}}, A_{i^{(t + 2)}}, \ldots, A_{i^{(|N|)}} \mbox{ are disjoint.}} \label{cond:single:chores:disjoint} \\[0pt]
    \tag{C2} \mspace{0mu}& \mathrlap{B^{(t)}_1\cup B^{(t)}_2\cup\cdots\cup B^{(t)}_t\cup A_{i^{(t + 1)}}\cup A_{i^{(t + 2)}}\cup\cdots\cup A_{i^{(|N|)}} = M.} \label{cond:single:chores:union} \\[0pt]
    \tag{C3} \mspace{0mu}& \mathrlap{q^+ \ge |B^{(t)}_1| \ge |B^{(t)}_2| \ge \cdots \ge |B^{(t)}_t| \ge q^-.} \label{cond:single:chores:size} \\[0pt]
    \tag{C4}
    \mspace{0mu}&
    \begin{aligned}
      & B^{(t)}_k \cap \left\{g_{|M|},g_{|M|-1},\ldots,g_{|M|-|N|+1}\right\} = \left\{g_{|M|-|N|+k}\right\} \\
      & \forall k\in\{1,2,\ldots,t\}.
    \end{aligned}
    \label{cond:single:chores:top-n} \\[0pt]
    \tag{C5}
    \mspace{0mu}&
    \mathrlap{
      \begin{aligned}
        & v_i(h_1) \le v_i(h_2) \le \cdots \le v_i(h_t) & \forall i\in N^{(t)},\\[-2pt]
        & \mathrlap{\forall h_1\in B^{(t)}_1\setminus\{g_{|M|-|N|+1}\}, \ldots, \forall h_t\in B^{(t)}_t\setminus\{g_{|M|-|N|+t}\}.}
    \end{aligned}}
    \label{cond:single:chores:value-order} \\[0pt]
    \tag{C6}
    \mspace{0mu} &
    \begin{aligned}
    & v_i\left(\bigcup_{k=r}^{t}B^{(t)}_k\right) \ge \left(t - r + 1 + (|N| - t)\left(\frac{3}{2} - \alpha\right)\right) \hat\mu_i \\[-2pt]
    & \forall r\in\{1,2,\ldots,t\},\forall i\in N^{(t)}.
    \end{aligned}
    \label{cond:single:chores:value-sum}
  \end{align}
  \reqnomode
\end{definition}

\begin{lemma} \label{thm:single:chores:invariant-init}
  At Line~\ref{alg:single:chores:var-init}, Conditions \eqref{cond:single:chores:disjoint} to \eqref{cond:single:chores:value-sum} hold for $t = |N|$.
\end{lemma}

\begin{proof}
  We obtain Conditions \eqref{cond:single:chores:disjoint} to \eqref{cond:single:chores:value-order} for $t = |N|$ from \cref{thm:single:chores:value-sum-MMS},
  and Condition \eqref{cond:single:chores:value-sum} for $t = |N|$ from Line~\ref{alg:single:chores:mu-hat}.
\end{proof}

\begin{lemma} \label{thm:single:chores:invariant}
  Let $s\in\{|N|, |N| - 1 \ldots, 1\}$ be arbitrary.
  Suppose that the for-loop over Lines \ref{alg:single:chores:outer-for-start} to \ref{alg:single:chores:outer-for-end} has successfully iterated for $t \in \{|N|,|N|-1,\ldots,s+1\}$, and that Conditions \eqref{cond:single:chores:disjoint} to \eqref{cond:single:chores:value-sum} now hold for $t = s$.
  Then the next iteration with $t = s$ succeeds, where the following hold at Line~\ref{alg:single:chores:before-assign}:
  \begin{align}
    & q^- \le |B| \le q^+. \label{eq:single:chores:bag-size} \\
    & \exists i^{(s)}\in N^{(s)} \ \ \mathrm{s.t.} \ \ v_{i^{(s)}}(B) \ge \alpha\,\hat\mu_{i^{(s)}}. \label{eq:single:chores:bag-value}
  \end{align}
  Furthermore, Conditions \eqref{cond:single:chores:disjoint} to \eqref{cond:single:chores:value-sum} hold for $t = s - 1$ at Line~\ref{alg:single:chores:assign} of the same iteration.
\end{lemma}

\begin{proof}
  First, we verify that Line~\ref{alg:single:chores:before-assign} is indeed reached, with all preceding operations done successfully; namely we show that the desired item exists at Line~\ref{alg:single:chores:get-from-B_k} and \ref{alg:single:chores:get-from-B}, as well as that the while-loop terminates.
  Given Condition \eqref{cond:single:chores:value-order} for $t = s$, $|B|$ never decreases, and $v_i(B)$ never increases for each $i\in N$, after Line~\ref{alg:single:chores:bag-init} until \ref{alg:single:chores:before-assign}.
  At Line~\ref{alg:single:chores:bag-init}, $B\ni g_{|M|-|N|+s}$ and $B^{(s-1)}_k\ni g_{|M|-|N|+k}$ for each $k\in\{1,2,\ldots,s-1\}$ hold due to Condition \eqref{cond:single:chores:top-n} for $t = s$, remaining true until Line~\ref{alg:single:chores:assign} by the definition of Lines~\ref{alg:single:chores:get-from-B} and \ref{alg:single:chores:get-from-B_k}.
  Let us go through a single iteration of the inner for-loop (Lines \ref{alg:single:chores:inner-for-start} to \ref{alg:single:chores:inner-for-end}) with an arbitrary $k\in\{1,2,\ldots,s-1\}$.
  Clearly, nothing occurs if $k < s - 1$ and the previous iteration has stopped with $v_i(B) < \left(\frac{1}{2}-\alpha\right)\hat\mu_i$ for every $i\in N^{(s)}$.
  Otherwise, at the beginning of this iteration, we have that $B\setminus\{g_{|M|-|N|+s}\} = B^{(s)}_{k+1}\setminus\{g_{|M|-|N|+k+1}\}$, that $B^{(s-1)}_{k} = B^{(s)}_k$, and hence $|B| = |B^{(s)}_{k+1}| \le |B^{(s)}_k| = |B^{(s-1)}_k|$ due to Condition \eqref{cond:single:chores:size} for $t = s$.
  Any single iteration of the while-loop (Lines \ref{alg:single:chores:inner-while-start} to \ref{alg:single:chores:inner-while-end}) keeps $B \cup B^{(s-1)}_{k}$ unchanged and increases $|B|$ by at most $1$, while maintaining that $g_{|M|-|N|+s} \in B$, $g_{|M|-|N|+k} \in B^{(s)}_k$, and $|B| \le |B^{(s)}_k|$.
  Therefore, that $B^{(s-1)}_k\setminus\{g_{|M|-|N|+k}\} = \emptyset$, which would imply that $B\setminus\{g_{|M|-|N|+s}\} = B^{(s)}_k\setminus\{g_{|M|-|N|+k}\}$, never occurs at Line~\ref{alg:single:chores:get-from-B_k} due to the second condition of the while-loop.
  We then obtain that $B\setminus\{g_{|M|-|N|+s}\} \neq \emptyset$ at Line~\ref{alg:single:chores:get-from-B}, where $|B| = |B^{(s)}_{k}| \ge |B^{(s)}_{k+1}| = |B^{(s-1)}_{k}| \ge 2$.
  Thanks to the choice of the items $g$ and $h$ at Lines~\ref{alg:single:chores:get-from-B} and \ref{alg:single:chores:get-from-B_k}, respectively, the while-loop successfully terminates.

  Next, we prove that \cref{eq:single:chores:bag-size,eq:single:chores:bag-value} hold at Line~\ref{alg:single:chores:before-assign}, where we define the following:
  \begin{align}
    k^* \coloneqq
    \begin{cases}
      0 & \text{\hspace{16pt} if }\ \exists i\in N^{(s)} \mbox{ \ s.t. \ } v_i(B) \ge \left(\alpha - \frac{1}{2}\right) \hat\mu_i; \\[3pt]
      \mathrlap{\min \left(\left\{k\in\{1,2,\ldots,s-1\} \mid B^{(s-1)}_k \neq B^{(s)}_k\right\} \cup \{s\}\right)} & \\
       & \text{\hspace{16pt} otherwise.\hspace{100pt}}
    \end{cases}
    \label{eq:single:chores:k*}
  \end{align}
  The value of $k^*$ represents at which point the updates on $B$ have stopped:
  $k^* = 0$ when all possible updates on $B$ run out, i.e., when it holds at Line~\ref{alg:single:chores:before-assign} that $B = (B^{(s)}_1\setminus\{g_{|M|-|N|+1}\})\cup\{g_{|M|-|N|+s}\}$, $B^{(s-1)}_1 = (B^{(s)}_2\setminus\{g_{|M|-|N|+2}\})\cup\{g_{|M|-|N|+1}\}$, \ldots, and $B^{(s-1)}_{s-1} = (B^{(s)}_{s}\setminus\{g_{|M|-|N|+s}\})\cup\{g_{|M|-|N|+s-1}\}$;
  $k^* = s$ if $B$ has not been updated at Line~\ref{alg:single:chores:before-assign} since Line~\ref{alg:single:chores:bag-init};
  otherwise (if $0 < k < s$), $B^{(s-1)}_{k^*}$ denotes the bag that has been last updated at either Line~\ref{alg:single:chores:move} or \ref{alg:single:chores:swap} (as $B^{(t-1)}_k$).
  Based on the observations in the previous paragraph and Conditions \eqref{cond:single:chores:disjoint} to \eqref{cond:single:chores:top-n} for $t = s$, we see that the following hold at Line~\ref{alg:single:chores:before-assign}:
  {\allowdisplaybreaks
  \begin{align}
    &\begin{aligned}
      B^{(s-1)}_{k} = B^{(s)}_{k} & & \forall k\in\{1, 2, \ldots, k^*-1\}.
    \end{aligned} \label{eq:single:chores:B_k-same} \\
    &\begin{aligned}
      B^{(s-1)}_{k} & = \left(B^{(s)}_{k+1}\setminus\left\{g_{|M|-|N|+k+1}\right\}\right)\cup\left\{g_{|M|-|N|+k}\right\} \\
      & \forall k\in\{k^*+1,k^*+2, \ldots, s-1\}.
    \end{aligned} \label{eq:single:chores:B_k-shift} \\
    &\mspace{-40mu}\mathrlap{
      \begin{cases}
        \mathrlap{B = \left(B^{(s)}_{\max\left\{k^*,1\right\}} \setminus \left\{g_{\max\left\{k^*,1\right\}}\right\}\right) \cup \{g_{|M|-|N|+s}\}} & \\
        & \text{if }\ k^* \in \{0, s\}; \\[3pt]
        \mathrlap{\left(B\setminus\left\{g_{|M|-|N|+s}\right\}\right) \cup B^{(s-1)}_{k^*} = B^{(s)}_{k^*} \cup \left(B^{(s)}_{k^*+1}\setminus\left\{g_{|M|-|N|+k^*+1}\right\}\right)} & \\
        & \text{otherwise.}
    \end{cases}} & & \label{eq:single:chores:union} \\
    &\mathrlap{B \cap \left\{g_{|M|},g_{|M|-1},\ldots,g_{|M|-|N|+1}\right\} = \left\{g_{|M|-|N|+s}\right\}.} \label{eq:single:chores:special-item} \\
    &q^+ \ge \left|B^{(s)}_{\max\left\{k^*,1\right\}}\right| \ge |B| \ge \left|B^{(s)}_{\min\left\{k^*+1,s\right\}}\right| \ge q^-. \label{eq:single:chores:size-ineq} \\
    &\mathrlap{B^{(s-1)}_1,B^{(s-1)}_2,\ldots,B^{(s-1)}_{s-1}, \mbox{ and $B$ are disjoint.}} \label{eq:single:chores:disjoint}
  \end{align}
  }%
  In particular, \cref{eq:single:chores:size-ineq} immediately gives \cref{eq:single:chores:bag-size}.
  Furthermore, applying \cref{eq:single:chores:value-half} of \cref{thm:single:chores:normalize} to the item $g\in B^{(s-1)}_k\setminus\left\{g_{|M|-|N|+k}\right\}$ at Line~\ref{alg:single:chores:get-from-B_k}, combined with the first condition of the while-loop, guarantees that
  \begin{align*}
    &\exists i\in N^{(s)} \quad \mathrm{s.t.} \quad v_i(B) \ge \left(\alpha - \frac{1}{2}\right)\hat\mu_i + \frac{1}{2}\,\hat\mu_i = \alpha\,\hat\mu_i,
  \end{align*}
  whenever $B$ gets updated at either Line~\ref{alg:single:chores:move} or \ref{alg:single:chores:swap}.
  At Line~\ref{alg:single:chores:bag-init}, Condition \eqref{cond:single:chores:value-sum} for $t = s$ also gives
  \begin{align*}
      v_i(B) & = v_i\left(B^{(s)}_{s}\right) \\
      & \ge \left(1 + (|N| - s)\left(\frac{3}{2} - \alpha\right)\right) \hat\mu_i \\
      & \ge \left(1 + (|N| - 1)\left(\frac{3}{2} - \alpha\right)\right) \hat\mu_i \\
      & \ge \frac{3 - |N|^{-1}}{2}\,\hat\mu_i \\
      & \ge \alpha\,\hat\mu_i
    &\forall i\in N^{(s)}.
  \end{align*}
  Therefore, we ensure that \cref{eq:single:chores:bag-value} holds at Line~\ref{alg:single:chores:before-assign}.

  Finally, we show that Conditions \eqref{cond:single:chores:disjoint} to \eqref{cond:single:chores:value-sum} hold for $t = s - 1$ at Line~\ref{alg:single:chores:assign}.
  Conditions \eqref{cond:single:chores:disjoint} to \eqref{cond:single:chores:value-order} for $t = s - 1$ follow from those for $t = s$ and \crefrange{eq:single:chores:B_k-same}{eq:single:chores:disjoint}.
  We fix any pair $r\in\{1,2,\ldots,s-1\}$ and $i\in N^{(s-1)}\subset N^{(s)}$, for which the inequality in Condition \eqref{cond:single:chores:value-sum} for $t = s - 1$ is obtained in either of the following cases.

  \begin{case} \label{case:single:chores:1}
    Suppose $r > k^*$.
    Then given \cref{eq:single:chores:B_k-shift}, Conditions \eqref{cond:single:chores:disjoint}, \eqref{cond:single:chores:top-n}, and \eqref{cond:single:chores:value-sum} for $t = s$, and that $\alpha\le\frac{3}{2}$, we establish that
    \begin{align*}
      v_i\left(\bigcup_{k=r}^{s-1}B^{(s-1)}_k\right) &= v_i\left(\bigcup_{k=r}^{s-1}\left(\left(B^{(s)}_{k+1}\setminus\{g_{|M|-|N|+k+1}\}\right)\cup\{g_{|M|-|N|+k}\}\right)\right) \\
      &\ge v_i\left(\bigcup_{k=r+1}^{s}B^{(s)}_{k}\right) \\
      &\ge \left(s - r + (|N| - s)\left(\frac{3}{2} - \alpha\right)\right) \hat\mu_i \\
      &\ge \left(s - r + (|N| - s + 1)\left(\frac{3}{2} - \alpha\right)\right) \hat\mu_i.
    \end{align*}
  \end{case}

  \begin{case} \label{case:single:chores:2}
    Suppose $1\le r \le k^*$.
    Then by the definition of $k^*$ in \cref{eq:single:chores:k*}, we must have that
    \begin{align*}
      v_i(B) < \left(\alpha - \frac{1}{2}\right)\hat\mu_i 
    \end{align*}
    at Line~\ref{alg:single:chores:before-assign}.
    It also follows from $r \le k^*$ and \crefrange{eq:single:chores:B_k-same}{eq:single:chores:special-item} that
    \begin{align*}
      B \cup B^{(s-1)}_r \cup B^{(s-1)}_{r + 1} \cup\cdots\cup B^{(s-1)}_{s-1} =  B^{(s)}_r \cup B^{(s)}_{r + 1} \cup\cdots\cup B^{(s)}_{s}.
    \end{align*}
    These observations and Condition \eqref{cond:single:chores:value-sum} for $t = s$ together establish that
    \begin{align*}
      v_i\left(\bigcup_{k=r}^{s-1}B^{(s-1)}_k\right)
      &\ge v_i\left(\bigcup_{k=r}^{s}B^{(s)}_{k}\right) - v_i(B) \\
      &> \left(s - r + 1 + (|N| - s)\left(\frac{3}{2} - \alpha\right)\right) \hat\mu_i - \left(\alpha - \frac{1}{2}\right)\hat\mu_i \\
      &= \left(s - r + (|N| - s + 1)\left(\frac{3}{2} - \alpha\right)\right) \hat\mu_i.
    \end{align*}
  \end{case}

  The proof of \cref{thm:single:chores:invariant} is now completed.
\end{proof}

\begin{corollary} \label{thm:single:chores:invariant-all}
  After Line~\ref{alg:single:chores:B_k-init-for-start} is reached, an $\alpha$-MMS allocation $(A_i)_{i\in N}$ for $\cI$ is returned at Line~\ref{alg:single:chores:allocation-final}.
\end{corollary}

\begin{proof}
  \cref{thm:single:chores:invariant-init,thm:single:chores:invariant} together ensure that Line~\ref{alg:single:chores:allocation-final} is successfully reached after Line~\ref{alg:single:chores:var-init}, with Conditions \eqref{cond:single:chores:disjoint} to \eqref{cond:single:chores:value-sum} fulfilled for every $t\in\left\{|N|,|N|-1,\ldots,1,0\right\}$.
  Then in particular, Conditions \eqref{cond:single:chores:disjoint} to \eqref{cond:single:chores:union} for $t = 0$ and \cref{eq:single:chores:bag-size} of \cref{thm:single:chores:invariant} imply that $(A_i)_{i\in N} = \left(A_{i^{(t)}}\right)_{t\in\{|N|,|N|-1,\ldots,1\}}$ is a feasible allocation for $\cI$.
  \cref{thm:single:chores:normalize} and \cref{eq:single:chores:bag-value} of \cref{thm:single:chores:invariant} also yield that
  \begin{align*}
    &v_{i^{(t)}}\!\left(A_{i^{(t)}}\right) \ge \alpha\,\hat\mu_{i^{(t)}} \ge \alpha\,\mu_{i^{(t)}}(\cI) &\forall t\in\left\{|N|,|N|-1,\ldots,1\right\}.
  \end{align*}
  establishing that $(A_i)_{i\in N}$ is an $\alpha$-MMS allocation for $\cI$.
\end{proof}



\begin{lemma} \label{thm:single:chores:runtime}
  \hyperref[alg:single:chores:main-start]{\textproc{ApproxChores}} runs in time $O(|N||M|)$.
\end{lemma}

\begin{proof}
  Lines \ref{alg:single:chores:B_k-init-for-start} to \ref{alg:single:chores:var-init} run in time $O(|M||N|)$.
  In each of the $|N|$ iterations of the outer for-loop (Lines \ref{alg:single:chores:outer-for-start} to \ref{alg:single:chores:outer-for-end}), the inner for-loop (Lines \ref{alg:single:chores:inner-for-start} to \ref{alg:single:chores:inner-for-end}) iterates $t - 1 < |N| < |M|$ times, during which the while-loop (Lines \ref{alg:single:chores:inner-while-start} to \ref{alg:single:chores:inner-while-end}) iterates at most $\sum_{k=1}^{t-1}(|B^{(t)}_k|-1) \le |M|$ times in total.
\end{proof}

Finally, \cref{thm:single:chores:few-items,thm:single:chores:invariant-all} together establish the correctness.
We conclude the proof of \cref{thm:single:chores:three-halves} by showing that \hyperref[alg:single:chores:main-start]{\textproc{ApproxChores}} runs in time $O(|N||M|)$; refer to \cref{thm:common:ordered} for the additional time complexity due to the reduction to an ordered instance.
In each of the $|N|$ iterations over Lines \ref{alg:single:chores:outer-for-start} to \ref{alg:single:chores:outer-for-end}, the inner for-loop (Lines \ref{alg:single:chores:inner-for-start} to \ref{alg:single:chores:inner-for-end}) iterates $t - 1 < |N| < |M|$ times, during which the while-loop (Lines \ref{alg:single:chores:inner-while-start} to \ref{alg:single:chores:inner-while-end}) iterates at most $\sum_{k=1}^{t-1}(|B^{(t)}_k|-1) \le |M|$ times in total.

\subsection{Tight Instances for the Algorithm} \label{sec:single:chores:tight}

Our analysis above is tight for any number of agents, even when they have identical valuations.

\begin{theorem} \label{thm:single:chores:tight}
  For any $n \in \{1,2,\ldots\}$, there is a single-category ordered instance $\cI_n$ with $n$ agents, $2n$ chores, and identical valuations, which satisfies the following: for any $\beta \in \left(-\infty, \frac{3n - 1}{2n}\right)$, no $\beta$-MMS allocation for $\cI_n$ is obtained by \hyperref[alg:single:chores:main-start]{\textproc{ApproxChores}} given $\cI_{n}$ and any $\alpha\in\bR$ as input.
\end{theorem}

\begin{proof}
  Let $n$ be an arbitrary positive integer.
  We prove the claimed property of the single-category ordered instance of chores $\cI_n = \left(N, M, (v_i)_{i\in N}, (q^-, q^+)\right)$ defined as follows:
  \begin{align*}
    N & \coloneqq \{1,2,\ldots,n\}, \\
    M & \coloneqq  \{g_1,g_2,\ldots,g_{2n}\},\\
    \left(q^-,q^+\right) & \coloneqq (1,n+1); \\
    v_i(g_j) & \coloneqq
    \begin{cases}
      -\frac{1}{2n} & \text{ if }\, j < n, \\[2pt]
      -\frac{1}{2} & \text{ if }\, n \le j \le n + 1, \\[2pt]
      \frac{1}{2n} - 1  & \text{ otherwise}
    \end{cases} &\forall g_j\in M,\forall i\in N.
  \end{align*}
  Given the partition $P \coloneqq \left(P_k \coloneqq \{g_k, g_{2n+1-k}\}\right)_{k\in\{1,2,\ldots,n\}}\in\cF(\cI_n)$,
  we have that
  \begin{align}
    \mu_i(\cI_n) &= v_i(P_k) = -1 = \frac{v_i(M)}{n} &\forall k\in\{1,2,\ldots,n\},\forall i\in N. \label{eq:single:chores:tight-MMS}
  \end{align}

  Let us fix any $\beta \in \left(-\infty, \frac{3n - 1}{2n}\right)$, and suppose that \hyperref[alg:single:chores:main-start]{\textproc{ApproxChores}} runs given $\cI_{n}$ and some $\alpha$ as input.
  At Line~\ref{alg:single:chores:check-few-items}, $\cI_n$ satisfies $|M| > |N|$.
  At Line~\ref{alg:single:chores:mu-hat}, it follows that
  \begin{align*}
    & \begin{aligned}
      & 2\,v_i\left(g_{|M|-|N|}\right) = -1 & \forall i\in N,
    \end{aligned} \\
    & v_i\left(B^{(|N|)}_k\right) =
    \begin{cases}
      v_i\left(\left\{g_1,g_2,\ldots,g_{n+1}\right\}\right) = \frac{1}{2n} - \frac{3}{2} & \text{ if }\, k = 1, \\[3pt]
      v_i(\{g_{n + k}\}) = \frac{1}{2n} - 1 &\text{ otherwise}
    \end{cases} \\
    & \quad \forall k\in\{1,2,\ldots,n\},\forall i\in N.
  \end{align*}
  Therefore, Line~\ref{alg:single:chores:var-init} is reached with $\hat\mu_i = -1$ for every $i\in N$.
  The rest of the analysis varies by the regime of $\alpha$.

  \begin{case}
    Suppose $\alpha < 1$.
    Then \hyperref[alg:single:chores:main-start]{\textproc{ApproxChores}} never successfully terminates; otherwise, Lines~\ref{alg:single:chores:before-assign} and \ref{alg:single:chores:assign} would imply that
    \begin{align*}
      v_i(M) = \sum_{t=1}^{n} v_{i}(A_{i^{(t)}}) & = \sum_{t=1}^{n} v_{i^{(t)}}(A_{i^{(t)}}) \\
      & \ge \sum_{t=1}^n \alpha\,\hat\mu_{i^{(t)}} = -n\alpha > -n & \forall i\in N,
    \end{align*}
    which contradicts the definition of $\cI_n$.
  \end{case}

  \begin{case}
    Suppose $1 \le \alpha < \frac{3n - 1}{2n}$.
    Then for each $t \in \{n, n - 1, \ldots, 2\}$, \hyperref[alg:single:chores:main-start]{\textproc{ApproxChores}} gets
    \begin{align*}
      A_{i^{(t)}} = B^{(t)}_t = B^{(|N|)}_t = \{g_{n + t}\}
    \end{align*}
    at Line~\ref{alg:single:chores:assign}.
    In the final iteration with $t = 1$, Line~\ref{alg:single:chores:bag-init} is thus reached with
    \begin{align*}
      v_i(B) &= v_i\!\left(B^{(1)}_1\right) = v_i\!\left(B^{(|N|)}_1\right) = -\frac{3n - 1}{2n} < -\alpha = \alpha\,\hat\mu_i &\forall i\in N^{(1)},
    \end{align*}
    which immediately makes Line~\ref{alg:single:chores:before-assign} unsuccessful.
  \end{case}

  \begin{case}
    Suppose $\alpha \ge \frac{3n - 1}{2n}$.
    Then in the first iteration with $t = n$, Line~\ref{alg:single:chores:assign} is reached with
    \begin{align*}
      v_{i^{(n)}}\!\left(A_{i^{(n)}}\right) &= v_{i^{(n)}}(\{g_n, g_{2n}\}) = -\frac{3n - 1}{2n} < -\beta = \beta\,\mu_{i^{(n)}}(\cI).
    \end{align*}
    Thus, \hyperref[alg:single:chores:main-start]{\textproc{ApproxChores}} never returns a $\beta$-MMS allocation for $\cI_n$.
  \end{case}

  The proof of \cref{thm:single:chores:tight} is now completed.
\end{proof}

\section{Multiple categories} \label{sec:multi}

Here, we turn to general (multi-category) instances defined in \cref{sec:prelim} and establish the following theorems, restated from \cref{sec:main_results}.
Our algorithms build on and extend their counterparts by \citet{hummel2022maximin}.

\MultiGoodsMain*
\MultiChoresMain*

\subsection{Goods: An $\left(\frac{n}{2n-1}\right)$-MMS Allocation Algorithm} \label{sec:multi:goods}

We present a polynomial-time algorithm to compute a $\left(\frac{n}{2n - 1}\right)$-MMS allocation for an ordered instance of goods with $n$ agents.
Proving \cref{thm:multi:goods:half:appendix}, combined with \cref{thm:common:ordered}, establishes \cref{thm:multi:goods:half}.

\begin{algorithm}[tb]
  \centering

  \begin{algorithmic}[1]
    \Function{ApproxCategorizedGoods}{$\cI = \left(N,M,(v_i)_{i\in N},\cC,(q_{C}^-,q_{C}^+)_{C\in\cC}\right),\alpha$} \label{alg:multi:goods:main-start}

    \State $\hat\mu_i \gets \frac{1}{|N|}\, v_i(M)$ \ for each $i\in N$. \label{alg:multi:goods:mu-hat}

    \If{$\exists i^*\in N,\exists C^*\in\cC$ s.t. $C^*\neq\emptyset$ and $v_{i^*}(g^{C^*}_1) \ge \alpha\,\hat\mu_i$} \label{alg:multi:goods:check-best-item}
    \State{$
      \begin{aligned}
        A_{i^*} & \gets \left\{g^{C^*}_{1}\right\} \\
        & \cup \left\{g^{C^*}_{|C^*| - j} \mid 0\le j < \max\left\{q_{C^*}^-,|C^*|-q_{C^*}^+(|N|-1)\right\} - 1 \right\} \\
        & \cup \bigcup_{C\in\cC\setminus\{C^*\}}\left\{g^{C}_{|C| - j} \mid 0\le j < \max\left\{q_{C}^-,|C|-q_{C}^+(|N|-1)\right\}\right\}.
    \end{aligned}$} \label{alg:multi:goods:valid-reduce-assign}
    \State $\begin{aligned}
    \cI' \gets \left(N\setminus\{i^*\}, M\setminus A_{i^*}, \left(v_{i}|_{2^{M\setminus A_{i^*}}}\right)_{i\in N\setminus \{i^*\}},\right. \\ \left.\{C\setminus A_{i^*}\mid C\in\cC\}, \left(q_{C}^-,q_{C}^+\right)_{C\in\cC}\right).
    \end{aligned}$ \label{alg:multi:goods:valid-reduced-instance}
    \State $(A_{i})_{i\in N\setminus \{i^*\}}\gets\Call{\hyperref[alg:multi:goods:main-start]{ApproxCategorizedGoods}}{\cI',\alpha}$. \label{alg:multi:goods:valid-reduce}
    \State \Return $(A_i)_{i\in N}$ \label{alg:multi:goods:allocation-valid-reduction}
    \EndIf

    \State $N^{(|N|)}\gets N$, $M^{(|N|)}\gets M$. \label{alg:multi:goods:var-init}

    \For{$t \gets |N|,|N|-1,\ldots,1$} \label{alg:multi:goods:outer-for-start}
    \State $\displaystyle B \gets \bigcup_{C\in\cC} \left\{\mbox{the least valuable $\left\lfloor\frac{|C\cap M^{(t)}|}{t}\right\rfloor$ items in $C\cap M^{(t)}$}\right\}.$ \label{alg:multi:goods:bag-init}

    \For{\textbf{each} $C \in \cC$} \label{alg:multi:goods:inner-for-start}
    \While{$\left(\forall i\in N^{(t)}, \ v_i(B) < \alpha\,\hat\mu_i\right)$ and $B \cap C \neq \left\{\mbox{the most valuable $\left\lceil\frac{|C\cap M^{(t)}|}{t}\right\rceil$ items in $C\cap M^{(t)}$}\right\}$} \label{alg:multi:goods:inner-while-start}
    \If{$|B\cap C| = \left\lceil\frac{|C\cap M^{(t)}|}{t}\right\rceil$}
    \State Remove from $B$ the least valuable item in $B\cap C$. \label{alg:multi:goods:remove-from-B}
    \EndIf
    \State Add to $B$ the most valuable item in $(C \cap M^{(t)}) \setminus B$.
    \EndWhile \label{alg:multi:goods:inner-while-end}
    \EndFor \label{alg:multi:goods:inner-for-end}

    \State Find $i^{(t)}\in N^{(t)}$ s.t. $v_{i^{(t)}}(B) \ge \alpha\,\hat\mu_{i^{(t)}}$.\label{alg:multi:goods:before-assign}
    \State $A_{i^{(t)}}\gets B$, $N^{(t - 1)}\gets N^{(t)}\setminus \left\{i^{(t)}\right\}$, $M^{(t-1)} \gets M^{(t)}\setminus B$. \hspace{-5pt} \label{alg:multi:goods:assign}

    \EndFor \label{alg:multi:goods:outer-for-end}

    \State \Return $(A_i)_{i\in N}$. \label{alg:multi:goods:allocation-final}

    \EndFunction
  \end{algorithmic}

  \caption{Compute a $\left(\frac{1}{2 - (\max\{1,n\})^{-1}}\right)$-MMS allocation for an ordered instance of goods with $n$ agents (cf.~\citep[Algorithm 3]{hummel2022maximin}).}
  \label{alg:multi:goods}
\end{algorithm}

\begin{theorem} \label{thm:multi:goods:half:appendix}
  Given as input an arbitrary ordered instance of goods $\cI = \left(N,M,(v_i)_{i\in N},\cC,(q_{C}^-,q_{C}^+)_{C\in\cC}\right)$ and any real constant $\alpha \in \left[\frac{1}{2}, \frac{1}{2 - \left(\max\left\{|N|,1\right\}\right)^{-1}}\right]$, \hyperref[alg:multi:goods:main-start]{\textproc{ApproxCategorizedGoods}} defined in \cref{alg:multi:goods} returns an $\alpha$-MMS allocation for $\cI$ in time $O(|N||M|\log |M|)$.
\end{theorem}

We prove \cref{thm:multi:goods:half:appendix} by induction on the number of agents, $|N|$.
In what follows, let us fix an arbitrary ordered instance of goods $\cI = \left(N,M,(v_i)_{i\in N},\cC,(q_{C}^-,q_{C}^+)_{C\in\cC}\right)$ and any real constant $\alpha \in \left[\frac{1}{2}, \frac{1}{2 - \left(\max\left\{|N|,1\right\}\right)^{-1}}\right]$ as input to \hyperref[alg:multi:goods:main-start]{\textproc{ApproxCategorizedGoods}}.
It suffices to show \cref{thm:multi:goods:valid-reduction,thm:multi:goods:alpha-MMS-final,thm:multi:goods:runtime} below.
All line numbers in this section refer to \cref{alg:multi:goods}.

\begin{lemma} \label{thm:multi:goods:normalize}
  At Line~\ref{alg:multi:goods:mu-hat}, it holds for each $i\in N$ that $\hat\mu_i \ge \mu_i(\cI)$.
\end{lemma}

\begin{proof}
  The claim is immediate from \cref{def:MMS}.
\end{proof}

\begin{lemma} \label{thm:multi:goods:valid-reduction}
  If Line~\ref{alg:multi:goods:valid-reduce-assign} is reached, an $\alpha$-MMS allocation $(A_i)_{i\in N}$ for $\cI$ is returned at Line~\ref{alg:multi:goods:allocation-valid-reduction}.
  Otherwise, it follows at Line~\ref{alg:multi:goods:var-init} that $v_i(g) < \alpha\,\hat\mu_i$ for every $g\in M$ and $i\in N$.
\end{lemma}

\begin{proof}
  Due to \cref{thm:common:valid-reduction} for $d = 0$, the instance $\cI'$ at Line~\ref{alg:multi:goods:valid-reduced-instance} is well-defined as an ordered instance of goods.
  Since we have
  \begin{align}
    \alpha \le \frac{1}{2 - \max\left\{|N|,1\right\}^{-1}} \le \frac{1}{2 - \max\left\{|N\setminus\{i^*\}|,1\right\}^{-1}},
  \end{align}
  the induction hypothesis ensures that $(A_i)_{i\in N\setminus\{i^*\}}$ at Line~\ref{alg:multi:goods:valid-reduce} is an $\alpha$-MMS allocation for $\cI'$.
  Because Line~\ref{alg:multi:goods:check-best-item} and \cref{thm:multi:goods:valid-reduction} together guarantee that
  \begin{align}
    v_{i^*}(A_{i^*}) \ge v_{i^*}\!\left(g^{C^*}_1\right) \ge \alpha\,\hat\mu_{i^*} \ge \alpha\,\mu_{i^*}(\cI),
  \end{align}
  \cref{thm:common:cor-valid-reduction} completes the proof.
\end{proof}

\begin{definition} \label{def:multi:goods:invariants}
  Along with the outer for-loop over Lines \ref{alg:multi:goods:outer-for-start} to \ref{alg:multi:goods:outer-for-end}, we consider the following conditions for each $t\in\{|N|, |N| - 1, \ldots, 1, 0\}$:
  \begin{align}
    \tag{C1} & \mathrlap{A_{i^{(t + 1)}}, A_{i^{(t + 2)}}, \ldots, A_{i^{(|N|)}}, \mbox{and $M^{(t)}$ are disjoint.}} & \label{cond:multi:goods:disjoint} \\[0pt]
    \tag{C2} & \mathrlap{A_{i^{(t + 1)}}\cup A_{i^{(t + 2)}}\cup\cdots\cup A_{i^{(|N|)}} \cup M^{(t)} = M.} & \label{cond:multi:goods:union} \\[0pt]
    \tag{C3} & q_C^-\, t \le |C \cap M^{(t)}| \le q_C^+\, t & & \forall C\in \cC. \label{cond:multi:goods:size} \\[0pt]
    \tag{C4} & v_i\left(M^{(t)}\right) \ge \left(t - (|N| - t)\left(2\,\alpha - 1\right)\right) \hat\mu_i & & \forall i\in N^{(t)}.
    \label{cond:multi:goods:value-sum}
  \end{align}
  \reqnomode
\end{definition}

\begin{lemma} \label{thm:multi:goods:invariant-init}
  At Line~\ref{alg:multi:goods:var-init}, Conditions \eqref{cond:multi:goods:disjoint} to \eqref{cond:multi:goods:value-sum} hold for $t = |N|$.
\end{lemma}

\begin{proof}
  The claim is immediate by definition.
\end{proof}

\begin{lemma}[{\citep[Lemmas 2 and 3]{hummel2022maximin}}] \label{thm:multi:goods:invariant}
  Let $s\in\{|N|, |N| - 1 \ldots, 1\}$ be arbitrary.
  Suppose that the for-loop over Lines \ref{alg:multi:goods:outer-for-start} to \ref{alg:multi:goods:outer-for-end} has successfully iterated for $t \in \{|N|,|N|-1,\ldots,s+1\}$, and that Conditions \eqref{cond:multi:goods:disjoint} to \eqref{cond:multi:goods:value-sum} now hold for $t = s$.
  Then the next iteration with $t = s$ succeeds, where the following hold at Line~\ref{alg:multi:goods:before-assign}:
  \begin{align}
    & \forall C \in \cC, \ q^-_C \le \left\lfloor\frac{|C \cap M^{(s)}|}{s}\right\rfloor \le |B \cap C| \le \left\lceil\frac{|C \cap M^{(s)}|}{s}\right\rceil \le q^+_C. \label{eq:multi:goods:bag-size} \\
    & \exists i^{(s)}\in N^{(s)} \ \ \mathrm{s.t.}\ \ v_{i^{(s)}}(B) \ge \alpha\,\hat\mu_{i^{(s)}}. \label{eq:multi:goods:bag-value}
  \end{align}
  Furthermore, Conditions \eqref{cond:multi:goods:disjoint} to \eqref{cond:multi:goods:value-sum} hold for $t = s - 1$ at Line~\ref{alg:multi:goods:assign} of the same iteration.
\end{lemma}

\begin{proof}
  First, we confirm that Line~\ref{alg:multi:goods:before-assign} is indeed reached, with all preceding operations done successfully.
  It suffices to see at Line~\ref{alg:multi:goods:remove-from-B} that $C\cap M^{(t)} \neq \emptyset$, which must hold due to the second condition of the while-loop (Lines \ref{alg:multi:goods:inner-while-start} to \ref{alg:multi:goods:inner-while-end}).
  We also have that the while-loop eventually terminates.

  Next, we prove that \cref{eq:multi:goods:bag-size,eq:multi:goods:bag-value} hold at Line~\ref{alg:multi:goods:before-assign}.
  From Line~\ref{alg:multi:goods:bag-init} to \ref{alg:multi:goods:before-assign}, it always holds for each $C\in \cC$ that
  \begin{align}
    \left\lfloor\frac{|C \cap M^{(s)}|}{s}\right\rfloor \le |B \cap C| \le \left\lceil\frac{|C \cap M^{(s)}|}{s}\right\rceil,
  \end{align}
  which, together with Condition \eqref{cond:multi:goods:size} for $t = s$, gives \cref{eq:multi:goods:bag-size} at Line~\ref{alg:multi:goods:before-assign}.
  Condition \eqref{cond:multi:goods:value-sum} for $t = s$ also guarantees for each $i\in N^{(s)}$ that
  \begin{align}
    & v_i\left(\bigcup_{C\in\cC} \left\{\mbox{the most valuable $\left\lceil\frac{|C\cap M^{(s)}|}{s}\right\rceil$ items in $C\cap M^{(s)}$}\right\}\right) \\
    & \ge \frac{s - (|N| - s)\left(2\,\alpha - 1\right)}{s}\,\hat\mu_i \\
    & \ge \left(1 - (|N| - 1)\left(2\,\alpha - 1\right)\right)\,\hat\mu_i \\
    & \ge \alpha\,\hat\mu_i,
  \end{align}
  because $\alpha \in \left[\frac{1}{2},\frac{1}{2 - |N|^{-1}}\right]$.
  Along with the conditions of the while-loop, this implies \cref{eq:multi:goods:bag-value} at Line~\ref{alg:multi:goods:before-assign}.

  Finally, we show that Conditions \eqref{cond:multi:goods:disjoint} to \eqref{cond:multi:goods:value-sum} for $t = s-1$ hold at Line~\ref{alg:multi:goods:assign}.
  Conditions \eqref{cond:multi:goods:disjoint} to \eqref{cond:multi:goods:size} for $t = s-1$ follow from those for $t = s$, \cref{eq:multi:goods:bag-size} at Line~\ref{alg:multi:goods:before-assign}, as well as Line~\ref{alg:multi:goods:assign}.
  Combining the latter statement of \cref{thm:multi:goods:valid-reduction} with the first condition of the while-loop, it must hold at Line~\ref{alg:multi:goods:before-assign} that
  \begin{align}
    & v_i(B) < \alpha\,\hat\mu_i + \alpha\,\hat\mu_i = 2\,\alpha\,\hat\mu_i & \forall i\in N^{(s)}.
  \end{align}
  Together with Condition \eqref{cond:multi:goods:value-sum} for $t = s$, we obtain at Line \ref{alg:multi:goods:assign} that
  \begin{align}
    & v_i\left(M^{(s-1)}\right) \\
    & = v_i\left(M^{(s)}\right) - v_i(B) \\
    & \ge \left(s - (|N| - s)\left(2\,\alpha - 1\right)\right) \hat\mu_i - 2\,\alpha\,\hat\mu_i \\
    & \ge \left(s - 1 - (|N| - s + 1)\left(2\,\alpha - 1\right)\right) \hat\mu_i
    & \forall i\in N^{(s)},
  \end{align}
  where Condition \eqref{cond:multi:goods:value-sum} for $t = s-1$ is established.
\end{proof}

\begin{corollary} \label{thm:multi:goods:alpha-MMS-final}
  If Line~\ref{alg:multi:goods:var-init} is reached, an $\alpha$-MMS allocation $(A_i)_{i\in N}$ for $\cI$ is returned at Line~\ref{alg:multi:goods:allocation-final}.
\end{corollary}

\begin{proof}
  \cref{thm:multi:goods:invariant-init,thm:multi:goods:invariant} together ensure that Line~\ref{alg:multi:goods:allocation-final} is successfully reached after Line~\ref{alg:multi:goods:var-init}, with Conditions \eqref{cond:multi:goods:disjoint} to \eqref{cond:multi:goods:value-sum} fulfilled for every $t\in\left\{|N|,|N|-1,\ldots,1,0\right\}$.
  Also given Conditions \eqref{cond:multi:goods:disjoint} to \eqref{cond:multi:goods:union} for $t = 0$ and \cref{eq:multi:goods:bag-size} of \cref{thm:multi:goods:invariant}, it is implied that $(A_i)_{i\in N} = \left(A_{i^{(t)}}\right)_{t\in\{|N|,|N|-1,\ldots,1\}}$ is a feasible allocation for $\cI$.
  \cref{thm:multi:goods:normalize} and \cref{eq:multi:goods:bag-value} of \cref{thm:multi:goods:invariant} also yield that
  \begin{align*}
    &v_{i^{(t)}}\!\left(A_{i^{(t)}}\right) \ge \alpha\,\hat\mu_{i^{(t)}} \ge \alpha\,\mu_{i^{(t)}}(\cI) &\forall t\in\left\{|N|,|N|-1,\ldots,1\right\},
  \end{align*}
  establishing that $(A_i)_{i\in N}$ is an $\alpha$-MMS allocation for $\cI$.
\end{proof}

\begin{lemma} \label{thm:multi:goods:runtime}
  \hyperref[alg:multi:goods:main-start]{\textproc{ApproxCategorizedGoods}} runs in time $O(|N||M|)$.
\end{lemma}

\begin{proof}
  The recursion in Line~\ref{alg:multi:goods:valid-reduce} reduces $\min\{|N|,|M|\}$ by one, while Lines \ref{alg:multi:goods:mu-hat} to \ref{alg:multi:goods:valid-reduced-instance} run in time $O(|N|+|M|)$.
  In each of the $|N|$ iterations of the outer for-loop (Lines \ref{alg:multi:goods:outer-for-start} to \ref{alg:multi:goods:outer-for-end}), the inner for-loop (Lines \ref{alg:multi:goods:inner-for-start} to \ref{alg:multi:goods:inner-for-end}) iterates over all categories in $\cC$, during which the while-loop (Lines \ref{alg:multi:goods:inner-while-start} to \ref{alg:multi:goods:inner-while-end}) iterates at most $\sum_{C\in\cC} |C\cap M^{(t)}| \le |M|$ times in total.
\end{proof}

\subsection{Chores: A $\left(2 - \frac{1}{n}\right)$-MMS Allocation Algorithm} \label{sec:multi:chores}

We present a polynomial-time algorithm to compute a $\left(2 - \frac{1}{n}\right)$-MMS allocation for an ordered instance of chores with $n$ agents.
Proving \cref{thm:multi:chores:two:appendix}, combined with \cref{thm:common:ordered}, establishes \cref{thm:multi:chores:two}.

\begin{algorithm}[tb]
  \centering

  \begin{algorithmic}[1]
    \Function{ApproxCategorizedChores}{$\cI = \left(N,M,(v_i)_{i\in N},\cC,(q_{C}^-,q_{C}^+)_{C\in\cC}\right),\alpha$} \label{alg:multi:chores:main-start}

    \State $\displaystyle \hat\mu_i \gets \min\left\{\frac{1}{|N|}\, v_i(M), \min_{g\in M} v_i(g)\right\}$ \ for each $i\in N$. \label{alg:multi:chores:mu-hat}

    \State $N^{(|N|)}\gets N$, $M^{(|N|)}\gets M$. \label{alg:multi:chores:var-init}

    \For{$t \gets |N|,|N|-1,\ldots,1$} \label{alg:multi:chores:outer-for-start}
    \State $\displaystyle B \gets \bigcup_{C\in\cC} \left\{\mbox{the least valuable $\left\lceil\frac{|C\cap M^{(t)}|}{t}\right\rceil$ items in $C\cap M^{(t)}$}\right\}.$ \label{alg:multi:chores:bag-init}

    \For{\textbf{each} $C \in \cC$} \label{alg:multi:chores:inner-for-start}
    \While{$\left(\forall i\in N^{(t)}, \ v_i(B) < \alpha\,\hat\mu_i\right)$ and $B \cap C \neq \left\{\mbox{the most valuable $\left\lfloor\frac{|C\cap M^{(t)}|}{t}\right\rfloor$ items in $C\cap M^{(t)}$}\right\}$} \label{alg:multi:chores:inner-while-start}
    \State Remove from $B$ the least valuable item in $B\cap C$. \hspace{-5pt} \label{alg:multi:chores:remove-from-B}
    \If{$|B\cap C| = \left\lfloor\frac{|C\cap M^{(t)}|}{t}\right\rfloor$}
    \State Add to $B$ the most valuable item in $(C \cap M^{(t)}) \setminus B$.
    \EndIf
    \EndWhile \label{alg:multi:chores:inner-while-end}
    \EndFor \label{alg:multi:chores:inner-for-end}

    \State Find $i^{(t)}\in N^{(t)}$ s.t. $v_{i^{(t)}}(B) \ge \alpha\,\hat\mu_{i^{(t)}}$.\label{alg:multi:chores:before-assign}
    \State $A_{i^{(t)}}\gets B$, $N^{(t - 1)}\gets N^{(t)}\setminus \left\{i^{(t)}\right\}$, $M^{(t-1)} \gets M^{(t)}\setminus B$. \hspace{-5pt} \label{alg:multi:chores:assign}

    \EndFor \label{alg:multi:chores:outer-for-end}

    \State \Return $(A_i)_{i\in N}$. \label{alg:multi:chores:allocation-final}

    \EndFunction
  \end{algorithmic}

  \caption{Compute a $\left(2 - \frac{1}{\max\{1,n\}}\right)$-MMS allocation for an ordered instance of chores with $n$ agents (cf.~\citep[Algorithm 5]{hummel2022maximin}).}
  \label{alg:multi:chores}
\end{algorithm}

\begin{theorem} \label{thm:multi:chores:two:appendix}
  Given as input an arbitrary ordered instance of chores $\cI = \left(N,M,(v_i)_{i\in N},\cC,(q_{C}^-,q_{C}^+)_{C\in\cC}\right)$ and any real constant $\alpha \in \left[2 - \frac{1}{\left(\max\left\{|N|,1\right\}\right)}, 2\right]$, \hyperref[alg:multi:chores:main-start]{\textproc{ApproxCategorizedChores}} defined in \cref{alg:multi:chores} returns an $\alpha$-MMS allocation for $\cI$ in time $O(|N||M|\log |M|)$.
\end{theorem}

In what follows, let us fix an arbitrary ordered instance of chores $\cI = \left(N,M,(v_i)_{i\in N},\cC,(q_{C}^-,q_{C}^+)_{C\in\cC}\right)$ and any real constant $\alpha \in \left[2 - \frac{1}{\left(\max\left\{|N|,1\right\}\right)}, 2\right]$ as input to \hyperref[alg:multi:chores:main-start]{\textproc{ApproxCategorizedChores}}.
It suffices to show \cref{thm:multi:chores:alpha-MMS-final,thm:multi:chores:runtime} below.
All line numbers in this section refer to \cref{alg:multi:chores}.

\begin{lemma} \label{thm:multi:chores:normalize}
  At Line~\ref{alg:multi:chores:mu-hat}, it holds for each $i\in N$ that $\hat\mu_i \ge \mu_i(\cI)$.
\end{lemma}

\begin{proof}
  The claim is immediate from \cref{def:MMS}.
\end{proof}

\begin{definition} \label{def:multi:chores:invariants}
  Along with the outer for-loop over Lines \ref{alg:multi:chores:outer-for-start} to \ref{alg:multi:chores:outer-for-end}, we consider the following conditions for each $t\in\{|N|, |N| - 1, \ldots, 1, 0\}$:
  \begin{align}
    \tag{C1} & \mathrlap{A_{i^{(t + 1)}}, A_{i^{(t + 2)}}, \ldots, A_{i^{(|N|)}}, \mbox{and $M^{(t)}$ are disjoint.}} & \label{cond:multi:chores:disjoint} \\[0pt]
    \tag{C2} & \mathrlap{A_{i^{(t + 1)}}\cup A_{i^{(t + 2)}}\cup\cdots\cup A_{i^{(|N|)}} \cup M^{(t)} = M.} & \label{cond:multi:chores:union} \\[0pt]
    \tag{C3} & q_C^-\, t \le |C \cap M^{(t)}| \le q_C^+\, t & & \forall C\in \cC. \label{cond:multi:chores:size} \\[0pt]
    \tag{C4} & v_i\left(M^{(t)}\right) \ge \left(t + (|N| - t)\left(2 - \alpha\right)\right) \hat\mu_i & & \forall i\in N^{(t)}.
    \label{cond:multi:chores:value-sum}
  \end{align}
  \reqnomode
\end{definition}

\begin{lemma} \label{thm:multi:chores:invariant-init}
  At Line~\ref{alg:multi:chores:var-init}, Conditions \eqref{cond:multi:chores:disjoint} to \eqref{cond:multi:chores:value-sum} hold for $t = |N|$.
\end{lemma}

\begin{proof}
  The claim is immediate by definition.
\end{proof}

\begin{lemma}[{\citep[Lemmas 8 and 9]{hummel2022maximin}}] \label{thm:multi:chores:invariant}
  Let $s\in\{|N|, |N| - 1 \ldots, 1\}$ be arbitrary.
  Suppose that the for-loop over Lines \ref{alg:multi:chores:outer-for-start} to \ref{alg:multi:chores:outer-for-end} has successfully iterated for $t \in \{|N|,|N|-1,\ldots,s+1\}$, and that Conditions \eqref{cond:multi:chores:disjoint} to \eqref{cond:multi:chores:value-sum} now hold for $t = s$.
  Then the next iteration with $t = s$ succeeds, where the following hold at Line~\ref{alg:multi:chores:before-assign}:
  \begin{align}
    & \forall C \in \cC, \ q^-_C \le \left\lfloor\frac{|C \cap M^{(s)}|}{s}\right\rfloor \le |B \cap C| \le \left\lceil\frac{|C \cap M^{(s)}|}{s}\right\rceil \le q^+_C. \label{eq:multi:chores:bag-size} \\
    & \exists i^{(s)}\in N^{(s)} \ \ \mathrm{s.t.}\ \ v_{i^{(s)}}(B) \ge \alpha\,\hat\mu_{i^{(s)}}. \label{eq:multi:chores:bag-value}
  \end{align}
  Furthermore, Conditions \eqref{cond:multi:chores:disjoint} to \eqref{cond:multi:chores:value-sum} hold for $t = s - 1$ at Line~\ref{alg:multi:chores:assign} of the same iteration.
\end{lemma}

\begin{proof}
   First, we confirm that Line~\ref{alg:multi:chores:before-assign} is indeed reached, with all preceding operations done successfully.
  It suffices to see at Line~\ref{alg:multi:chores:remove-from-B} that $B\cap C \neq \emptyset$, which must hold due to the second condition of the while-loop (Lines \ref{alg:multi:chores:inner-while-start} to \ref{alg:multi:chores:inner-while-end}).
  We also observe that the while-loop eventually terminates.

  Next, we prove that \cref{eq:multi:chores:bag-size,eq:multi:chores:bag-value} hold at Line~\ref{alg:multi:chores:before-assign}.
  From Line~\ref{alg:multi:chores:bag-init} to \ref{alg:multi:chores:before-assign}, it always holds for each $C\in \cC$ that
  \begin{align}
    \left\lfloor\frac{|C \cap M^{(s)}|}{s}\right\rfloor \le |B \cap C| \le \left\lceil\frac{|C \cap M^{(s)}|}{s}\right\rceil,
  \end{align}
  which, together with Condition \eqref{cond:multi:chores:size} for $t = s$, gives \cref{eq:multi:chores:bag-size} at Line~\ref{alg:multi:chores:before-assign}.
  Condition \eqref{cond:multi:chores:value-sum} for $t = s$ also guarantees for each $i\in N^{(s)}$ that
  \begin{align}
    & v_i\left(\bigcup_{C\in\cC} \left\{\mbox{the most valuable $\left\lfloor\frac{|C\cap M^{(s)}|}{s}\right\rfloor$ items in $C\cap M^{(s)}$}\right\}\right) \\
    & \ge \frac{s + (|N| - s)\left(2 - \alpha\right)}{s}\,\hat\mu_i \\
    & \ge \left(1 + (|N| - 1)\left(2 - \alpha\right)\right)\,\hat\mu_i \\
    & \ge \alpha\,\hat\mu_i,
  \end{align}
  because $\alpha \in \left[2 - \frac{1}{|N|}, 2\right]$.
  Along with the conditions of the while-loop, this implies \cref{eq:multi:chores:bag-value} at Line~\ref{alg:multi:chores:before-assign}.

  Finally, we show that Conditions \eqref{cond:multi:chores:disjoint} to \eqref{cond:multi:chores:value-sum} for $t = s-1$ hold at Line~\ref{alg:multi:chores:assign}.
  Conditions \eqref{cond:multi:chores:disjoint} to \eqref{cond:multi:chores:size} for $t = s-1$ follow from those for $t = s$, \cref{eq:multi:chores:bag-size} at Line~\ref{alg:multi:chores:before-assign}, as well as Line~\ref{alg:multi:chores:assign}.
  Given Line~\ref{alg:multi:chores:mu-hat} and the first condition of the while-loop, it must hold at Line~\ref{alg:multi:chores:before-assign} that
  \begin{align}
    & v_i(B) < \alpha\,\hat\mu_i - \hat\mu_i = \left(\alpha - 1\right) \hat\mu_i & \forall i\in N^{(s)}.
  \end{align}
  Together with Condition \eqref{cond:multi:chores:value-sum} for $t = s$, we obtain at Line \ref{alg:multi:chores:assign} that
  \begin{align}
    & v_i\left(M^{(s-1)}\right) \\
    & = v_i\left(M^{(s)}\right) - v_i(B) \\
    & \ge \left(s + (|N| - s)\left(2 - \alpha\right)\right) \hat\mu_i - \left(\alpha - 1\right) \hat\mu_i \\
    & \ge \left(s - 1 + (|N| - s + 1)\left(2 - \alpha\right)\right) \hat\mu_i
    & \forall i\in N^{(s)},
  \end{align}
  where Condition \eqref{cond:multi:chores:value-sum} for $t = s-1$ is established.
\end{proof}

\begin{corollary} \label{thm:multi:chores:alpha-MMS-final}
  An $\alpha$-MMS allocation $(A_i)_{i\in N}$ for $\cI$ is returned at Line~\ref{alg:multi:chores:allocation-final}.
\end{corollary}

\begin{proof}
  \cref{thm:multi:chores:invariant-init,thm:multi:chores:invariant} together ensure that Line~\ref{alg:multi:chores:allocation-final} is successfully reached, with Conditions \eqref{cond:multi:chores:disjoint} to \eqref{cond:multi:chores:value-sum} fulfilled for every $t\in\left\{|N|,|N|-1,\ldots,1,0\right\}$.
  Also given Conditions \eqref{cond:multi:chores:disjoint} to \eqref{cond:multi:chores:union} for $t = 0$ and \cref{eq:multi:chores:bag-size} of \cref{thm:multi:chores:invariant}, it is implied that $(A_i)_{i\in N} = \left(A_{i^{(t)}}\right)_{t\in\{|N|,|N|-1,\ldots,1\}}$ is a feasible allocation for $\cI$.
  \cref{thm:multi:chores:normalize} and \cref{eq:multi:chores:bag-value} of \cref{thm:multi:chores:invariant} also yield that
  \begin{align*}
    &v_{i^{(t)}}\!\left(A_{i^{(t)}}\right) \ge \alpha\,\hat\mu_{i^{(t)}} \ge \alpha\,\mu_{i^{(t)}}(\cI) &\forall t\in\left\{|N|,|N|-1,\ldots,1\right\},
  \end{align*}
  establishing that $(A_i)_{i\in N}$ is an $\alpha$-MMS allocation for $\cI$.
\end{proof}

\begin{lemma} \label{thm:multi:chores:runtime}
  \hyperref[alg:multi:chores:main-start]{\textproc{ApproxCategorizedChores}} runs in time $O(|N||M|)$.
\end{lemma}

\begin{proof}
  In each of the $|N|$ iterations of the outer for-loop (Lines \ref{alg:multi:chores:outer-for-start} to \ref{alg:multi:chores:outer-for-end}), the inner for-loop (Lines \ref{alg:multi:chores:inner-for-start} to \ref{alg:multi:chores:inner-for-end}) iterates over all categories in $\cC$, during which the while-loop (Lines \ref{alg:multi:chores:inner-while-start} to \ref{alg:multi:chores:inner-while-end}) iterates at most $\sum_{C\in\cC} |C\cap M^{(t)}| \le |M|$ times in total.
\end{proof}

\section{Special Cases} \label{sec:special}

We discuss specific classes of instances for which we can achieve (almost) exact MMS guarantees.

\subsection{A Constant Number of (Almost) Identical Agents and a Constant Number of Categories}

\citet{woeginger2000does}'s characterization implies the existence of a fully polynomial-time approximation scheme (FPTAS) for the following special case, where agents are said to be \emph{identical} if their valuations are equal as functions.

\begin{theorem} \label{thm:fptas-identical}
  For any real constant $\varepsilon > 0$, there is an algorithm that computes a $(1 - \varepsilon)$-MMS allocation for a single-category instance of goods with a constant number of identical agents and a constant number of categories, or a $(1 +\varepsilon)$-MMS allocation for that of chores, in time polynomial in the number of items $|M|$ and $\varepsilon^{-1}$.
\end{theorem}

\begin{proof}
  We only consider the case of goods, which straightforwardly extends to that of chores.
  Let $n$ and $\ell$ be fixed positive integers.
  Throughout the proof, we use the language of \citet{woeginger2000does}, who shows that every \emph{DP-benevolent} optimization problem admits an FPTAS \citep[Theorem 3.5]{woeginger2000does}.
  We show that the following problem is, more strongly, \emph{ex-benevolent} \citep[Section 5]{woeginger2000does}:

  \begin{quote}
    \emph{Input.}
    A positive integer $m$.
    Non-negative integers $p_1,p_2,\ldots,p_m$.
    A partition of $\{1,2,\ldots,m\}$ into subsets $C_1,C_2,\ldots,C_\ell$.
    A pair of non-negative integers $(q^-_j, q^+_j)$ for each $j\in\{1,2,\ldots,\ell\}$ such that $q^-_j n \le |C_j| \le q^+_j n$.

    \emph{Task.}
    Find an ordered partition $(A_i)_{i=1}^n$ of $\{1,2,\ldots,m\}$ such that $q^-_j \le |A_i \cap C_j| \le q^+_j$ for every $(i, j)\in \{1,2,\ldots,n\}\times\{1,2,\ldots,\ell\}$, and $\sum_{g\in A_i} p_g$ is maximized subject to these constraints.
  \end{quote}

  Let us configure a \emph{dynamic program} \citep[Definition 3.2 to 3.4]{woeginger2000does} for this optimization problem as follows:
  Start with an initial \emph{state space} $\cS_0 \coloneqq \{(0,(0)_{j=1}^\ell)_{i=1}^n\} \subseteq \bZ_{\ge 0}^{n(1 + \ell)}$, and let $j_k$ denote the unique $j\in\{1,2,\ldots,\ell\}$ such that $k \in C_j$.
  Sequentially for $k = 1,2,\ldots,m$, process the pair $(p_k, j_k)$ to produce a new state space $\cS_k \subseteq \bZ_{\ge 0}^{n(1 + \ell)}$ by
  \begin{align}
    \begin{split}
      \cS_k \coloneqq & \left\{F((p_k, j_k), (v_i, (c_{i,j})_{j=1}^\ell)_{i=1}^n) \right. \\
      & \left. \mid F\in\cF, (v_i, (c_{i,j})_{j=1}^\ell)_{i=1}^n\in\cS_{k-1}\right\},
    \end{split}
  \end{align}
  where
  \begin{align}
    \cF \coloneqq \left\{ F_i:\bZ^2 \times \bZ_{\ge 0}^{n(1 + \ell)} \to \bZ_{\ge 0}^{n(1 + \ell)} \mid i\in\{1,2,\ldots,n\}\right\},
  \end{align}
  and for any $i^*\in \{1,2,\ldots,n\}$, $(p, j^*)\in\bZ_{\ge 0}^2$, and $(v_i, (c_{i,j})_{j=1}^\ell)_{i=1}^n \in \bZ_{\ge 0}^{n(1+\ell)}$,
  \begin{align}
    &\mathrlap{F_{i^*}((p, j^*), (v_i, (c_{i,j})_{j=1}^\ell)_{i=1}^n) \coloneqq (v'_i, (c'_{i,j})_{j=1}^\ell)_{i=1}^n} \\
    &\quad \text{s.t.}
    & v'_i & \begin{aligned}
     & \coloneqq
    \begin{cases}
      v_i + p & \text{if } i = i^*, \\
      v_i & \text{otherwise}
    \end{cases} &\forall i\in \{1,2,\ldots,n\},
    \end{aligned} \\
    & & \begin{split}
     c'_{i,j} &\coloneqq
    \begin{cases}
      c_{i,j} + 1 & \text{if } i = i^* \text{ and } j = j^*, \\
      c_{i,j} & \text{otherwise}
    \end{cases} \\
    & \forall i\in \{1,2,\ldots,n\},\forall j\in \{1,2,\ldots,\ell\}.
    \end{split}
  \end{align}
  After all these transitions, compute the optimal objective as $\max \{G(S) \mid S\in\cS_m\}$, where for any $(v_i, (c_{i,j})_{j=1}^\ell)_{i=1}^n \in \bZ_{\ge 0}^{n(1+\ell)}$,
  \begin{align*}
    \begin{split}
      &G((v_i, (c_{i,j})_{j=1}^\ell)_{i=1}^n) \\
      &\coloneqq
      \begin{cases}
        \mathrlap{\min \left\{v_1,v_2,\ldots,v_n\right\}} & \\
        & \text{if }\ q^-_j \le c_{i,j} \le q^+_j\ \ \forall i\in\{1,2,\ldots,n\},\forall j\in\{1,2,\ldots,\ell\}, \\[2pt]
        0 & \text{otherwise.}
      \end{cases}
    \end{split} \label{eq:fptas-identical:G}
  \end{align*}

  Now that this dynamic program solves the original problem correctly, it suffices to confirm that Conditions C.3(i), C.4, and C.5 \citep{woeginger2000does} are satisfied for the following \emph{degree-vector}:
  \begin{align}
    D \coloneqq ((1, (0)_{j=1}^\ell))_{i=1}^n.
  \end{align}
  By the above setup, checking Condition C.4 is immediate.
  Fix any real constant $\Delta > 1$, and let $S = (v_i, (c_{i,j})_{j=1}^\ell)_{i=1}^n$ and $S' = (v'_i, (c'_{i,j})_{j=1}^\ell)_{i=1}^n$ be any two vectors in $\bZ_{\ge 0}^{n(1+\ell)}$ that are \emph{$[D,\Delta]$-close} \citep{woeginger2000does} to each other.
  By definition, it holds that
  \begin{align}
    &\Delta^{-1} v_i \le v'_i \le \Delta\,v_i &\forall i\in\{1,2,\ldots,n\}, \label{eq:fptas-identical:close-v} \\
    &c'_{i,j} = c_{i,j} &\forall i\in\{1,2,\ldots,n\},\forall j\in\{1,2,\ldots,\ell\}. \label{eq:fptas-identical:close-c}
  \end{align}
  Given \cref{eq:fptas-identical:G,eq:fptas-identical:close-v}, we have that
  \begin{align*}
    &G(S') \\
    &=
    \begin{cases}
        \mathrlap{\min \left\{v'_1,v'_2,\ldots,v'_n\right\}} & \\
        & \text{if }\ q^-_j \le c'_{i,j} \le q^+_j\ \ \forall i\in\{1,2,\ldots,n\},\forall j\in\{1,2,\ldots,\ell\}, \\[2pt]
        0 & \text{otherwise.}
      \end{cases} \\
    &\le
    \begin{cases}
      \mathrlap{\Delta \min \left\{v_1,v_2,\ldots,v_n\right\}} & \\
      & \text{if }\ q^-_j \le c_{i,j} \le q^+_j\ \ \forall i\in\{1,2,\ldots,n\},\forall j\in\{1,2,\ldots,\ell\}, \\
      0 & \text{otherwise}
    \end{cases} \\
    &= \Delta\,G(S)
  \end{align*}
  which implies that Condition C.3(i) is met for $g = 1$.
  Fix also any $i^*\in\{1,2,\ldots,n\}$ and $(p,j^*)\in\bZ_{\ge 0}^2$.
  Given \cref{eq:fptas-identical:close-c} and that $\Delta^{-1}(v_{i^*} + p) \le v'_{i^*} + p \le \Delta(v_{i^*} + p)$ due to \cref{eq:fptas-identical:close-v}, $F_{i^*}((p,j^*),S)$ and $F_{i^*}((p,j^*),S')$ are $[D,\Delta]$-close to each other.
  Therefore, Condition C.5 holds.
\end{proof}

\begin{remark}
  In contrast, a major difficulty in obtaining a PTAS for instances with a variable number of identical agents arises from the restricted feasibility by quotas.%
  \footnote{As this problem is strongly NP-hard, an FTPAS does not exist unless $P=NP$.}
  When the number of agents is no longer constant, the PTAS known in the unconstrained cases of goods \citep{woeginger1997polynomial} and chores \citep{hochbaum1987using}, or in other related problems, involve the separate treatment of small- and large-valued items, which appears hard to adapt in the presence of either lower or upper quotas, let alone when both apply.
  More specifically, effective approaches in the unconstrained scenarios are known to first partition the large-valued items optimally and then fill in the small-valued ones so that any bundle does not run out its $\epsilon$-slack.
  However, it is not immediate how to efficiently incorporate quota constraints throughout both these phases.
\end{remark}

Combining \cref{thm:fptas-identical} with the cut-and-choose protocol further gives \cref{thm:fptas-almost}.
Here, agents are said to be \emph{almost identical} if all but at most one of their valuations are equal as functions.
In particular, this applies when there are only $|N| = 2$ agents.

\begin{corollary} \label{thm:fptas-almost}
  For any real constant $\varepsilon > 0$, there is an algorithm that computes a $(1 - \varepsilon)$-MMS allocation for a single-category instance of goods with a constant number of almost identical agents and a constant number of categories, or a $(1 +\varepsilon)$-MMS allocation for that of chores, in time polynomial in the number of items $|M|$ and $\varepsilon^{-1}$.
\end{corollary}

\begin{proof}
  Let $\cI = \left(N,M,(v_i)_{i\in N},\cC,(q_{C}^-,q_{C}^+)_{C\in\cC}\right)$ be an arbitrary instance where agents except $i^\star\in N$ have equal valuations $v_*$.
  Consider another instance $\cI' = \left(N,M,(v'_i)_{i\in N},\cC,(q_{C}^-,q_{C}^+)_{C\in\cC}\right)$ with $v'_i = v_*$ for every $i\in N$, and suppose that $(A'_i)_{i\in N}$ is an $\alpha$-MMS allocation for $\cI'$ and some $\alpha\in\bR$.
  Then $(A_i)_{i\in N}$ such that
  \begin{align*}
    A_{i^\star} &= A'_{i^{\max}}, \\
    A_{i^{\max}} &= A'_{i^\star}, \\
    A_i &= A'_i &\forall i\in N\setminus\{i^\star,i^{\max}\},
  \end{align*}
  is an $\alpha$-MMS allocation for $\cI$, where $i^{\max} \in \arg\max_{i\in N} v_*(A'_i)$.
  This observation, combined with \cref{thm:fptas-identical}, establishes \cref{thm:fptas-almost}.
\end{proof}

\subsection{Single-Category Bivalued Instances} \label{sec:single:bivalued}

An instance $\left(N,M,(v_i)_{i\in N},\cC,(q_{C}^-,q_{C}^+)_{C\in\cC}\right)$ is said to be \textit{bivalued} when the following holds for some constants $(a,b)\in\bR^2$:
\begin{align}
  v_i(g) &\in \{a, b\}
  &\forall g\in M,\forall i\in N.
\end{align}
In words, the value of each item for each agent is restricted to one of two common candidates.
\citet{feige2022maximin} shows that any single-category, unconstrained (i.e., $q^- = 0$ and $q^+ = |M|$), bivalued instance admits an (exact) MMS allocation, which can be found in polynomial time.
We note that their algorithm and correctness proof can be extended to our problem setting by reinterpreting the unconstrained MMS values, partitions, and allocations as their quota-constrained counterparts.

\begin{theorem}[{\citep[cf.][Theorem 1]{feige2022maximin}}]\label{thm:bivalued}
  For an arbitrary single-category bivalued instance, an MMS allocation exists and can be found in polynomial time.
\end{theorem}

We leverage ideas from the proof by \citet{feige2022maximin} for the special case of $(q^-,q^+) = (0,|M|)$.
The underlying strategy is also similar to that of the proof of \cref{thm:common:valid-reduction}.

\begin{proof}
  We prove \cref{thm:bivalued} by induction on the number of agents.
  Fix an arbitrary single-category bivalued instance $\cI = (N, M, (v_i)_{i\in N}, (q^-,q^+))$, which satisfies
  \begin{align}
    a&\ge b, \label{eq:bivalued:a>b} \\
    v_i(g)&\in\{a,b\} &\forall g\in M,i\in N, \label{eq:bivalued}
  \end{align}
  for some $(a,b)\in\bR^2$.
  Let also
  \begin{align}
    \ell_i &\coloneqq \left\{g\in M\mid v_i(g) = a\right\} &\forall i\in N.
  \end{align}
  Due to \cref{thm:common:ordered}, it suffices to consider the case where $\cI$ is ordered.

  Suppose that $N\neq\emptyset$.
  Fix an arbitrary
  \begin{align}
    i^*&\in\arg\max_{i\in N} \ell_i. \label{eq:bivalued:def-i*}
  \end{align}
  Let $(P^*_k)_{k=1}^{|N|}$ be agent $i^*$'s MMS partition, which we assume, without loss of generality, satisfies
  \begin{align}
    \left|P^*_1\right|
    \begin{cases}
      \le \frac{|M|}{|N|} & \text{ if }\, b\ge 0; \\[2pt]
      \ge \frac{|M|}{|N|} & \text{ otherwise.}
    \end{cases}
  \end{align}
  Notice that \cref{eq:bivalued} allows us to compute $(P^*_k)_{k=1}^{|N|}$ in polynomial time.\footnote{For example, dynamic programming can be efficiently used, as each subproblem is characterized solely by the number of agents, the number of items of value $a$, and the number of items of value $b$.}
  Now, we define
  \begin{align}
    B &\coloneqq \left\{g_j \mid \ell_{i^*} - p_a < j\le \ell_{i^*} \right\} \cup \left\{g_j \mid |M| - p_b < j\le |M| \right\}, \\
    \cI' &\coloneqq \left(N\setminus\{i^*\},M\setminus B,\left(v_i|_{2^{M\setminus B}}\right)_{i\in N\setminus\{i^*\}},\left(q^-,q^+\right)\right), \label{eq:bivalued:reduced-instance}
  \end{align}
  where
  \begin{align}
    p_a &\coloneqq \left\{g\in P^*_1\mid v_i(g) = a\right\}, \\
    p_b &\coloneqq \left\{g\in P^*_1\mid v_i(g) = b\right\}.
  \end{align}
  Since $|B| = |P^*_1|$ holds by definition, $\cI'$ is well-defined by \cref{eq:bivalued:reduced-instance} as a single-category bivalued instance.
  Due to the induction hypothesis, it suffices to show
  \begin{align}
    \mu_i(\cI')&\ge\mu_i(\cI) &\forall i\in N\setminus\{i^*\}, \label{eq:bivalued:valid}
  \end{align}

  Fix an arbitrary $i\in N\setminus\{i^*\}$, for which we show \cref{eq:bivalued:valid} in the following.
  Let us also define
  \begin{align}
    \nu(M') \coloneqq \mu_i\left(N\setminus\{i^*\}, M', \left(v_{i'}|_{2^{M'}}\right)_{i'\in N\setminus \{i^*\}}, \left(q^-,q^+\right)\right) \label{eq:bivalued:mu-reduced}
  \end{align}
  for any $M'\subseteq M$ that ensures the instance in the right-hand side is well-defined.

  \begin{case} \label{case:bivalued:1}
    Suppose that $\ell_i > \ell_{i^*} - p_a$.
    Then it follows that
    \begin{align}
      v_i(g) &= v_{i^*}(g) &\forall g\in M\setminus B,
    \end{align}
    which, together with the definition of $B$, gives
    \begin{align}
      \begin{split}
        \mu_i(\cI') &= \nu(M\setminus B) \\
        &\ge \nu(M\setminus P^*_1) \\
        &\ge \min_{2\le k\le |N|} v_i\!\left(P^*_k\right) \\
        &= \min_{2\le k\le |N|} v_{i^*}\!\left(P^*_k\right) \\
        &\ge \mu_{i^*}(\cI)
      \end{split}
    \end{align}
    Since \crefrange{eq:bivalued:a>b}{eq:bivalued:def-i*} also imply $\mu_{i^*}(\cI) \ge \mu_{i}(\cI)$, we obtain \cref{eq:bivalued:valid} for the arbitrarily fixed $i$.
  \end{case}

  \begin{case} \label{case:bivalued:2}
    Suppose that $\ell_i \le \ell_{i^*} - p_a$ and $b\ge 0$.
    Then we have
    \begin{align}
      v_i(g) &\ge 0 &\forall g\in M, \label{eq:bivalued:positive-1} \\
      v_i(g) &= b &\forall g\in B. \label{eq:bivalued:all-b-1}
    \end{align}
    Let $(P_k)_{k=1}^{|N|}$ be agent $i$'s MMS partition, which we assume satisfies
    \begin{align}
      |P_1|\ge \frac{|M|}{|N|}\ge \left|P_1^*\right| = |B| \label{eq:bivalued:size-1}
    \end{align}
    without loss of generality.
    Given \cref{eq:bivalued:all-b-1,eq:bivalued:size-1}, there exists an injection $f:B\setminus P_1\to P_1\setminus B$ such that
    \begin{align}
      v_i(g) &\le v_i(f(g)) &\forall g\in B\setminus P_1. \label{eq:bivalued:injection-1}
    \end{align}
    From \cref{eq:bivalued:positive-1,eq:bivalued:injection-1}, as well as the definition of $(P_k)_{k=1}^{|N|}$ and $B$, we obtain \cref{eq:bivalued:valid} for the arbitrarily fixed $i$ as follows:
    \begin{align}
      \begin{split}
        \mu_{i}(\cI') &= \nu(M\setminus B) \\
        &\ge \nu(M\setminus ((B \cap P_1) \cup f(B\setminus P_1))) \\
        &\ge \nu(M\setminus P_1) \\
        &\ge \min_{2\le k \le |N|} v_{i}(P_{k}) \\
        &\ge \mu_{i}(\cI).
      \end{split}
    \end{align}
  \end{case}

  \begin{case} \label{case:bivalued:3}
    Suppose that $\ell_i \le \ell_{i^*} - p_a$ and $b < 0$.
    Then we have
    \begin{align}
      v_i(g) &= b < 0 &\forall g\in B. \label{eq:bivalued:all-b-2}
    \end{align}
    Let $(P_k)_{k=1}^{|N|}$ be agent $i$'s MMS partition, which we assume satisfies
    \begin{align}
      |P_1| \le \frac{|M|}{|N|}\le \left|P_1^*\right| = |B| \label{eq:bivalued:size-2}
    \end{align}
    without loss of generality.
    Given \cref{eq:bivalued:all-b-2,eq:bivalued:size-2}, there exists an injection $f:P_1\setminus B\to B\setminus P_1$ such that
    \begin{align}
      v_i(g) &\ge v_i(f(g)) &\forall g\in P_1\setminus B. \label{eq:bivalued:injection-2}
    \end{align}
    From \cref{eq:bivalued:all-b-2,eq:bivalued:injection-2}, as well as the definition of $(P_k)_{k=1}^{|N|}$ and $B$, we obtain \cref{eq:bivalued:valid} for the arbitrarily fixed $i$ as follows:
    \begin{align}
      \begin{split}
        \mu_{i}(\cI') &= \nu(M\setminus B) \\
        &\ge \nu(M\setminus ((B \cap P_1) \cup f(P_1\setminus B))) \\
        &\ge \nu(M\setminus P_1) \\
        &\ge \min_{2\le k \le |N|} v_{i}(P_{k}) \\
        &\ge \mu_{i}(\cI).
      \end{split}
    \end{align}
  \end{case}

  As \crefrange{case:bivalued:1}{case:bivalued:3} are exhaustive, we establish \cref{eq:bivalued:valid} and complete the proof of \cref{thm:bivalued}.
\end{proof}

\section{Inapproximability} \label{sec:inapprox}

Here, we consider one of the most relevant questions to this study: what is the best approximation ratio that can be guaranteed for any instance?
\cref{thm:inapprox:milp} asserts that we can compute in finite time the exact approximation bound for all instances with a fixed dimension, i.e., the number of agents and items, the sizes of categories, and the pair of quotas for each category.
Our reduction to mixed binary linear programming (MBLP) is motivated by \citet{feige2022improved}, who employ a similar formulation in the unconstrained setting.

\begin{theorem} \label{thm:inapprox:milp}
  Let $N$ and $M$ be arbitrary finite sets with $\mathcal{C}$ being an arbitrary partition of $M$, for which $(q^-_C,q^+_C)_{C\in\mathcal{C}}$ are arbitrary pairs of non-negative integers that satisfies \cref{eq:def:quotas}.
  The following value exists and can be found by solving a mixed binary linear program (MBLP) of finite size with integer coefficients: the maximum (resp.~minimum) real constant $\alpha$ such that any instance $\cI = \left(N,M,(v_i)_{i\in N},\cC,(q_{C}^-,q_{C}^+)_{C\in\cC}\right)$ of goods (resp.~chores) admits an $\alpha$-MMS allocation.
\end{theorem}

\begin{proof}
  We only discuss the case of goods, which easily extends to that of chores.
  Let $\alpha^*$ denote the supremum of real constants $\alpha$ such that any instance $\cI = \left(N,M,(v_i)_{i\in N},\cC,(q_{C}^-,q_{C}^+)_{C\in\cC}\right)$ of goods admits an $\alpha$-MMS allocation.
  Common to all such instances,%
  \footnote{Their only degree of freedom lies in the choice of valuations $(v_i)_{i\in N}$.}
  we let $\cS \coloneqq \{S \subseteq M \mid q^-_C \le |S\cap C| \le q^+_C\ \ \forall C\in\cC \}$ denote the set of all feasible bundles, and let $\cF\subseteq \cS^N$ denote the set of all partitions of $M$ into sets in $\cS$ labelled by $N$, which is guaranteed to be non-empty by the assumption of \cref{thm:inapprox:milp}.
  The proof revolves around the following MBLP:
  \begin{align}
    & \text{Minimize} & & \alpha; \label{eq:inapprox:mblp:obj} \\
    & \text{Subject to} & & \alpha \in [0, \infty), \label{eq:inapprox:mblp:var:alpha} \\
    & & & u_{i,g} \in [0, 1] & & \forall i\in N, \forall g\in M, \label{eq:inapprox:mblp:var:util} \\
    & & & p_{i,k,S} \in \{0, 1\} & & \forall i,k\in N, \forall S\in\cS, \label{eq:inapprox:mblp:var:mms-partition} \\
    & & & \sum_{S\in\cS} p_{i,k,S} = 1 & & \forall i,k\in N, \label{eq:inapprox:mblp:cons:mms-partition-agent-bundle} \\
    & & & \sum_{k\in N}\sum_{g\in S\in\cS} p_{i,k,S} = 1 & & \forall i\in N,\forall g\in M, \label{eq:inapprox:mblp:cons:mms-partition-disjoint} \\
    & & & \sum_{g\in S} u_{i,g} \ge p_{i,k,S} & & \forall i,k\in N,\forall S\in\cS, \label{eq:inapprox:mblp:cons:mms-at-least-one} \\
    & & & \ell_{i, S} \in \{0, 1\} & & \forall i\in N, \forall S\in\cS, \label{eq:inapprox:mblp:var:at-most-alpha} \\
    & & & \sum_{g\in S} u_{i,g} \le \alpha + |S|(1 - \ell_{i,S}) & & \forall i\in N,\forall S\in\cS, \label{eq:inapprox:mblp:cons:at-most-alpha} \\
    & & & \sum_{i\in N} \ell_{i,A_i} \ge 1 & & \forall A\in\cF. \label{eq:inapprox:mblp:cons:someone-at-most-alpha}
  \end{align}
  In what follows, we show that this MBLP is feasible, that its optimal value $\alpha^\star$ is equal to $\alpha^*$, and that the supremum $\alpha^*$ can actually be attained.
  Note that the above MBLP could never be unbounded due to \cref{eq:inapprox:mblp:obj,eq:inapprox:mblp:var:alpha}.

  First, let us fix an arbitrary real constant $\alpha\in\left(\alpha^*,\infty\right]$.
  By definition, there exists an instance $\cI = \left(N,M,(v_i)_{i\in N},\cC,(q_{C}^-,q_{C}^+)_{C\in\cC}\right)$ that satisfies
  \begin{align}
    v_i(A_i) < \alpha\,\mu_i(\cI) & &\exists i\in N,\forall A\in\cF. \label{eq:inapprox:less-than-alpha-mms}
  \end{align}
  We consider setting the variables of the above MBLP except $\alpha$ as
  \begin{align}
    u_{i, g} & \coloneqq v_i(g) & & \forall i\in N, \forall g\in M, \\
    p_{i, k, S} & \coloneqq
    \begin{cases}
      1 & \text{ if }\, S = P^{(i)}_k \\
      0 & \text{ otherwise}
    \end{cases} & & \forall i,k\in N, \forall S\in \cS, \\
    \ell_{i,S} & \coloneqq
    \begin{cases}
      1 & \text{ if }\,v_i(S) < \alpha \\
      0 & \text{ otherwise}
    \end{cases} & & \forall i\in N, \forall S\in\cS,
  \end{align}
  where $(P^{(i)}_k)_{k\in N}$ denotes agent $i$'s MMS partition in $\cI$.
  \cref{eq:inapprox:less-than-alpha-mms} ensures that these, along with the arbitrarily fixed $\alpha$, form a feasible solution to the above MBLP.
  In particular, $\alpha^\star \le \alpha^*$ is implied.

  Next, we let $\left(\alpha, (u_{i,g})_{i\in N,g\in M}, (p_{i,k,S})_{i,k\in N,S\in\cS}, (\ell_{i,S})_{i\in N,S\in\cS}\right) = \left(\alpha^\star, (u^\star_{i,g})_{i\in N,g\in M}, (p^\star_{i,k,S})_{i,k\in N,S\in\cS}, (\ell^\star_{i,S})_{i\in N,S\in\cS}\right)$ denote any optimal solution to the above MBLP.
  Then we define an instance $\cI = \left(N,M,(v_i)_{i\in N},\cC,(q_{C}^-,q_{C}^+)_{C\in\cC}\right)$ by letting
  \begin{align}
    v_i(g) & \coloneqq u^\star_{i, g} & & \forall i\in N, \forall g\in M. \label{eq:inapprox:values}
  \end{align}
  It follows from \cref{eq:inapprox:mblp:var:util,eq:inapprox:mblp:var:mms-partition,eq:inapprox:mblp:cons:mms-partition-agent-bundle,eq:inapprox:mblp:cons:mms-partition-disjoint,eq:inapprox:mblp:cons:mms-at-least-one,eq:inapprox:values} that $\cI$ allows every agent to find a feasible partition where they value any bundle at no less than $1$, i.e.,
  \begin{align}
    \mu_i(\cI) &\ge 1 & & \forall i\in N. \label{eq:inapprox:mms-at-least-one}
  \end{align}
  We also observe from \cref{eq:inapprox:mblp:var:util,eq:inapprox:mblp:var:at-most-alpha,eq:inapprox:mblp:cons:someone-at-most-alpha,eq:inapprox:mblp:cons:at-most-alpha,eq:inapprox:values} that any feasible allocation for $\cI$ results in some agent gaining no more than $\alpha^\star$, i.e.,
  \begin{align}
    \min_{i\in N}\, v_i(A_i) &\le \alpha^\star & & \forall A\in \cF. \label{eq:inapprox:someone-at-most-alpha}
  \end{align}
  \cref{eq:inapprox:mms-at-least-one,eq:inapprox:someone-at-most-alpha} together necessitate $\alpha^* \le \alpha^\star$ by the definition of $\alpha^*$.

  Now that $\alpha^\star = \alpha^*$, it suffices to show that any instance $\cI = \left(N,M,(v_i)_{i\in N},\cC,(q_{C}^-,q_{C}^+)_{C\in\cC}\right)$ of goods admits an $\alpha^\star$-MMS allocation.
  Let us define
  \begin{align}
    u^\cI_{i, g} &\coloneqq
    \begin{dcases}
      \frac{v_i(g)}{v_i(P^{(i)}_{k})} & \text{ if }\,v_i(P^{(i)}_{k}) > 0 \\
      1 & \text{ otherwise}
    \end{dcases} & & \forall g\in P^{(i)}_k, \forall i,k\in N,\label{eq:inapprox:def-util} \\
    p^\cI_{i, k, S} & \coloneqq
    \begin{cases}
      1 & \text{ if }\, S = P^{(i)}_k \\
      0 & \text{ otherwise}
    \end{cases} & & \forall i,k\in N, \forall S\in \cS, \label{eq:inapprox:def-partition}
  \end{align}
  where $(P^{(i)}_k)_{k\in N}$ denotes agent $i$'s MMS partition in $\cI$.
  We then construct a new MBLP by combining \crefrange{eq:inapprox:mblp:obj}{eq:inapprox:mblp:cons:someone-at-most-alpha} with additional constraints that fix certain variables as follows:
  \begin{align}
    u_{i, g} &= u^{\cI}_{i, g} & & \forall i\in N,\forall g\in M, \label{eq:inapprox:new-mblp:cons:fix-util} \\
    p_{i, k, S} &= p^{\cI}_{i, k, S}  & & \forall i,k\in N, \forall S\in \cS. \label{eq:inapprox:new-mblp:cons:fix-partition}
  \end{align}
  Notice that the new MBLP has a feasible solution where the remaining variables are set to
  \begin{align}
    \alpha &\coloneqq |M|, \\
    \ell_{i,S} &\coloneqq 1 & & \forall i\in N, \forall S\in\cS.
  \end{align}
  Let $X^\cI = \left(\alpha^\cI, (u^\cI_{i,g})_{i\in N,g\in M}, (p^\cI_{i,k,S})_{i,k\in N,S\in\cS}, (\ell^\cI_{i,S})_{i\in N,S\in\cS}\right)$ denote the optimal solution to the new MBLP.
  We must have that
  \begin{align}
    \alpha^{\cI} &\ge \alpha^\star \label{eq:inapprox:alpha-no-less}
  \end{align}
  by the definition of the new MBLP, as well as that
  \begin{align}
    \sum_{g\in A_i} u^\cI_{i,g} &\ge \alpha^\cI & & \forall i\in N,\exists A\in\cF. \label{eq:inapprox:equality}
  \end{align}
  Indeed, violating \cref{eq:inapprox:equality} would imply that there exists $i^{(A)}\in N$ for every $A\in\cF$ such that
  \begin{align}
    \sum_{g\in A_{i^{(A)}}} u^\cI_{i^{(A)},g} &< \alpha^\cI,
  \end{align}
  and thus the new MBLP admits a feasible solution defined by \crefrange{eq:inapprox:def-util}{eq:inapprox:new-mblp:cons:fix-partition} and
  \begin{align}
    \alpha &\coloneqq \max_{A\in\cF} \sum_{g\in A_{i^{(A)}}} u^\cI_{i^{(A)},g} < \alpha^\cI, \\
    \ell_{i,S} &\coloneqq
    \begin{cases}
      1 & \text{ if }\,v_i(S) < \alpha \\
      0 & \text{ otherwise}
    \end{cases}  & & \forall i\in N,\forall S\in\cS,
  \end{align}
  which contradicts the optimality of the solution $X^\cI$.
  Finally, recalling that
  \begin{align}
    \mu_i(\cI) = \min_{k\in N}\, v_i(P^{(i)}_k) & & \forall i\in N, \label{eq:inapprox:MMS-def}
  \end{align}
  \cref{eq:inapprox:alpha-no-less,eq:inapprox:def-util,eq:inapprox:equality,eq:inapprox:MMS-def} together yield that
  \begin{align}
    v_i(A_i) & = \sum_{g\in A_i} v_i(g) \\
    & \ge \sum_{g\in A_i} u^\cI_{i,g}\,\mu_i(\cI) \\
    & \ge \alpha^\cI \mu_i(\cI) \\
    & \ge \alpha^\star \mu_i(\cI)
    & \forall i\in N,\exists A\in\cF,
  \end{align}
  which guarantees the existence of an $\alpha^\star$-MMS allocation for $\cI$ and thus concludes the proof.
\end{proof}

\begin{corollary} \label{thm:inapprox:rational}
  The maximum (resp.~minimum) real constant guaranteed by \cref{thm:inapprox:milp} is rational.
\end{corollary}

\begin{proof}
  Every linear program induced by fixing all binary variables arbitrarily has integer coefficients only, thus admitting a rational optimal value \citep{korte2018combinatorial}.
\end{proof}

\begin{remark}
  By allowing for non-binary integer variables and a larger program size, one can also configure a mixed integer linear program (MILP) parametrized only by $|N|$ and $|M|$, where the category partition $\cC$ can be arbitrary.
  The above arguments can be applied for any type of constraints on feasible bundles, beyond the quota constraints in our particular case.
\end{remark}

The best known results for the unconstrained settings \citep{feige2021tight} establish that when $\mathcal{C} = \{M\}$ and $(q^-_M, q^+_M) = (0, |M|)$, the maximum constant for the case of goods is at most $\frac{39}{40}$, while the minimum constant for the case of chores is at least $\frac{44}{43}$.
Although we have let general solvers to solve the MBLP defined by \crefrange{eq:inapprox:mblp:obj}{eq:inapprox:someone-at-most-alpha},
no better bound has been able to found due to the exponentially large number of variables and constraints.\footnote{The MBLP consists of $O(|N|^2 2^{|M|})$ variables and $O(|N|^{|M|})$ constraints.}
It is thus left open whether category partitions and associated quotas would make differences in the approximation bounds from their unconstrained counterparts.

\end{document}